\documentclass[preprint, 12pt]{elsarticle}
\pdfoutput=1

\usepackage[margin = 2.5cm]{geometry}
\usepackage[utf8]{inputenc}
\usepackage[english]{babel}
\usepackage{hyperref}

\usepackage{graphicx,color}
\usepackage{caption} 
\usepackage{subcaption}

\usepackage{tikz}
\usetikzlibrary {positioning}
\usetikzlibrary{bayesnet}

\usepackage[ruled,linesnumbered]{algorithm2e}
\usepackage{algorithmic}
\usepackage{amsmath,latexsym,amssymb}
\usepackage{amsthm}
\usepackage{color}
\usepackage{soul} % for strike through
\usepackage{blkarray}
\usepackage{enumitem}
\usepackage{epstopdf}
\usepackage{cleveref}

\SetKwInput{KwData}{Input}
\SetKwInput{KwResult}{Output}
\newcommand\numberthis{\addtocounter{equation}{1}\tag{\theequation}}

\usepackage{amsfonts,bbm}

\newcommand{\Td}{\mathrm{d}}

\newcommand{\TI}{\mathrm{I}}

\newcommand{\TP}{\mathrm{P}}

\newcommand{\Bv}{\mathbf{v}}
\newcommand{\Bw}{\mathbf{w}}
\newcommand{\Bx}{\mathbf{x}}

\newcommand{\Bz}{\mathbf{z}}

\newcommand{\BI}{\mathbf{I}}

\newcommand{\BX}{\mathbf{X}}
\newcommand{\BY}{\mathbf{Y}}
\newcommand{\BZ}{\mathbf{Z}}

\newcommand{\BIg}{{\boldsymbol{g}}}

\newcommand{\CC}{\mathcal{C}}

\newcommand{\CF}{\mathcal{F}}

\newcommand{\CH}{\mathcal{H}}

\newcommand{\CN}{\mathcal{N}}

\newcommand{\CR}{\mathcal{R}}

\newcommand{\CT}{\mathcal{T}}

\newcommand{\CX}{\mathcal{X}}

\newcommand{\BSigma}{{\boldsymbol{\Sigma}}}

\newcommand{\Bepsilon}{{\boldsymbol{\epsilon}}}

\newcommand{\Bmu}{{\boldsymbol{\mu}}}

\newcommand{\NORMAL}{\CN} % Normal distribution
\newcommand{\GP}{\mathcal{GP}} % Gaussian process
\newcommand{\R}{\mathbb{R}} % real number
\newcommand{\N}{\mathbb{N}} % natural number
\newcommand{\veczero}{\mathbf{0}} % vector of zeros
\newcommand{\INDI}{\mathbbm{1}} % indicator function
\newcommand{\TRANSP}{\mathsf{T}} % transpose
\newcommand{\DIFFX}[1]{\,\Td{#1}} % differentiation with a measure or a variable name
\newcommand{\PROB}{\mathbb{P}} % probability measure
\newcommand{\EXP}{\mathbb{E}} % expectation
\newcommand{\COV}{\mathrm{Cov}} % covariance
\newcommand{\rom}[1]{%
  \textup{\uppercase\expandafter{\romannumeral#1}}%
}

\DeclareMathOperator*{\argmin}{arg\,min} % argument that minimizes
\DeclareMathOperator*{\argmax}{arg\,max} % argument that maximizes
\newtheorem{theorem}{\noindent \textbf{Theorem}}
\newtheorem{lemma}[theorem]{\noindent \textbf{Lemma}}
\newtheorem{definition}[theorem]{\noindent \textbf{Definition}}
\newtheorem{proposition}[theorem]{\noindent \textbf{Proposition}}

\newtheorem{corollary}[theorem]{Corollary}
\newtheorem{remark}[theorem]{Remark}
\newtheorem{assumption}[theorem]{Assumption}

\begin{document}
\setstcolor{red}

\begin{frontmatter}
\title{Binary Spatial Random Field Reconstruction from Non-Gaussian Inhomogeneous Time-series Observations}
\author[1]{Shunan Sheng}
\ead{ss6574@columbia.edu}
\author[2]{Qikun Xiang}
\ead{qikun001@e.ntu.edu.sg}
\author[3]{Ido Nevat}
\ead{ido.nevat@tum-create.edu.sg}
\author[2]{Ariel Neufeld}
\ead{ariel.neufeld@ntu.edu.sg}
\affiliation[1]{organization={Department of Statistics, Columbia University},%Department and Organization
            addressline={ 1255 Amsterdam Avenue}, 
            city={New York},
            postcode={ NY 10027}, 
            country={USA}}
 
\affiliation[2]{organization={Division of Mathematical Sciences, Nanyang Technological University},%Department and Organization
            addressline={  21 Nanyang Link}, 
            city={Singapore},
            postcode={637371}, 
            country={Singapore}}
\affiliation[3]{organization={TUMCREATE},%Department and Organization
            addressline={1 Create Way, \#10-02 CREATE Tower}, 
            city={Singapore},
            postcode={138602}, 
            country={Singapore}}

\begin{abstract}
We develop a new model for \textit{spatial random field reconstruction} of a binary-valued spatial phenomenon. In our model, sensors are deployed in a wireless sensor network across a large geographical region. Each sensor measures a non-Gaussian inhomogeneous temporal process which depends on the spatial phenomenon. Two types of sensors are employed: one collects point observations at specific time points, while the other collects integral observations over time intervals. Subsequently, the sensors transmit these time-series observations to a Fusion Center (FC), and the FC infers the spatial phenomenon from these observations. We show that the resulting posterior predictive distribution is intractable and develop a tractable two-step procedure to perform inference. Firstly, we develop algorithms to perform approximate Likelihood Ratio Tests on the time-series observations, compressing them to a single bit for both point sensors and integral sensors. Secondly, once the compressed observations are transmitted to the FC, we utilize a Spatial Best Linear Unbiased Estimator (S-BLUE) to reconstruct the binary spatial random field at any desired spatial location. 
The performance of the proposed approach is studied using simulation. We further illustrate the effectiveness of our method using a weather dataset from the National Environment Agency (NEA) of Singapore with fields including temperature and relative humidity.
\end{abstract}
\begin{keyword}
Binary spatial random field reconstruction, Sensor Networks, Warped Gaussian Process, Likelihood Ratio Test (LRT), Spatial Best Linear Unbiased Estimator (S-BLUE).
\end{keyword}

\end{frontmatter}

\section{Introduction}
Wireless sensor networks (WSNs) have captivated substantial attention due to its wide applications in environmental monitoring \cite{HART2006}, weather forecasts \cite{Rajasegarar2014,Kottas2012,French2013}, surveillance \cite{Sohraby2006}, building monitoring \cite{Chintalapudi2006}, and automation \cite{Akyildiz2002}. Recent studies focus on the estimation of a single point source, like source localization \cite{Luo2005,Xia2008,Chiu2011,Masazade2010} and source detection \cite{Cohen2011,Msechu2012,Zhang2014,Ido2014,Luo2015}, which generally assume independence of the observations. 
In this paper, we consider a WSN consisting of spatially distributed sensors with limited energy and communication bandwidth. The sensors monitor non-Gaussian temporal processes with desired features such as precipitation, humidity, temperature, concentration of substance, etc., that are dependent on a binary spatial random field \cite{Gelfand2019}. After the sensors transmit these observations to a Fusion Center (FC) \cite{Fazel2012}, the FC then reconstructs the binary spatial random field at spatial locations where no sensor is placed, based on which further decisions can be made.

Binary spatial random fields are commonly used to model ecological phenomena that take binary values, such as defoliation \cite{Heagerty1998} and pest outbreak \cite{Zhu2005}.  Some spatial phenomena such as volcanic activity \cite{FALKNER2007} and vision perception \cite{LEE1999} may be hard to observe directly. Therefore, in our model, we analyze temporal observations generated based on the values of the binary spatial random field and seek to reconstruct the field using the temporal observations as proxy. The proposed dependence structure has its application in vision research \cite{LEE1999} where two types of neuron are activated with respect to low-pass and high-pass components in graphs. It can also be adopted in environmental monitoring problems to model spatial-temporal phenomena where the spatial phenomenon is not directly observable and only the temporal processes are observed.

In many cases, these spatial and temporal processes are modeled using Gaussian Processes (GPs) \cite{Ido2015,Zhang2018,Xiang2019}. Though GPs lead to a concise and elegant probabilistic framework that allows further nonparametric regression and classification as shown in \cite{C.E.Rasmussen2006}, their practicality is much restricted due to their normality and exponentially decaying tails, which are inappropriate for modeling categorical \cite{Gelfand2019} as well as long-tailed observations \cite{Embrechts1999, xiang2021}. 
To ameliorate these limitations, non-linear distortions of a GP, called Warped Gaussian Processes (WGPs) \cite{Miguel2012, Rios2019, Snelson2003, Vinokur2021}, are often adopted as an alternative. The warping function in a WGP for distorting a GP can take any parametric form, like the sum of  $\mathrm{tanh}$ functions \cite{Snelson2003} or Tukey's family of transformations used in \cite{Gareth2021} for environmental monitoring. Moreover, when the warping function is strictly increasing and continuous, the marginal likelihood of the WGP is analytically tractable \cite[Section~3]{Snelson2003}, which makes the WGP superior to other non-parametric models. 

The main purpose of this paper is to develop a low complexity algorithm that can reconstruct a binary spatial random field given transmitted time-series data from sensors. The binary spatial random field is modeled by a WGP with the warping function being an indicator function.  Conditional on the values of the binary spatial random field, the temporal processes are modeled by WGPs following specific non-Gaussian marginal distributions. Meanwhile, two types of sensors are deployed to observe either point or integral observations. The Point sensors defined later in \ref{Temporal: PointObservation}  take point observations that are noisy realizations of the temporal processes at some specific time points while the Integral sensors defined later in \ref{Temporal: IntegralObservation} take integral observations that are the averages of realizations of the temporal processes over time intervals with additive noise. In practice, Point sensors are commonly used to measure real-time features such as daily temperature, whereas Integral sensors are used to track cumulative features such as total daily precipitation \cite{allart2015disaggregating}, computed tomography (CT) scans \cite{Tanskanen2020}, and areal data \cite{tanaka2019spatially}. Therefore, both types of sensors can be utilized to monitor various temporal phenomena that stem from the same spatial phenomenon and enhance the performance of spatial field reconstruction \cite{Sahu2010}.

A number of studies have been devoted to developing tractable methods to reconstruct latent spatial field through Gaussian Process-based data. In those studies, authors commonly adopt the hierarchical Bayesian framework and resort to Markov Chain Monte Carlo (MCMC) \cite{Sahu2010} (in particular, Hamiltonian Monte Carlo \cite{Wang2019}) to infer the posterior predictive distribution. Alternatively, the posterior predictive distribution can be derived analytically when conjugacy is assumed in the model \cite{Ido2015}. However, MCMC-based methods are not suitable for accommodating the continuous online inflow of sensor observations due to their sophisticated and time-consuming nature, whereas conjugacy is not present in the set-up using WGPs. 

The contributions of this paper are four-fold:
\begin{enumerate}
\item We propose a novel model to represent the hierarchical spatial-temporal physical phenomenon using WGPs such that the temporal processes may follow arbitrary distributions that appear in real applications.
\item We develop the Warped Gaussian Process Likelihood Ratio Test (WGPLRT) and the Neighborhood-density-based Likelihood Ratio Test (NLRT) tailored to approximately performing Likelihood Ratio Tests on time-series data for sensors collecting point or integral observations, respectively.
\item We derive the Spatial Best Linear Unbiased Estimator (S-BLUE) for aggregating outputs of Likelihood Ratio Tests and reconstructing the spatial phenomenon, which is computationally efficient.
\item We perform both synthetic data experiments and real-world experiments to validate our model and algorithms. In the real-world experiments, we use a weather dataset from the National Environment Agency (NEA) of Singapore that includes fields such as temperature and relative humidity to support the proposed approach.
\end{enumerate} 

The rest of the paper is structured as follows. We introduce the definitions of GPs and WGPs and present the system model in Section \ref{Sec:Definition}.
We show that the posterior predictive distribution in the proposed problem is analytically intractable and motivate a two-step procedure to reconstruct the binary spatial random field in Section \ref{Sec:Posterior}.  
In Section \ref{Sec:LRTs}, we develop the Warped Gaussian Process Likelihood Ratio Test (WGPLRT) and the Neighborhood-density-based Likelihood Ratio Test (NLRT) for inferring temporal processes. 
Section \ref{Sec:S-BLUE} introduces the Spatial Best Linear Unbiased Estimator (S-BLUE) and its properties. In Sections \ref{Sec:Experiments Synthetic} and \ref{Sec:Experiments Semisynthetic}, we showcase the proposed model and algorithm by performing experiments using synthetic and real-world datasets\footnote{Codes implemented in MATLAB can be found on GitHub: \url{https://github.com/ShunanSheng/WarpedGaussianProcesses}.}, respectively. Finally, we conclude the paper in Section \ref{Sec:Conclusion}.

\section{Definitions \& System Model }\label{Sec:Definition}
In this section, we introduce the definitions of Gaussian Processes (GPs) and Warped Gaussian Processes (WGPs) and present the system model. Throughout this paper, all random variables are defined on a probability space $(\Omega,\CF,\PROB)$. 
Let us first state the definition of a GP\footnote{We use Gaussian Process and Gaussian random field interchangeably.}.
\begin{definition}[Gaussian Process {(see\textit{,} e.g., \cite[Definition~2.1]{C.E.Rasmussen2006})}] 
\label{DefGP}
Let $\mathcal{X} \subseteq \R^d$ and let $f:\Omega \times \CX\to \R$ denote a stochastic process parametrized by $\Bx \in \CX$. Then, $f$ is a Gaussian Process (GP) with the mean function $\mu:\CX\to\R$ and the covariance function $\CC:\CX\times\CX\to\R$, i.e., $f
\sim\GP\left(\mu \left(\cdot\right),\CC\left(\cdot,\cdot\right)\right)$, if all its finite dimensional distributions are Gaussian, that is, for any $m \in \mathbb{N}$ and $\Bx_{1:m} := (\Bx_1,\cdots, \Bx_m) \in \CX^m$, the random variables $\left(f\left(\Bx_1\right),\ldots, f\left(\Bx_m\right)\right)^\TRANSP$ are jointly normally distributed with mean $\mu(\Bx_{1:m}):=(\mu(\Bx_1),\cdots, \mu(\Bx_m))^\TRANSP \in \R^m$ and covariance matrix $\CC(\Bx_{1:m},\Bx_{1:m}) \in \R^{m\times m}$ where $\left(\CC(\Bx_{1:m},\Bx_{1:m})\right)_{i,j}$ $:=\CC(\Bx_i,\Bx_j)$ for $1\leq i,j\leq m$.
\end{definition}

We can therefore characterize a GP by the following class of random functions:
\begin{align*}
\begin{split}
\mathfrak{F} :=
\{
& f : \Omega\times\mathcal{X} \to  \mathbb{R}\;
\text{s.t.}\; f
\sim\GP\left(\mu \left(\cdot\right),\CC\left(\cdot,\cdot\right)\right),\;\\
&\;\mathrm{with}\; \mu:\CX \to  \mathbb{R},\; \Bx\mapsto \EXP\left[f(\Bx)\right],\\
&\; \CC:\mathcal{X} \times \mathcal{X} \to  \mathbb{R}
,\; (\Bx,\Bx')\mapsto \EXP\left[\left(f(\Bx)-\mu(\Bx)\right)
\left(f(\Bx')-\mu(\Bx')\right)
\right]\}.
\end{split}
\end{align*}
Subsequently, we define a Warped Gaussian Process (WGP) as the point-wise transformation of a GP, as detailed below.

\begin{definition}[Warped Gaussian Process]\label{Definition:WGP}
Let $z: \Omega \times \CX\to \R$ be a stochastic process indexed by $\Bx\in\CX\subseteq \R^d$.
We call $z$ a Warped Gaussian Process (WGP) if it is the point-wise transformation of a GP $f:\Omega\times\CX\to\R$ under a Borel-measurable warping function $W: \R \to \R$, that is,
\begin{align*}\label{WGP}
\begin{split}
z(\Bx) = W\Big(f(\Bx)\Big) \qquad \text{for all } \Bx\in \CX.
\end{split}
\end{align*}
\end{definition}

A candidate for the warping function is $W=F^{-1}\circ \Phi$, where $F$ is the cumulative distribution function (CDF) of  a random variable, $F^{-1}(u) :=\inf \left\{x\in \R : F(x)\geq u\right\}$ for $u\in [0,1]$ is the Generalized Inverse of $F$, and $\Phi$ is the CDF of the standard normal distribution $\CN(0,1)$. Consequently, if $f\sim \GP(0,\CC(\cdot,\cdot))$ with $\CC(\Bx,\Bx)=1$ for all $\Bx\in \CX$, then the CDF of $z(\Bx)$ is $F$ for all $\Bx\in\CX$.

Having formally defined the semi-parametric class of WGP models, we proceed with presenting our system model. The system model consists of the two following parts:
\begin{enumerate}
    \item The physical phenomenon: a collection of spatial-temporal random processes including a binary spatial random field and a collection of temporal processes whose characteristics are based on the local values of the binary spatial field\footnote{We use binary spatial random field and binary spatial field interchangeably.} (see \ref{Temporal:LatentProcess} and \ref{Temporal:WarpedProcess} below).
    \item  The sensor network: a wireless sensor network is deployed to observe the local temporal processes at specific spatial locations. The wireless sensor network consists of two types of sensors: the first collects point observations of the temporal process at specific time points (see \ref{Temporal: PointObservation} below); the second collects integral observations of the temporal process over time intervals (see \ref{Temporal: IntegralObservation} below).
\end{enumerate}
\pagebreak
We now present the system model: % label =[text\arabic*]
\begin{enumerate}[label = (BSF\arabic*), leftmargin = 40 pt]
\item[] \noindent\rule{5.05cm}{0.5pt}
Binary spatial random field
\noindent\rule{5.1cm}{0.5pt} 
\item Consider a latent spatial random field $g:\Omega\times \CX\to \R$ defined over $\CX\subset \R^2$, which is modeled as a GP with mean function  $\mu(\cdot)$ and covariance function $\CC(\cdot,\cdot)$, that is
\begin{equation}
    g\sim \GP\Big(\mu(\cdot),\CC(\cdot,\cdot)\Big).
\end{equation}

\item The binary spatial random field $y:\Omega \times \CX \to \{0,1\}$ is defined to be the point-wise transformation of the latent spatial random field $g$ such that 
\begin{equation}\label{Spatial:definition}
    y(\Bx):=\INDI_{\{g(\Bx)\geq c\}},
\end{equation}
where $c\in\R$ is a constant threshold. 
\end{enumerate}

\begin{enumerate}[label = (TP\arabic*), leftmargin = 40 pt]
\item[] \noindent\rule{5.7cm}{0.5pt}
Temporal processes
\noindent\rule{5.8cm}{0.5pt}
\item \label{Temporal:LatentProcess}
At each spatial location $\Bx\in \CX$, let $f(\cdot\;;\Bx):\Omega\times [0,T]\to\R$ be a temporal latent GP for some $T > 0$ with characteristics depending on the value of $y(\Bx)$ such that
 \begin{equation}
      f(\cdot;\Bx)\sim
      \begin{cases}
      \GP(0,\CC_0(\cdot,\cdot)),\quad \text{if}\quad y(\Bx)=0,\;\;\;\\
\GP(0,\CC_1(\cdot, \cdot)),\quad \text{if}\quad y(\Bx)=1,\;\;\; 
\end{cases}
 \end{equation}
 where $\CC_i:[0,T]\times [0,T]\to \R$, for $i=0,1$, are the covariance functions of the respective temporal processes with $\CC_i(t,t)=1$ for all $t\in[0,T]$. Given any finite collection of spatial locations, such as $\Bx_1,\cdots,\Bx_n \in \CX$ for $n\in \N$, the random variables $f(\cdot\,;\Bx_1),\cdots, f(\cdot\,;\Bx_n)$ are assumed to be independent conditional on $(g(\Bx_1),\cdots, g(\Bx_n))$.

\item \label{Temporal:WarpedProcess}
The temporal process at $\Bx$ is defined to be a WGP $\tilde z(\cdot;\Bx): \Omega \times [0,T] \to \R$ via a point-wise transformation of $f$, depending on the value of $y(\Bx)$ such that
\begin{equation}
    \tilde z(t;\Bx):=\begin{cases}
    W_0(f(t;\Bx)),\quad \text{if}\quad y(\Bx)=0,\;\;\;\\
    W_1(f(t;\Bx)),\quad \text{if}\quad y(\Bx)=1, \;\;\;
    \end{cases}
\end{equation}
where $W_i=F_i^{-1}\circ \Phi$, for $i=0,1$. Each $F_i$ is the cumulative distribution function (CDF) of  a random variable, $F_i^{-1}$ is the Generalized Inverse of $F_i$, and $\Phi$ is the CDF of the standard normal distribution $\CN(0,1)$. 
\end{enumerate}

\begin{enumerate}[label = (SN\arabic*), leftmargin = 40 pt]
\item[] \noindent\rule{6.15cm}{0.5pt}
Sensor network
\noindent\rule{6.15cm}{0.5pt}
\item \label{Temporal: independence}
Let $N\in\N$ be the total number of sensors that are deployed over the 2-dimensional space $\CX\subseteq\R^2$ to make observations over the time period $[0,T]$. 
There are two types of sensors that are referred to as Point sensors (abbreviated to P-sensors) and Integral sensors (abbreviated to I-sensors). We assume that there are $N^{\TP}$ P-sensors deployed at $(\Bx_n^{\TP})_{n=1:N^{\TP}}$ and $N^{\TI}$ I-sensors deployed at $(\Bx_n^{\TI})_{n=1:N^{\TI}}$, where $N^{\TP}, N^{\TI}\in \N$ and $N^{\TP}+N^{\TI}=N$. Note that $\big\{f(\cdot\;;\Bx^{\TP}_n)\big\}_{n=1:N^{\TP}}$, $\big\{f(\cdot\;;\Bx^{\TI}_n)\big\}_{n=1:N^{\TI}}$ are independent conditional on $\left(g(\Bx^{\TP}_1),\ldots,g(\Bx^{\TP}_{N^{\TP}}),g(\Bx^{\TI}_1),\ldots,g(\Bx^{\TI}_{N^{\TI}})\right)$.

\item \label{Temporal: PointObservation}
\underline{Point sensors:} 
Point sensors collect noisy observations of the temporal process at the time instants $(t_k)_{k=1}^M\subset[0,T]$ for some $M\in \N$. At each time instant $t_k$, the $n$-th P-sensor makes a noisy observation, which is given by
\begin{equation}
    z_{n,k}^{\TP}:=\tilde{z}(t_k;\Bx^{\TP}_n)+\epsilon^{\TP}_{n,k},
\end{equation}
where $\epsilon^{\TP}_{n,1},\cdots,\epsilon^{\TP}_{n,M}\overset{i.i.d.}{\sim} \NORMAL(0,\sigma^2_{\TP})$ for some noise variance $\sigma_{\TP}^2>0$ with $\epsilon^{\TP}_{n,k}$ and $\tilde{z}(t_k;\Bx^{\TP}_n)$ being independent for all $k=1,\cdots,M$.

\item \label{Temporal: IntegralObservation}
\underline{Integral sensors:} Integral sensors collect integral observations of the temporal process over consecutive time intervals $\big[\frac{(k-1)T}{K},\frac{kT}{K}\big]$ for $k=1,\cdots,K$, for some fixed $K\in\N$. The noisy observation at $\Bx_n^{\TI}$ over the time interval $\big[\frac{(k-1)T}{K},\frac{kT}{K}\big]$ is given by
\begin{equation}
  z^{\TI}_{n,k}:=\frac{K}{T}\int_{\frac{(k-1)T}{K}}^{\frac{kT}{K}}\tilde z(t;\Bx^{\TI}_n)\DIFFX{t}+\epsilon^{\TI}_{n,k},
\end{equation}
where $\epsilon^{\TI}_{n,1},\cdots,\epsilon^{\TI}_{n,K}\overset{i.i.d.}{\sim} \NORMAL(0,\sigma^2_I)$ for some noise variance $\sigma_{\TI}^2>0$ with $\epsilon^{\TI}_{n,k}$ and $\tilde z(t;\Bx^{\TI}_n)$ being independent for all $t\in[0,T]$. 
\end{enumerate}

\begin{table}[t] 
\centering
\resizebox{\columnwidth}{!}{%
\begin{tabular}{ r l }
 \hline
 \textbf{Symbol} & \textbf{Interpretation} \\
 \hline
$\Bx^{\TP}_n\in \CX$ & The spatial location of the $n$-th P-sensor for $n=1,\cdots,N^{\TP}$ \\
$\Bx^{\TI}_n\in \CX$  & The spatial location of the $n$-th I-sensor for $n=1,\cdots,N^{\TI}$ \\
$X_{1:N}:=(\Bx_1^{\TP},\cdots,\Bx_{N^{\TP}}^{\TP},$&The spatial locations of the $N$ sensors deployed in the field \\
$\Bx^{\TI}_1,\cdots,\Bx^{\TI}_{N^{\TI}})$ & \\ 
$T^{\TP}_{1:M}:=(t_1,\cdots,t_M)$ & The time instants at which the P-sensors collect observations  \\
 $T^{\TI}_{1:K}:=\Big\{\big[\frac{(k-1)T}{K},\frac{kT}{K}\big]\Big\}_{k=1}^K$&  The time intervals during which the I-sensors collect observations \\
$\BIg:=g(X_{1:N})$ & The realizations of the latent spatial random field $g$ at $X_{1:N}$\\
$y_{n}^{\TP}:=y(\Bx_{n}^{\TP})$ & The realization of the binary spatial random field $y$ at $\Bx_n^{\TP}$\\
 $y_{n}^{\TI}:=y(\Bx_{n}^{\TI})$ & The realization of the binary spatial random field $y$ at $\Bx_n^{\TI}$\\
 $\BZ^{\TP}_n:=(z^{\TP}_{n,1},\cdots,z^{\TP}_{n,M})^\TRANSP$ & The collection of the point observations at $\Bx_n^{\TP}$ over $T_{1:M}^{\TP}$ \\
 $\tilde \BZ^{\TP}_n:=(\tilde z^{\TP}_{n,1},\cdots,\tilde z^{\TP}_{n,M})^\TRANSP$ & The ground-truth values of the point observations at $\Bx_n^{\TP}$ over $T_{1:M}^{\TP}$\\

$\BZ^{\TI}_n:=(z^{\TI}_{n,1},\cdots,z^{\TI}_{n,K})^\TRANSP$ & The collection of the integral observations at $\Bx_n^{\TI}$ over $T_{1:K}^{\TI}$\\
 $\tilde \BZ^{\TI}_n:=(\tilde z^{\TI}_{n,1},\cdots,\tilde z^{\TI}_{n,K})^\TRANSP$ & The ground-truth values of the integral observations at $\Bx_n^{\TI}$ over $T_{1:K}^{\TI}$
\end{tabular}
}
\caption{Symbols used in the following sections.}
\label{Table:symbols}
\end{table}

A comprehensive list of symbols used in the following sections is provided in Table \ref{Table:symbols} and the graphical structure of the proposed model encoding conditional independence relations is presented in Figure \ref{fig:plate} as a directed acyclic-graph (DAG) using plate notations. The main objective of this paper is \textbf{Binary spatial random field reconstruction}, abbreviated as \textbf{Spatial field reconstruction}. In other words, given an un-monitored location $\Bx_*\in \CX$, we want to infer the value of the binary spatial random field $y_*:=y(\Bx_*)$ based on the data transmitted by the $N$ sensors.  

\begin{remark}\label{Remark: spatial field}
 If $\mu(\Bx)=0$ and $\CC(\Bx,\Bx)=1$ for all $\Bx \in \CX$, then the binary spatial random field satisfies $y(\Bx)=\INDI_{\{g(\Bx)\geq c\}}=F_B^{-1}(\Phi(g(\Bx)))$, $\PROB$-a.s., where $F_B^{-1}$ is the Generalized Inverse of $\mathrm{Bernoulli}(\pi)$ with $\pi=1-\Phi(c)$.
\end{remark}
\begin{assumption}
Throughout this paper, we assume that:
\begin{enumerate}
\item the integral observations in \ref{Temporal: IntegralObservation} are well-defined, that is,
\[
\int_{\frac{(k-1)T}{K}}^{\frac{kT}{K}} \left|\tilde z(t;\Bx^{\TI}_n) \right|\DIFFX{t}<\infty\quad \text{for all } k=1,\cdots,K,\,n=1,\cdots,N^{\TI};
\]
\item all covariance functions are symmetric and positive definite (see details in \cite[Section~2.1]{C.E.Rasmussen2006}).
\end{enumerate}
\end{assumption}

\begin{figure}[t]
\resizebox{\columnwidth}{!}{%
\begin{tikzpicture}
	[inner sep=2mm,
	latent/.style={circle,draw=black!100,fill=white!20,thick,minimum size=17mm},
	obs/.style={circle,draw=black!100,fill=black!20,thick,minimum size=17mm},
	fac/.style={rectangle,draw=black!100,fill=black!30,thick,minimum width=17mm,minimum height=17mm}]

 % Define nodes
    \node[latent]		(g)     {$g$};
    \node[latent]		(yx) 	[right=of g]	{$y(\Bx)$};
    \node[fac]		(sg) [above=of g]	{$\mu,\CC$};
    \node[fac]		(sc) [above=of yx]	{$c$};
    \node[latent]		(f)	[right=of yx]	{$f(\cdot;\Bx)$};
    \node[fac]		(x)	[below=of g]	{$\Bx$};
    \node[latent]		(ztilde) [right=of f] 	{$\tilde z(\cdot;\Bx)$};
    \node[obs]		(zP)	[below=of f]	{$z_{n,k}^{\TP}$};
    \node[obs]		(zI)	[below=of ztilde]	{$z_{n,k}^{\TI}$};
    \node[fac]		(tf)	[above=of f]		{$\mu_i,\CC_i$};
    \node[fac]		(tw)	[above=of ztilde]		{$W_i$};
    \node[latent]		(ep)	[below=of zP]		{$\epsilon_{n,k}^{\TP}$};
    \node[latent]		(eI)	[below=of zI]		{$\epsilon_{n,k}^{\TI}$};
    \node[fac]		(np)	[left=of ep]		{$\sigma^2_{\TP}$};
    \node[fac]		(nI)	[right=of eI]		{$\sigma^2_{\TI}$};
    \node[fac]		(tp)	[left=of zP]		{$t_k$};
    \node[fac]		(tI)	[right=of zI]		{$\left[\frac{(k-1)T}{K},\frac{kT}{K}\right]$};
    
    % Connect the nodes
    \draw [->] (sg) -- (g);
    \draw [->] (sc) -- (yx);
    \draw [->] (tf)--(f);
    \draw [->] (tw) -- (ztilde);    
    \draw [->] (g) -- (yx);
    \draw [->] (yx) -- (f);
    \draw [->] (f) -- (ztilde);
    \draw [->] (ztilde) -- (zP);    
    \draw [->] (ztilde) -- (zI);    
    \draw [->] (x)--(yx);
    \draw [->] (tp) -- (zP);
    \draw [->] (tI) -- (zI);        
    \draw [->] (np) -- (ep);
    \draw [->] (ep) -- (zP);
    \draw [->] (nI) -- (eI);
    \draw [->] (eI) -- (zI);            
    
    % Add plates
    \tikzset{plate caption/.append style={below right=1pt and -35pt of #1.north east}}
     \plate [inner sep=0.45cm, xshift=-0.4cm, yshift=-0.25cm,color=black] {px} {(g)(nI)} {$\Bx=\Bx_1^{\TP},\cdots,\Bx_{N^{\TP}}^{\TP},\Bx^{\TI}_1,\cdots,\Bx^{\TI}_{N^{\TI}}$};
     \tikzset{plate caption/.append style={below right=5pt and 0pt of #1.south west}}
     \plate [inner sep=0.3cm, xshift=0.01cm, yshift=0.15cm,color=black] {pp} {(tp)(ep)} {$k=1,\cdots,M$}; 
     \plate [inner sep=0.3cm, xshift=0.1cm, yshift=0.15cm,color=black] {pI} {(tI)(eI)} {$k=1,\cdots,K$}; 
     \tikzset{plate caption/.append style={below right=3pt and 0pt of #1.south east}}
      \plate [inner sep=0.3cm, xshift=0.1cm, yshift=0.25cm,color=black] {pt} {(tf)(tw)} {$i=0,1$}; 
\end{tikzpicture}}
\caption{Directed acyclic graph (DAG) of the model encoding conditional independence relations using plate notations. The shaded rectangles represent the constants and observed covariates. The white circles indicate the latent variables. The shaded circles ensemble the observed random variables. }
\label{fig:plate}
\end{figure}
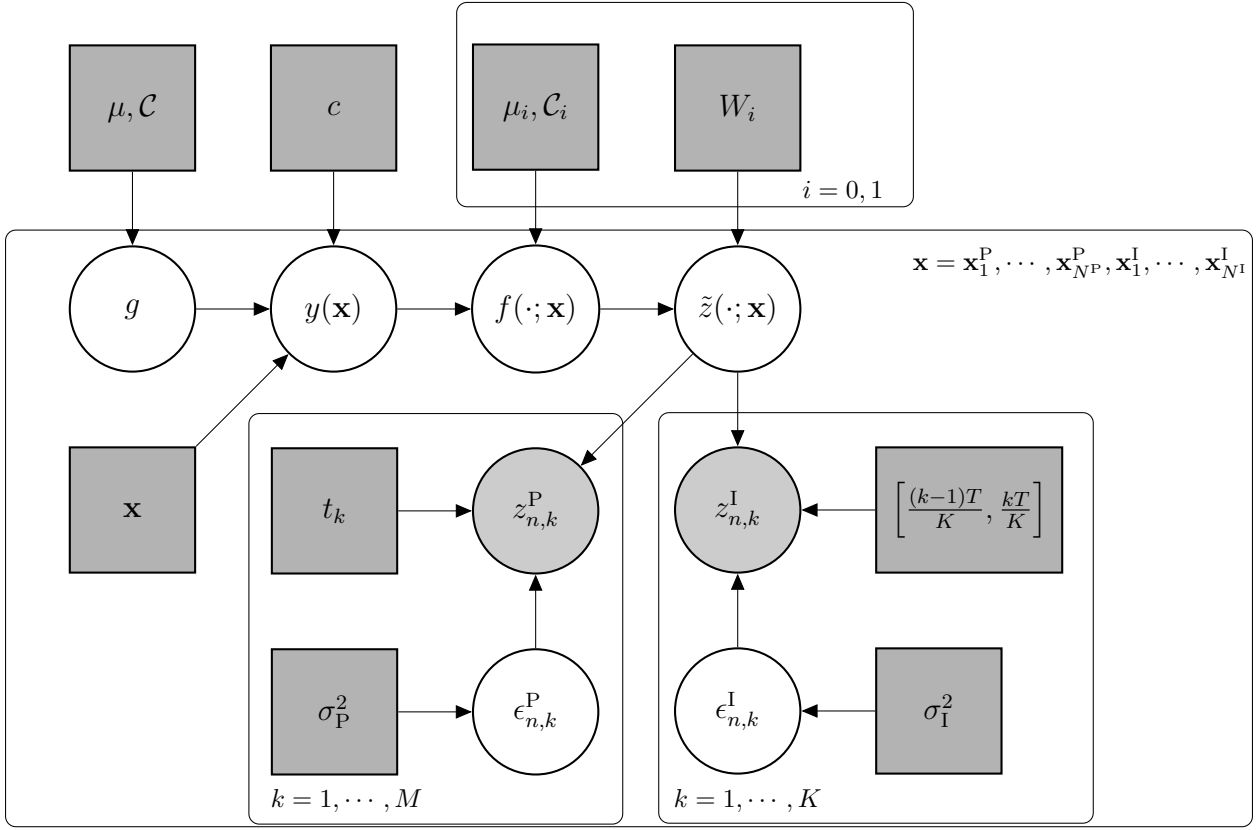

\section{Posterior Predictive Distribution \& Our Approach}
\label{Sec:Posterior}
In this section, we introduce the posterior predictive distribution for inferring the values of the binary spatial random field and highlight its computational intractability in (\ref{Spatial:Y_posterior}). The posterior predictive distribution of $y_*:=y(\Bx_*)$ at some un-monitored spatial location $\Bx_*$ given the Gaussian prior on $\BIg$, defined by $\BIg \sim \CN\left(\mu(X_{1:N}), \CC(X_{1:N},X_{1:N})\right)$, as well as the sensor observations $(\BZ^{\TP}_{1:N^{\TP}},$ $\BZ^{\TI}_{1:N^{\TI}})$ is presented in the following Proposition.
\begin{proposition}[Posterior predictive distribution]
\label{L1}
The posterior predictive distribution of $y_*$ given the sensor observations $(\BZ^{\TP}_{1:N^{\TP}}, \BZ^{\TI}_{1:N^{\TI}})$ is $\mathrm{Bernoulli}(\pi_*)$ with $\pi_*$ given by:
\begin{align}
\begin{split}
\label{Spatial:Y_posterior}
\pi_*&:=
\int_{c}^{\infty}
\int_{\R^N} p(g_*|\BIg)
\frac{
\big(\prod_{n=1}^{N^{\TP}}p(\BZ^{\TP}_{n}|\BIg)\big)\big(\prod_{n=1}^{N^{\TI}}p(\BZ^{\TI}_{n}|\BIg)\big)
p(\BIg)}
{\int_{\R^N}
\big(\prod_{n=1}^{N^{\TP}}p(\BZ^{\TP}_{n}|\BIg')\big)\big(\prod_{n=1}^{N^{\TI}}p(\BZ^{\TI}_{n}|\BIg')\big)
p(\BIg')\DIFFX{\BIg'}}\DIFFX{\BIg}
\DIFFX{g_*}.
\end{split}
\end{align}
Moreover, the term
\begin{align*}
\int_{\R^N}\textstyle
\big(\prod_{n=1}^{N^{\TP}}p(\BZ^{\TP}_{n}|\BIg')\big)\big(\prod_{n=1}^{N^{\TI}}p(\BZ^{\TI}_{n}|\BIg')\big)
p(\BIg')\DIFFX{\BIg'}
\end{align*}
in \eqref{Spatial:Y_posterior} contains a sum of $2^{N^{\TI}+N^{\TP}}$ terms, and hence is computationally intractable.
\end{proposition}
\begin{proof}
See \ref{L1_P}.
\end{proof}

Despite the fact that the Bernoulli parameter $\pi_*$ of the posterior predictive distribution is computationally intractable, the conditional independence of $\BZ_n^{\TP},\BZ_n^{\TI}$ given $\BIg$ motivates us to process raw observations at each sensor and aggregate the processed data at the Fusion Center. Specifically, we adopt the following two-step approach to circumvent the difficulty of evaluating the computationally intractable posterior predictive distribution. 

\begin{itemize}
\item \underline{Step 1}: For $n=1,\cdots, N$, the $n$-th sensor performs inference based on its observations of the temporal process and transmits a binary decision $\widehat y_n^j$ as an estimator of $y_n^{j}$, $j\in \{\TP,\TI\}$. The decision process for Point sensors is detailed in Section~\ref{Temporal:Point} and Algorithm~\ref{Temporal:WGPLRT}, and the decision process for Integral sensors is detailed in Section~\ref{Temporal:Integral} and Algorithm~\ref{Temporal:NLRT}.
\item \underline{Step 2}: The Fusion Center collects all the binary decisions from the sensors, denoted by $\widehat \BY_{1:N}=(\widehat y^{\TP}_1,\cdots,\widehat y^{\TP}_{N^{\TP}},\widehat y^{\TI}_1,$ $\cdots,\widehat y^{\TI}_{N^{\TI}})^\TRANSP$, and applies the Spatial Best Linear Unbiased Estimator (S-BLUE) to aggregate the results and predict $y_*=y(\Bx_*)$ at an un-monitored location $\Bx_*$. This is detailed in Section~\ref{Sec:S-BLUE} and Algorithm~\ref{SBLUE:algo}.
\end{itemize}

\begin{remark}
Markov Chain Monte Carlo (MCMC) \cite{gelman2013bayesian}, such as Gibbs sampling,  could also be employed to approximate the posterior predictive distribution. However, the Gibbs sampler tends to be inefficient in practice due to the high correlation in the posterior distribution over $\BIg$ \cite{gelman2013bayesian,neil2008}. Although alternative sampling schemes such as, e.g., sampling using control variables \cite{neil2008} can address this issue, employing those methods for inference often requires communicating complete observations among the sensor network, which is extremely costly. Consequently, MCMC is not suitable in our context, as it does not align with our objective of developing an efficient yet simple algorithm.
\end{remark}

\section{Step 1: Local Binary Decisions via Likelihood Ratio Tests}
\label{Sec:LRTs}
To infer a binary decision from noisy observations, it is natural to use the Likelihood Ratio Test as it is the uniformly most powerful test given a fixed significance level (see\textit{,} e.g.,\ \cite[Section~\rom{3} (a)]{Neyman1933}). Let us first model the decision problem into a hypothesis testing problem. At the spatial location $\Bx_n^j$ for $j \in \{\TP,\TI\}$, the null and alternative hypotheses are given by  
\[
\CH_0: y(\Bx_n^j)=0\qquad \text{and} \qquad \CH_1: y(\Bx_n^j)=1.
\]

Let $p(\BZ_n^{j}|\CH_i)$, $i=0,1$, denote the marginal likelihood of $\BZ_n^{j}$ under the hypothesis $\CH_i$ and let $\gamma^{\TI},\gamma^{\TP}>0$ be the test thresholds for the P-sensors and the I-sensors, respectively. The Likelihood Ratio Test uses the following test statistic:
\begin{equation}\label{Temporal:test statistic}
    \Lambda(\BZ_n^{j})=\frac{p(\BZ_n^{j}|\CH_0)}{p(\BZ_n^{j}|\CH_1)}
    \begin{array}{c }
\CH_0\\
\geq\\
<\\
\CH_1
\end{array}
\gamma^j\;\;\; \text{for}\;j\in \{\TP,\TI\}.
\end{equation}
In the following subsections, we propose two methods for approximating the LRT. Specifically, we propose the Warped Gaussian Process Likelihood Ratio Test (WGPLRT) for point observations and we propose the Neighborhood-density-based Likelihood Ratio Test (NLRT) for integral observations.

\subsection{Local Binary Decisions for Point Observations}\label{Temporal:Point}
In this subsection, we present an algorithm called \textbf{Warped Gaussian Process Likelihood Ratio Test (WGPLRT)} for performing local inferences at the Point sensors. The WGPLRT exploits the Laplace Approximation to estimate the marginal likelihood function of point observations and provides a formula to approximate the test statistic defined in (\ref{Temporal:test statistic}). 

To begin with, the following Proposition gives the analytic expression for the marginal likelihood of $\BZ_n^{\TP}$ under {$\CH_0$ and $\CH_1$. 

\begin{proposition}[Marginal likelihood of $p(\BZ_n^{\TP}|\CH_i)$]\label{L2}For $i=0,1$, assume that the warping function is $W_i = F_i^{-1} \circ \Phi$ where $F_i$ is strictly increasing and is the CDF of a continuous random variable with continuous density. Let $G_i:= \Phi^{-1}\circ F_i$ be the inverse of $W_i$ and let $K_i:=\CC_i(T^{\TP}_{1:M},T^{\TP}_{1:M})\in \R^{M\times M}$ be the covariance matrix evaluated at $T^{\TP}_{1:M}$. Let $\Bepsilon_n^{\TP}=(\epsilon^{\TP}_{n,1},\cdots,\epsilon^{\TP}_{n,M})^\TRANSP$ be the additive noise at $\Bx_n^{\TP}$. Then, the marginal likelihood at the $n$-th P-sensor is given by
\begin{equation}\label{Temporal:WGPmarginal}
\begin{split}
  p(\BZ^{\TP}_n|\CH_i)&=\int_{\R^M}\exp\bigg(\!\!-\frac{1}{2}G_i(\BZ^{\TP}_n-\Bepsilon_n^{\TP})^\TRANSP K_i^{-1}G_i(\BZ^{\TP}_n-\Bepsilon_n^{\TP})-\frac{1}{2}\log \det K_i-\frac{M}{2}\log{2\pi}
 \\&\qquad\qquad+\sum_{m=1}^M \log \frac{\partial G_i(z)}{\partial z}\bigg|_{z^{\TP}_{n,m}-\epsilon^{\TP}_{n,m}}\bigg )(2\pi\sigma_{\TP}^2)^{-\frac{M}{2}}\exp{\left(-\frac{1}{2}\sigma_{\TP}^{-2}(\Bepsilon_n^{\TP})^{\TRANSP}\Bepsilon_n^{\TP}\right)}\DIFFX \Bepsilon_n^{\TP}.
\end{split}
\end{equation}
\end{proposition}
\begin{proof}
See \ref{L2_P}.
\end{proof}

When the sensor observes the ground-truth values of the temporal process without observation errors, i.e., $\BZ^{\TP}_n = \tilde \BZ^{\TP}_n$, then the marginal likelihood $  p(\BZ^{\TP}_n|\CH_i)$ for $i=0,1$ in \eqref{Temporal:WGPmarginal} reduces to the case suggested in Section 3 in \cite{Snelson2003}.

However, the marginal likelihood in \eqref{Temporal:WGPmarginal} is computationally intractable as the integrand cannot be expressed as the probability density function of a multivariate normal distribution. Our approach is to approximate the inner term involving $G_i(\BZ_n^{\TP}-\Bepsilon_n^{\TP})$ using the \textbf{Laplace approximation}. The details of the Laplace approximation as well as the approximated marginal likelihood function of $p(\BZ_n^{\TP}|\CH_i)$ is given by the following Proposition.
\begin{proposition}[Laplace approximation]
\label{Temporal: Laplace Approximation}
For $i=0,1$, let
\begin{equation}
\label{Temporal:Q}
   \left(\text{range}(W_i)\right)^M\ni(v_1,\cdots,v_M)^\TRANSP=\Bv \mapsto  Q_i(\Bv):=-\frac{1}{2}G_i(\Bv)^TK_i^{-1}G_i(\Bv)+\sum_{m=1}^M \log \frac{\partial G_i(v)}{\partial v}\Big|_{v_m},
\end{equation}
where $K_i:=\CC_i(T^{\TP}_{1:M},T^{\TP}_{1:M})$ for $i=0,1$ are defined as in Proposition \ref{L2}.
Assume that there is a $\widehat \Bv_i \in \left(\text{range}(W_i)\right)^M$ satisfying the following conditions:
\begin{enumerate}
\item $\widehat \Bv_i $ maximizes $Q_i$;
\item $\widehat \Bv_i $ is in the interior of $\left(\text{range}(W_i)\right)^M$;
\item $Q_i$ is twice differentiable at $\widehat \Bv_i$;
\item $A_i:=-\nabla^2 Q_i(\Bv)|_{\Bv=\widehat \Bv_i}$ is positive definite. 
\end{enumerate}
Moreover, let $Q_i(\Bv)$ be approximated by its second-order Taylor polynomial as
\begin{equation}\label{Temporal:Qhat}
    \widehat Q_i(\Bv):=-\frac{1}{2}(\Bv-\widehat \Bv_i)^\TRANSP A_i(\Bv-\widehat \Bv_i)+Q_i(\widehat \Bv_i),\quad \Bv\in \left(\text{range}(W_i)\right)^M,
\end{equation}
and the approximated marginal likelihood function be given by
\begin{align}\label{Temporal:approximated likelihood}
\begin{split}
    \widehat p(\BZ_n^{\TP}|\CH_i)&:=\int_{\R^M}\exp\left(\widehat Q_i(\BZ_n^{\TP}-\Bepsilon_n^{\TP})-\frac{1}{2}\log \det K_i-\frac{M}{2}\log{2\pi}\right )\\
    &\qquad\qquad\times(2\pi\sigma_{\TP}^2)^{-\frac{M}{2}}\exp{\left(-\frac{1}{2}\sigma_{\TP}^{-2}(\Bepsilon_n^{\TP})^{\TRANSP}\Bepsilon_n^{\TP}\right)}\DIFFX \Bepsilon_n^{\TP}.
\end{split} 
\end{align} 
Then,
\begin{equation}
      \widehat p(\BZ_n^{\TP}|\CH_i) = \widehat C_i \exp\left(-\frac{1}{2}(\BZ_n^{\TP}-\widehat \Bv_i)^\TRANSP (A_i^{-1}+\sigma^{2}_{\TP}I)^{-1}(\BZ_n^{\TP}-\widehat \Bv_i)\right),
\end{equation}
where 
\[\widehat C_i:=\exp\left(-\frac{1}{2}\log \det K_i -\frac{M}{2}\log 2\pi-M\log \sigma_{\TP}+ Q(\widehat \Bv_i)-\frac{1}{2}\log\det (A_i+\sigma^{-2}_{\TP}I)\right).\]
\end{proposition}
\begin{proof}
See \ref{L3_P}.
\end{proof}

Consequently, the test statistic (\ref{Temporal:test statistic}) is approximated by
\begin{equation}\label{Temporal:WGPtest}
   \begin{aligned}
     -\log 
    \widehat \Lambda (\BZ^{\TP}_n)&=-\log \widehat p(\BZ^{\TP}_n|\CH_0)+\log \widehat p(\BZ^{\TP}_n|\CH_1)\\
     &=\frac{1}{2}\Big(\log\det (A_0+\sigma^{-2}_{\TP}I)+\log\det K_0-2Q(\widehat \Bv_0)-\log\det (A_1+\sigma^{-2}_{\TP}I)\\
     &\phantom{=} -\log\det K_1+2Q(\widehat \Bv_1)\Big) 
     +\frac{1}{2}(\BZ^{\TP}_n-\widehat \Bv_0)^\TRANSP (A_0^{-1}+\sigma_{\TP}^{2}I)^{-1}(\BZ^{\TP}_n-\widehat \Bv_0)\\
     &\phantom{=} -\frac{1}{2}(\BZ^{\TP}_n-\widehat \Bv_1)^\TRANSP (A_1^{-1}+\sigma_{\TP}^{2}I)^{-1}(\BZ^{\TP}_n-\widehat \Bv_1).
\end{aligned} 
\end{equation}

We use $-\log \widehat \Lambda(\BZ^{\TP}_n)$ as the test statistic in the WGPLRT and reject $\CH_0$ if $ -\log 
    \widehat \Lambda (\BZ^{\TP}_n)>-\log \gamma^{\TP}$. 
Notice that given $K_0$ and $K_1$, the terms $\widehat{\Bv}_0$, $\widehat{\Bv}_1$, $Q(\widehat{\Bv}_0)$, $Q(\widehat{\Bv}_1)$, $\log\det (A_0+\sigma^{-2}_{\TP}I)$, $\log\det (A_1+\sigma^{-2}_{\TP}I)$, $(A_0^{-1}+\sigma_{\TP}^{2}I)^{-1}$, and $(A_1^{-1}+\sigma_{\TP}^{2}I)^{-1}$ in (\ref{Temporal:WGPtest}) do not depend on the observations $\BZ^{\TP}_n$.
Thus, the test procedure can be divided into the \textbf{offline phase} and the \textbf{online phase}. 
In the offline phase, the values of $\widehat{\Bv}_0$, $\widehat{\Bv}_1$, $(A_0^{-1}+\sigma_{\TP}^{2}I)^{-1}$, $(A_1^{-1}+\sigma_{\TP}^{2}I)^{-1}$ as well as the constant term $\frac{1}{2}\big(\log\det (A_0+\sigma^{-2}_{\TP}I)+\log\det K_0-2Q(\widehat \Bv_0)-\log\det (A_1+\sigma^{-2}_{\TP}I)-\log\det K_1+2Q(\widehat \Bv_1)\big)$ are evaluated and stored.
Subsequently, in the online phase, after making the observations $\BZ^{\TP}_n$, the test statistic $-\log \widehat \Lambda(\BZ^{\TP}_n)$ is computed using the terms that are pre-computed in the offline phase.
Our test procedure is summarized in Algorithm~\ref{Temporal:WGPLRT}. 
    
\begin{algorithm}[t]
\SetAlgoLined
\KwData{the point observations $\BZ^{\TP}_n$ at the $n$-th P-sensor, the covariance matrix $K_i$ defined as in Proposition \ref{L2} and the warping function $W_i$ for $i=0,1$, the noise variance $\sigma^2_{\TP}>0$, and the test threshold $\gamma^{\TP}$.}
\KwResult{ the log-approximated test statistic $-\log 
    \widehat \Lambda (\BZ^{\TP}_n)$ and the binary decision $\widehat y^{\TP}_n$ at $\Bx_n^{\TP}$.}
\textbf{Offline phase:} Compute $\widehat \Bv_i=\arg \max_\Bv Q_i(\Bv)$ and $A_i=-\nabla^2 Q_i(\Bv)|_{\Bv=\widehat \Bv_i}$ for $i=0,1$.\
Compute the values of $(A_0^{-1}+\sigma_{\TP}^{2}I)^{-1}$, $(A_1^{-1}+\sigma_{\TP}^{2}I)^{-1}$, and $\frac{1}{2}\big(\log\det (A_0+\sigma^{-2}_{\TP}I)+\log\det K_0-2Q(\widehat \Bv_0)-\log\det (A_1+\sigma^{-2}_{\TP}I)-\log\det K_1+2Q(\widehat \Bv_1)\big)$.\label{Temporal:WGPLRT-offline}

%Compute $\widehat p(\BZ^{\TP}_n|\CH_i)$ for $i=0,1$ using (\ref{Temporal:approximated likelihood}) in Proposition \ref{Temporal: Laplace Approximation}.\

\textbf{Online phase:} After observing $\BZ^{\TP}_n$, compute $-\log 
    \widehat \Lambda (\BZ^{\TP}_n)$ from (\ref{Temporal:WGPtest}).\
The decision $\widehat y^{\TP}_n$ is given by 
\[
\widehat y^{\TP}_n=\begin{cases}
1,\quad \text{if } -\log 
    \widehat \Lambda (\BZ^{\TP}_n)>-\log \gamma^{\TP}\\
0, \quad \text{if } -\log 
    \widehat \Lambda (\BZ^{\TP}_n)\leq -\log \gamma^{\TP}.
\end{cases}
\]
 \caption{Warped Gaussian Process Likelihood Ratio Test (WGPLRT)}
 \label{Temporal:WGPLRT}
\end{algorithm}
\begin{remark}
The Laplace approximation in Proposition \ref{Temporal: Laplace Approximation} may perform poorly when \linebreak$\text{range}(W_i) \neq \R$ unless the distance between $\widehat{\Bv}_i$ and the boundary of $\left(\text{range}(W_i)\right)^M$ is large relative to $\sigma_{\TP}$. 
Moreover, since the Laplace approximation captures only the local characteristics of the integrand in \eqref{Temporal:WGPmarginal} around its maximum, WGPLRT sometimes approximates the LRT poorly when the only difference between $\CH_0$ and $\CH_1$ lies in the tails of the warping distributions.
\end{remark}    
We have now acquired the analytic expression of the approximated test statistic in the WGPLRT for point observations. We devote the next subsection to studying the case of integral observations.

\subsection{Local Binary Decisions for Integral Observations}\label{Temporal:Integral}
We now derive the algorithm called {\textbf{Neighborhood-density-based Likelihood Ratio Test (NLRT)} for making local binary decisions at the Integral sensors. At a glance, NLRT uses a Monte Carlo method to approximate the likelihood ratio. The samples generated in the Monte Carlo method can be used to compute the significance and power of this approximate likelihood ratio test. To begin with, for $i=0,1$, we are again interested in the marginal likelihood $p(\BZ^{\TI}_n|\CH_i)$, which is given by
\begin{equation}\label{Temporal: Integral1}
    p(\BZ^{\TI}_n|\CH_i)=\int_{\R^K} p(\BZ^{\TI}_n|\tilde \BZ^{\TI}_n)p(\tilde \BZ^{\TI}_n|\CH_i)d\tilde \BZ^{\TI}_n.
\end{equation}
 Note that $\BZ^{\TI}_n|\tilde \BZ^{\TI}_n\sim \CN(\tilde \BZ^{\TI}_n,\sigma^2_{\TI} \BI_K)$. 
Unfortunately, the marginal likelihood in \eqref{Temporal: Integral1} is also computationally intractable as the observations $\BZ_n^{\TI}$ are defined by integrals. To overcome this problem, we propose the NLRT, which is based on the idea of the Approximate Bayesian Computation method (see\textit{,} e.g., \cite[Section~2.1]{Toni_2008}).

In NLRT,  we generate samples under $\CH_0$ and $\CH_1$ and accept each sample if it is within a certain error tolerance $\delta^\TI>0$ to the given observation. Subsequently, the test statistic \eqref{Temporal:test statistic} is approximated by the ratio of the number of accepted samples under $\CH_0$ and $\CH_1$; see details in Algorithm~\ref{Temporal:NLRT}. To measure the distances between the generated samples and the observations, the Euclidean distance is a natural and effective candidate in most cases. However, for high-dimensional data, the Euclidean distance fails to define a meaningful notion of proximity as shown in Theorem~1 in \cite{Beyer1999}. Therefore, summary statistics are required to project high-dimensional data to a low-dimensional space while preserving their distinct characteristics under $\CH_0$ and $\CH_1$. In our case, depending on the type of the hypotheses, feature-engineering is needed for designing summary statistics: mode, mean, variance, kurtosis can be used to capture the differences in the warping functions; and autocorrelation (ACF) is helpful to detect the differences between covariance functions $\CC_0$ and $\CC_1$.

\begin{algorithm}[t]
\SetAlgoLined
\KwData{ the integral observations $\BZ^{\TI}_n$ at the $n$-th I-sensor, the number of generated samples $J$, the warping functions $W_i$ and the covariance functions $\CC_i(\cdot,\cdot)$ for $i=0,1$, the summary statistics $S:\R^K\to \R^l$ for some $l\in \N$ with $l<K$, the distance measure $d(\cdot,\cdot): \R^l \times \R^l \to \R$, the noise variance $\sigma_{\TI}^2>0$, the error tolerance $\delta^{\TI}>0$,  some $\epsilon^{\TI}>0$ small enough, and the test threshold $\gamma^{\TI}$.}
\KwResult{the approximated test statistic $\widehat \Lambda(\BZ_n^{\TI})$ and the decision $\widehat y_n^{\TI}$ at $\Bx_N^{\TI}$.}

\textbf{Offline phase:} For $i=0,1$ and $j=1,\cdots,J$, generate a sample $\widehat \BZ_j^{\TI,i}$ according to \ref{Temporal: IntegralObservation}.\
 
\textbf{Online phase:} For $i=0,1$ and $j=1,\cdots,J$, accept $\widehat \BZ_j^{\TI,i}$ if 
 \[
 d(S(\BZ_n^{\TI}),S(\widehat \BZ_j^{\TI,i}))\leq \delta^{\TI}.
 \]

For $i=0,1$,  denote the number of accepted samples as $n_{\CH_i}$. Compute the approximated test statistic 
 \[\widehat\Lambda(\BZ_n^{\TI})= \frac{  n_{\CH_0}+\epsilon^{\TI}}{  n_{\CH_1}+\epsilon^{\TI}}.
 \] \label{Temporal:epsilon}
 The decision $\widehat y_n^{\TI}$ is given by
 \[
 \widehat y_n^{\TI}=\begin{cases}
1,\quad \text{if } \widehat\Lambda(\BZ_n^{\TI})<\gamma^{\TI}\\
0, \quad \text{if } \widehat\Lambda(\BZ_n^{\TI})\geq \gamma^{\TI}.
    \end{cases}
 \]  
\caption{Neighborhood-density-based Likelihood Ratio Test (NLRT)}
\label{Temporal:NLRT}
\end{algorithm}
\begin{remark}
The Laplace Approximation technique does not work for the integral case as the marginal likelihood $p(\BZ^{\TI}_n|\CH_i)$ for $i=0,1$ cannot be expressed analytically.
\end{remark}
\begin{remark}
A small positive number $\epsilon^{\TI}$ is added in Line \ref{Temporal:epsilon} of Algorithm \ref{Temporal:NLRT} when evaluating the approximated test statistic to avoid division by zero.
\end{remark}
Similar to WGPLRT, since the samples $\big\{\widehat \BZ_j^{\TI,i}\big\}_{j=1:J,\,i=0,1}$ can be generated before making the observations $\BZ_n^{\TI}$, we can also divide the test procedure of NLRT into the \textbf{offline phase} and the \textbf{online phase}.
In the offline phase, the samples $\big\{\widehat \BZ_j^{\TI,i}\big\}_{j=1:J,\,i=0,1}$ are generated, and their summary statistics $\big\{S(\widehat \BZ_j^{\TI,i})\big\}_{j=1:J,\,i=0,1}$ are stored.
Subsequently, in the online phase, after making the observations $\BZ_n^{\TI}$, the test statistic $\widehat\Lambda(\BZ_n^{\TI})$ is computed with respect to the summary statistics that are pre-computed in the offline phase.

\section{Step 2: Spatial Best Linear Unbiased Estimator (S-BLUE)}\label{Sec:S-BLUE}

After all $\TP$-sensors compute the binary decisions $\{\widehat{y}^\TP_j\}_{i=1}^{N^\TP}$ by Algorithm~\ref{Temporal:WGPLRT} and after all $\TI$-sensors compute the binary decisions $\{\widehat{y}^\TI_j\}_{j=1}^{N^{\TI}}$ by Algorithm~\ref{Temporal:NLRT}, let $\widehat{\BY}_{1:N}$ denote the collection of binary decisions $(\widehat y^{\TP}_1,\cdots,\widehat y^{\TP}_{N^{\TP}},\widehat y^{\TI}_1,\cdots,\widehat y^{\TI}_{N^{\TI}})^\TRANSP$ at all N sensors. We aim to derive the \textbf{Spatial Best Linear Unbiased Estimator (S-BLUE)} $\widehat g_*$ for $g_*:=g(\Bx_*)$ at a fixed un-monitored spatial location $\Bx_* \in \CX$. Then, the prediction of $y_*:=y(\Bx_*)$ is given by 
\begin{equation}\label{Spatial:prediction}
    \widehat y_*= \INDI_{\{ \widehat g_*\geq c\}}.
\end{equation}
Let $l:\R\times \R \to \R_+$ denote the loss function. Let $\CR[h]$ denote the Bayes risk of any $h:\R^N\to \R$, that is
\begin{equation}\label{SBLUE:BayesRisk}
    \CR[h]=\EXP[l(h(\widehat \BY_{1:N}),g_*)].
\end{equation}

We restrain the estimator to be a member of the family of linear estimators, $\CH:=\{{h:\R^N \to \R}: h(\widehat \BY_{1:N})=\Bw^\TRANSP\widehat \BY_{1:N}+b, \Bw \in \R^N, b\in\R \}$, where $\Bw$ is called a weight vector, $b$ is called an intercept, and neither is a function of $\widehat \BY_{1:N}$. The S-BLUE is defined to be the optimal linear estimator minimizing the Bayes risk under the quadratic loss function defined by
\begin{equation}
\begin{split}
    \widehat h_{\mathrm{S\mbox{-}BLUE}}&:=\argmin_{h\in\CH} \EXP[(h(\widehat \BY_{1:N})-g_*)^2].
\end{split}
\end{equation}

Let us re-index $(\Bx^{\TP}_{1},\cdots,\Bx^{\TP}_{N^{\TP}},\Bx^{\TI}_{1},\cdots,\Bx^{\TI}_{N^{\TI}})$ as  $(\Bx_1,\cdots,\Bx_N)$ and re-index $\widehat \BY_{1:N}$ as \linebreak$(\widehat y_1,\cdots, \widehat y_N)^\TRANSP$. Observe that the binary decision $\widehat y_n$ at location $\Bx_n$ from the Likelihood Ratio Test is an estimator of the ground-truth $y_n$ with type $\rom{1}$ and type $\rom{2}$ errors. Therefore, the effect of the Likelihood Ratio Test at each $\Bx_n$ can be treated as applying a transition matrix that adds noise to the true label during data transmission, as detailed in the following Remark. 
\begin{remark}\label{SBLUE:noisy channel}
For $n=1,\cdots,N$, let $\gamma^n\in \{\gamma^{\TP},\gamma^{\TI}\}$ denote the test threshold for the Likelihood Ratio Test (either WGPLRT or NLRT) at $\Bx_n$, then the effect of the Likelihood Ratio Test at $\Bx_n$ is equivalent to transmitting the ground-truth $y(\Bx_n)$ via a noisy channel, where the transition matrix $U_n$ is given by
\begin{equation}\label{LRT:noisy channel}
U_n=\begin{blockarray}{ccc}
&\scriptstyle 0 & \scriptstyle1  \\
\begin{block}{c(cc)}
\scriptstyle0&p^n_{00}&p^n_{01} \\
\scriptstyle1&p^n_{10} & p^n_{11}\\
\end{block}
\end{blockarray}\,,
\end{equation}
where $p^n_{01}:=\PROB[\widehat \Lambda < \gamma^n|\CH_0]$ is the type $\rom{1}$ error rate, $p^n_{00}:=1-p^n_{01}$, $p^n_{10}:=\PROB[\widehat \Lambda \geq \gamma^n|\CH_1]$ is the type $\rom{2}$ error rate, and $p^n_{11}:=1-p^n_{01}$.
\end{remark}

Therefore, our spatial field reconstruction procedure can now be partitioned into two steps: in the first step, we perform the approximated Likelihood Ratio Tests (LRTs) to compute the binary decisions; and in the the second step, we reconstruct the binary spatial field via S-BLUE with these noisy binary inputs. The S-BLUE is presented in the following theorem.

\begin{theorem}[Spatial BLUE]\label{SBLUE:Noisy_diff_U}Given the binary decisions $\widehat \BY_{1:N}=(\widehat y_1,\cdots,\widehat y_N)^\TRANSP$ obtained in Algorithm \ref{Temporal:WGPLRT} \& \ref{Temporal:NLRT} and the transition matrices $\{U_n\}_{n=1}^N$ specified in \eqref{LRT:noisy channel},
the S-BLUE of $g_*=g(\Bx_*)$ for a fixed un-monitored location $\Bx_*\in\CX$ is given by
\begin{equation}
\label{SBLUE:Noisy_estimator}
   \widehat g_*:= \widehat h_{\mathrm{S\mbox{-}BLUE}}(\widehat \BY_{1:N})=\mu_*+\COV[g_*, \widehat \BY_{1:N}]\COV[\widehat \BY_{1:N}]^{-1}(\widehat \BY_{1:N}-\EXP[\widehat \BY_{1:N}]),
\end{equation}
where $\mu_*:=\EXP[g_*]$, and for $i,j=1,\cdots,N$,
\begin{align}\label{SBLUE:cov}
\begin{split}
\EXP[\widehat y_i]&=p^i_{11}\Phi\left(-\frac{c-\mu_i}{\sigma_i}\right)+p^i_{01}\Phi\left(\frac{c-\mu_i}{\sigma_i}\right),\\
    \COV[\widehat y_i,\widehat y_j]&=p^i_{01}p^j_{01}\PROB(g_i<c,g_j<c)+p^i_{01}p^j_{11}\PROB(g_i<c,g_j\geq c)\\
    &\phantom{=}+p^i_{11}p^j_{01}\PROB(g_i\geq c,g_j<c)+p^i_{11}p^j_{11}\PROB(g_i\geq c,g_j\geq c)\\
    &\phantom{=}-\left [p^i_{11}\Phi\left(-\frac{c-\mu_i}{\sigma_i}\right)+p^i_{01}\Phi\left(\frac{c-\mu_i}{\sigma_i}\right)\right] \left [p^j_{11}\Phi\left(-\frac{c-\mu_j}{\sigma_j}\right)+p^j_{01}\Phi\left(\frac{c-\mu_j}{\sigma_j}\right)\right],\\
    \COV[g_*,\widehat y_i]&=\frac{1}{\sqrt{2\pi}\sigma_i}(p^i_{11}-p^i_{01})\CC(\Bx_*,\Bx_i)\exp\left(-\frac{(c-\mu_i)^2}{2\sigma_i^2}\right).\\
\end{split}
\end{align}
with $g_i:=g(\Bx_i)$, $\mu_i:=\mu(\Bx_i)$, $\sigma_i^2:=\CC(\Bx_i,\Bx_i)$, and $c$ being the constant threshold in \eqref{Spatial:definition}.
\end{theorem}
\begin{proof}
See \ref{NoisySBLUE_proof}.
\end{proof}

\begin{remark}\label{SBLUE:individual-terms}
For $i,j=1,\cdots, N$, $\PROB(g_i<c,g_j<c)$, $\PROB(g_i<c,g_j\geq c)$,  $\PROB(g_i\geq c,g_j<c)$, $\PROB(g_i\geq c,g_j \geq c)$ are the probabilities of the bivariate normal $ \CN\left (\begin{pmatrix}
\mu_i \\ \mu_j
\end{pmatrix},
\begin{pmatrix}
\sigma_i^2&\CC(\Bx_i,\Bx_j)\\
\CC(\Bx_i,\Bx_j)&\sigma_j^2
\end{pmatrix}\right)$ over the regions $(-\infty,c)\times(-\infty,c)$, $(-\infty,c)\times[c,\infty)$, $[c,\infty)\times(-\infty,c)$, $[c,\infty)\times[c,\infty)$, respectively.  
\end{remark}

As a consequence,  the theoretical Bayes risk of S-BLUE can be computed analytically as shown in the following Corollary.
\begin{corollary}[Bayes risk]
\label{SBLUE: bayes risk}
Under the quadratic loss function, the Bayes risk \eqref{SBLUE:BayesRisk} associated with $\widehat h_{\mathrm{S\mbox{-}BLUE}}(\widehat \BY_{1:N})$ defined in \eqref{SBLUE:Noisy_estimator} is given by
\begin{equation}\label{SBLUE:TheoreticalRisk}
    \CR[\widehat h_{\mathrm{S\mbox{-}BLUE}}(\widehat \BY_{1:N})]=\COV[g_*]-\COV[g_*,\widehat \BY_{1:N}]\COV[\widehat \BY_{1:N}]^{-1}\COV[g_*,\widehat \BY_{1:N}]^\TRANSP.
\end{equation}
\end{corollary}

\begin{proof}
See \ref{proof: bayes risk}.
\end{proof}

We can execute the S-BLUE algorithm in two phases: the \textbf{offline phase} and the \textbf{online phase}.
In the offline phase, we compute and store the values of $\mu_*$, $\EXP[\widehat \BY_{1:N}]$, \linebreak$\COV[g_*, \widehat \BY_{1:N}]\COV[\widehat \BY_{1:N}]^{-1}$, and $\CR[\widehat h_\text{S-BLUE}(\widehat \BY_{1:N})]$.
Subsequently, in the online phase, given the decisions $\widehat \BY_{1:N}$,  we compute the estimator $\widehat g_*$ using (\ref{SBLUE:Noisy_estimator}).
This is presented in Algorithm~\ref{SBLUE:algo}.

After formulating the S-BLUE, the overall algorithm is outlined in Algorithm \ref{Overall:algo}.
\begin{algorithm}[t] 
\SetAlgoLined
\KwData{ the binary decisions $\widehat \BY_{1:N}$,  the transition matrices $\{U_n\}_{n=1}^N$, the constant threshold $c$, the mean function $\mu(\cdot)$, the covariance function $\CC(\cdot,\cdot)$, the un-monitored location $\Bx_*$. }
\KwResult{the prediction $\widehat y_*$ at $\Bx_*$ and the Bayes risk $\CR[\widehat h_\text{S-BLUE}(\widehat \BY_{1:N})]$.}

 \textbf{Offline phase:}\
 Compute $\mu_*$, $\COV[g_*,\widehat \BY_{1:N}]$, $\COV[\widehat \BY_{1:N}]$, $\EXP[\widehat \BY_{1:N}]$ using (\ref{SBLUE:cov}).\ \linebreak
Compute $ \CR[\widehat h_\text{S-BLUE}(\widehat \BY_{1:N})]$ using (\ref{SBLUE:TheoreticalRisk}).\
 
 \textbf{Online phase:}\
 After collecting the decisions $\widehat \BY_{1:N}$, compute
 $\widehat{g}_*= \widehat h_\text{S-BLUE}(\widehat \BY_{1:N})=\mu_*+\COV[g_*, \widehat \BY_{1:N}]\COV[\widehat \BY_{1:N}]^{-1}(\widehat \BY_{1:N}-\EXP[\widehat \BY_{1:N}])$.\
 
 Compute $\widehat y_*= \INDI_{\{ \widehat g_*\geq c\}}$.\
 
 \caption{Classification of binary spatial random field}
\label{SBLUE:algo}
\end{algorithm}

\begin{algorithm}[t]
\caption{The Overall Algorithm}
\label{Overall:algo}
\SetAlgoLined
\KwData{ the observations $( \BZ^{\TP}_{1:N^{\TP}},\BZ^{\TI}_{1:N^{\TI}})$, the time points $T_{1:M}^{\TP}$, the time intervals $T_{1:K}^{\TI}$, the test threshold $\gamma^{j}$ for the Likelihood Ratio Test and the noise variance $\sigma_j^2>0$ for $j\in \{\TI,\TP\}$, the covariance function $\CC_i(\cdot,\cdot)$ and the warping function $W_i$ of the temporal processes for $i=0,1$, the summary statistics $S(\cdot)$, the distance measure $d(\cdot,\cdot)$, the number of generated samples $J$,  the error tolerance $\delta^{\TI}>0$, some $\epsilon^{\TI}>0$ small enough,  the mean function $\mu(\cdot)$ and the covariance function $\CC(\cdot,\cdot)$ of the binary spatial field, the transition matrices $\{U_n\}_{n=1}^N$, the constant threshold $c$, the un-monitored location $\Bx_*$.}
\KwResult{the prediction $\widehat y_*$ at $\Bx_*$ and the Bayes risk $\CR[\widehat h_\text{S-BLUE}(\widehat \BY_{1:N})]$.}

 Compute $K_i=\CC_i(T_{1:M}^{\TP},T_{1:M}^{\TP})$ by definition for $i=0,1$.\
 
 For $j=1,\cdots,N^{\TP}$, compute $\widehat y_j^{\TP}$ using Algorithm \ref{Temporal:WGPLRT}.\
 
 For $j=1,\cdots,N^{\TI}$, compute $\widehat y_j^{\TI}$ using Algorithm \ref{Temporal:NLRT}.\
  
  Compute $\widehat y_*$ and $\CR[\widehat h_\text{S-BLUE}(\widehat \BY_{1:N})]$ using Algorithm \ref{SBLUE:algo}.\
\end{algorithm}

\begin{remark}[Computational cost of Algorithm~\ref{Overall:algo}]\label{Remark: computational cost}
Let us denote the computational cost of the optimization in the Laplace approximation in Line~\ref{Temporal:WGPLRT-offline} of WGPLRT as $\CT_{\text{opt}}$. 
Moreover, let us denote the computational cost of generating a sample of integral observations, the computational cost of the summary statistics of each sample of integral observations, and the computational cost of each pairwise distance $d(\cdot,\,\cdot)$ in NLRT as $\CT_{\text{samp}}$, $\CT_{\text{summ}}$, and $\CT_{\text{dist}}$, respectively.  
Recall that $N$ denotes the total number of sensors, $M$ denotes the number of point observations at each P-sensor, and $J$ denotes the number of generated samples in NLRT. 
Then, the computational costs incurred at each P-sensor, each I-sensor, and the FC in Algorithm~\ref{Overall:algo} are given as follows.
\begin{itemize}
\item At each P-sensor, the \textbf{offline phase} costs $O(\CT_{\text{opt}}+M^3)$ and
the \textbf{online phase} costs $O(M^2)$ for each time-series of point observations.
\item At each I-sensor, the \textbf{offline phase} costs $O(J(\CT_{\text{samp}}+\CT_{\text{summ}}))$ and the \textbf{online phase} costs $O(\CT_{\text{summ}}+J\CT_{\text{dist}})$ for each time-series of integral observations.
\item At the FC, the \textbf{offline phase} costs $O(N^3)$ and the \textbf{online phase} costs $O(N)$ for each set of binary decisions from the P-sensors and I-sensors.
\end{itemize}
See \ref{proof: computational complexity} for the detailed analyses of the computational cost.

\end{remark}

\section{Experiments with Synthetic Data}\label{Sec:Experiments Synthetic}
We conduct two experiments with synthetic data to study the performance of the proposed method in Algorithm \ref{Overall:algo}. Section \ref{Section:synthetic overall} demonstrates the ability of our proposed method to reconstruct the binary spatial random field. In Section \ref{Section:synthetic sensitivity}, we perform a sensitivity analysis on the two LRT algorithms when the noise variance and the number of observations vary. To simplify the nomenclature, from now on, we say a process is warped by a random variable when the warping function is of the form $W=F^{-1}\circ \Phi$, where $F$ is the CDF of that random variable.

The proposed algorithm is compared to the $k$-Nearest Neighbors (KNN) algorithm.
The KNN algorithm takes the binary decisions $\widehat \BY_{1:N}$ from the sensors and assigns an un-monitored spatial location $\Bx_*$ to a class, i.e., $y(\Bx_*)=1$ or $y(\Bx_*)=0$, based on a plurality vote among the $k$ nearest neighbors (i.e., sensors) of $\Bx_*$. 
If the value of $k$ and the distance metric are fixed beforehand (as opposed to, e.g., using a cross-validation scheme to determine them), then the $k$ nearest neighbors of $\Bx_*$ can be determined in the \textbf{offline phase} and subsequently the decision $y(\Bx_*)$ is computed in the \textbf{online phase} once the sensors transmit their binary decisions.
Hence, assuming that the computational cost of evaluating each pairwise distance is $O(1)$, the computational cost of KNN is given by $O(N+k\log N)$ in the \textbf{offline phase} and $O(k)$ in the \textbf{online phase} for each set of binary decisions from the P-sensors and the I-sensors; see \ref{proof: computational complexity} for the detailed analysis.
Despite that the computational cost of KNN in the offline phase is lower than the computational cost $O(N^3)$ of S-BLUE in the offline phase (see Remark~\ref{Remark: computational cost}), we would like to remark that KNN is a heuristic method that does not use the correlation structure of the spatial random field $y(\cdot)$.
Moreover, as discussed in Remark~\ref{Remark: computational cost}, the online phase of S-BLUE costs only $O(N)$, which makes it highly efficient.

Besides the KNN algorithm, we also compare the proposed algorithm to an \textit{unattainable} benchmark called the \textit{oracle}, which is able to access the values of the latent Gaussian Process, i.e., $g(X_{1:N})$, that are assumed to be unobservable. 
The oracle reconstructs the latent Gaussian Process at an un-monitored spatial location $\Bx_*$ using Gaussian Process regression (see\textit{,} e.g., \cite[Section~2.2]{C.E.Rasmussen2006}), that is, it estimates $g(\Bx_*)$ by
\[
\widehat g_{\text{oracle}}:= \mu(\Bx_*)+\CC(\Bx_*, X_{1:N})\CC(X_{1:N}, X_{1:N})^{-1}(g(X_{1:N})-\mu(X_{1:N})),
\]
where $\mu(\cdot)$ is the mean function and $\CC(\cdot, \cdot)$ is the covariance function of the latent Gaussian Process.
Subsequently, the oracle reconstructs the spatial field by $\widehat{y}_{\text{oracle}}:=\INDI_{\{\widehat{g}_{\text{oracle}}\geq c\}}$.

\subsection{Experiment 1: Evaluation on Synthetic Datasets -- Spatial Field Reconstruction}\label{Section:synthetic overall}
We first examine the ability of the proposed method to reconstruct the binary spatial random field using synthetic datasets. 
In this experiment, we consider a sensor network deployed over the geographical region $[-5,5] \times [-5,5] \subset \R^2$. 
Within this region, we create a grid of $50\times 50$ evenly spaced spatial locations. 
Out of these $2500$ spatial locations, we sample $250$ locations uniformly at random without replacement to be monitored by sensors. 
These sensor locations are fixed\footnote{We have conducted the same experiments with different configurations of sensor locations and we have found that the locations of the sensors do not significantly impact the experimental results.}
throughout this experiment.
The remaining 2250 un-monitored locations are used to evaluate the performance of our algorithms. 
Subsequently, each realization of the synthetic dataset is independently generated via the following three-step procedure.

\underline{Step~1: generation of the binary spatial random field.}
The spatial random field $g$ is randomly generated from a GP with mean $0$ and the Squared Exponential covariance function $\CC_{l_g,s_g}(\cdot,\cdot)$ (see\textit{,} e.g., \cite[Section~4.2.1]{C.E.Rasmussen2006}), where the length-scale $l_g=1/2$ and the scale $s_g=1$, i.e., 
\begin{equation*}
\CC_{l_g,s_g}(\Bx,\Bx')= s_g^2 \exp\left(-\frac{(\Bx-\Bx')^\TRANSP (\Bx-\Bx')}{2l_g^2}\right).
\end{equation*}
The generated random field $g$ is then warped by a $\mathrm{Bernoulli(\pi)} $, where $\pi=0.5$, i.e., $c=\Phi^{-1}(1-\pi)=0$ (see details in Remark \ref{Remark: spatial field}) to generate the binary spatial field $y$ over the $50\times 50$ spatial locations.

\underline{Step~2: generation of the temporal processes at sensor locations.}
At each sensor location $\Bx$, based on the value $y(\Bx)$ of the binary spatial field, a temporal process $f(\cdot;\Bx)$ is generated according to \ref{Temporal:LatentProcess}, where the corresponding mean functions are equal to zero and the covariance functions $\CC_0, \CC_1$ are the Mat\'ern covariance functions (see\textit{,} e.g., \cite[Section~4.2.1]{C.E.Rasmussen2006}) with $\nu=1/2,5/2$, respectively. That is, for any $t,t'\in[0,T]$, we have
\begin{align}
&\,\CC_0(t,t')= s_f^2\exp\left(-\frac{r}{l_f}\right), \label{Synthetic:hypothese C0}\\
&\,\CC_1(t,t')= s_f^2\left(1+\frac{\sqrt{5}r}{l_f}+\frac{5r^2}{3l_f^2}\right)\exp\left(-\frac{\sqrt{5}r}{l_f}\right), \label{Synthetic:hypothese C1}
\end{align}
where the signal variance $s_f^2=1$, the length-scale $l_f=1$, and $r:=|t-t'|$ is the distance between $t$ and $t'$.

In either case, the temporal process is warped by a random variable following Tukey's $g$-and-$h(g, h, l, s)$ distribution where $g=0.1, h=0.4$, $l=1$, and $s=1$; see\textit{,} e.g., \cite{Zhang2019}. Note that a random variable $Y$ follows $g$-and-$h(g, h, l, s)$ if
\[
Y = l+ s\left(\frac{\exp(gZ)-1}{g}\right)\exp(hZ^2),
\]
where $Z\sim \CN(0,1)$ is a standard normal random variable.

\underline{Step~3: generation of the sensor observations.} Out of the 250 sensor locations, we randomly select half to deploy P-sensors, 
whereas I-sensors are deployed at the other half of the sensor locations, that is, $N^\TP = N^\TI = 125$. 
The P-sensors collect point observations at $M$ time points equally spaced over $[0, 20]$ and the I-sensors collect integral observations over $K$ consecutive time intervals of equal lengths over $[0,20]$, where the point observations of the P-sensors are randomly generated according to \ref{Temporal: PointObservation}, and the integral observations of the I-sensors are randomly generated according to \ref{Temporal: IntegralObservation}.

For each realization of the synthetic dataset, our objective is to reconstruct the values of the binary spatial field $y(\cdot)$ at the remaining $2250$ un-monitored locations from the point and integral observations $( \BZ^{\TP}_{1:N^{\TP}},\BZ^{\TI}_{1:N^{\TI}})$ collected from the sensors.

To study the performance of the proposed algorithms for spatial field reconstruction, we set $M=K=50$ and set the noise standard deviation of Point and Integral sensors to $\sigma_{\TP}=\sigma_{\TI}=0.1$.
For NLRT, we choose the summary statistics to be the autocorrelations (ACF) with lags 1 to 4 (see\textit{,} e.g., \cite[Section~2.1]{box2015}), and we choose the distance measure to be the Euclidean distance. 
Moreover, we set $\delta^{\TI}=0.1$, $\epsilon^{\TI}=0.1$, and we set the number of generated samples to $J=10000$.
Then, the LRT thresholds are set to $\gamma^{\TP}=\exp(176.9204)$ and $\gamma^{\TI}=\exp(-0.2555)$, such that the False Positive Rate (FPR) of each sensor is controlled by the significance level $\alpha=0.1$.  Finally, at each sensor location, we set the transition matrix (see \eqref{LRT:noisy channel}) to either 
\begin{align*}
U^{\TP}=\begin{pmatrix}
  0.9022 & 0.0978 \\
0.1772 & 0.8228
\end{pmatrix}
\quad\text{or}\quad
U^{\TI}=\begin{pmatrix}
    0.9016 & 0.0984 \\
0.0920 &  0.9080
\end{pmatrix}
\end{align*}
depending on the sensor type. The test thresholds $\gamma^{\TP}, \gamma^{\TI}$ and the transition matrices $U^{\TP}, U^{\TI}$ are computed via simulation beforehand, where we estimate the CDFs of the approximated test statistics under $\CH_0$ and $\CH_1$ via Monte Carlo}.

Under the parameters specified above, Figure~(\ref{fig:true field}) and (\ref{fig:reconstructed field}) show the true binary spatial field and the reconstructed binary spatial field of a single realization of the synthetic dataset, respectively. Table~\ref{table:synthetic F1} presents the average performance of our proposed approach and other algorithms based on 100 realizations of the synthetic dataset generated from the aforementioned three-step procedure. The table includes columns for the average mean-square-error (MSE), F1 score\footnote{The F1 score is given by 2TP/(2TP + FP + FN), where TP, FP, and FN stand for the numbers of true positive, false positive, and false negative cases, respectively.}, false positive rate (FPR), and true positive rate (TPR). Each row corresponds to a specific method: the oracle, our proposed method (S-BLUE), S-BLUE using only point observations from the P-sensors, S-BLUE using only integral observations from the I-sensors, and KNN. Overall, we observe that our proposed method outperforms KNN and achieves slightly worse results compared to the oracle. Furthermore, Table~\ref{table:synthetic F1} suggests that utilizing both Point and Integral sensors for reconstruction purposes yields superior results compared to the case where only one single type of sensor is utilized. Finally, Table~\ref{table:synthetic time} shows the average computational time of the algorithms over 100 realizations. We observe that S-BLUE is indeed a highly efficient algorithm.

%\begin{table}[h]
%\centering
%\begin{tabular}{|c |c | c |c|c|} 
% \hline
%  Algorithm & Time  Complexity\\ [0.5ex] 
%  \hline
%   Oracle & 0.1207& \\
% \hline
% S-BLUE& 0& $O()$ \\
% \hline
% \hline 
%  KNN & $O()$ & $O(n)$\\
% \hline
%\end{tabular}
%\caption{Real-world experiment -- Estimated WGP hyperparameters}
%\label{table:complexity}
%\end{table}

To further study the impacts of the noise variance when observing the temporal processes, we set $M=K=50$ and $\sigma_{\TP}=\sigma_{\TI}$ and plot the average MSE, F1 score, FPR, and TPR over $100$ realizations against the noise variance $\sigma^2_{\TP}=\sigma^2_{\TI}$ in Figure~(\ref{fig:sn}). The test thresholds and transition matrices are computed via simulation beforehand for each value of the noise variance. 
Note that these average metrics are computed based on 100 realizations of the synthetic dataset generated via the three-step procedure described above for each value of noise variance. When the noise variance becomes smaller, observe that the FPR stays relatively constant because the significance level is controlled to be around $0.1$. In contrast, notice that the TPR increases, which leads to an improvement in the F1 score and MSE. 

\begin{figure}[htp]
     \centering
     \begin{subfigure}[b]{0.328\textwidth}
         \centering
         \includegraphics[width=1\textwidth]{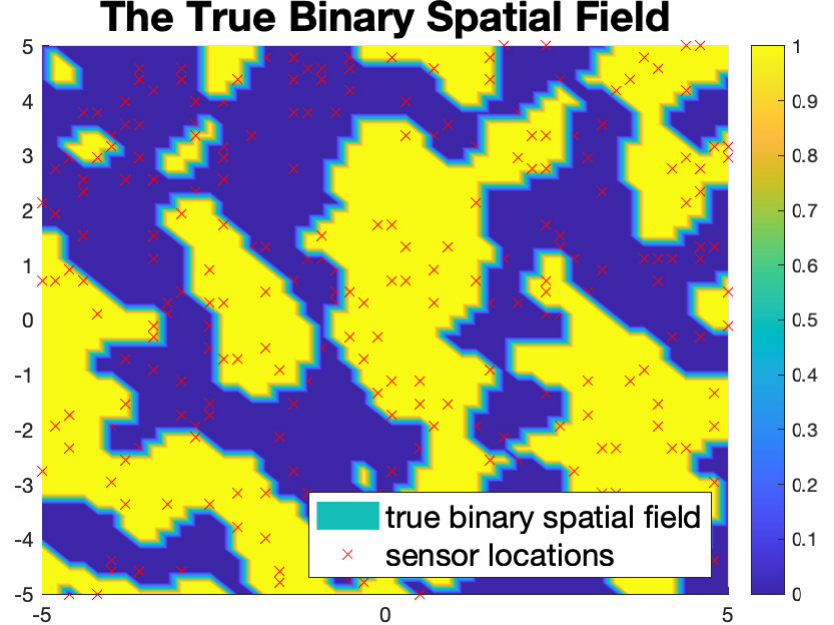}
         \caption{}
         \label{fig:true field}
     \end{subfigure}
     \hfill
     \begin{subfigure}[b]{0.328\textwidth}
         \centering
         \includegraphics[width=1\textwidth]{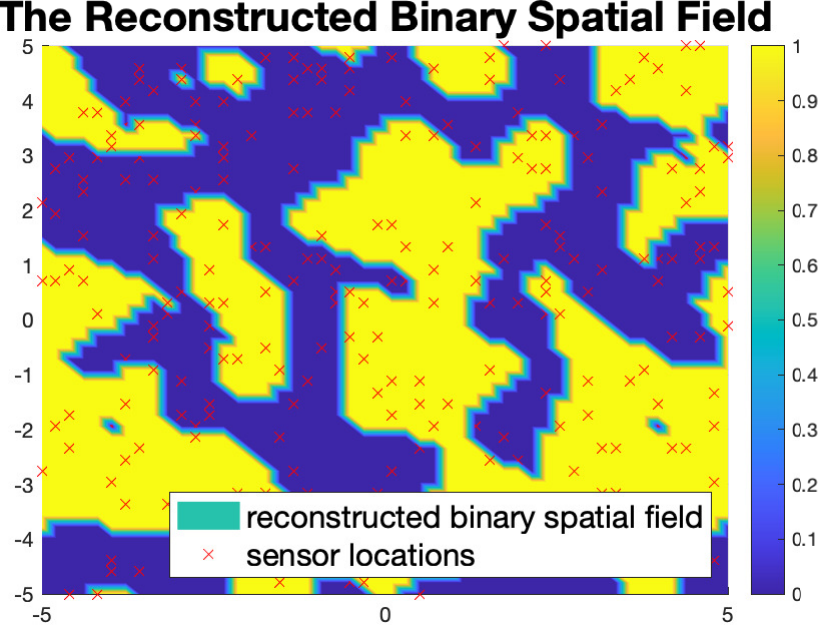}
         \caption{}
         \label{fig:reconstructed field}
     \end{subfigure}
     \hfill
     \begin{subfigure}[b]{0.328\textwidth}
         \centering
         \includegraphics[width=1\textwidth]{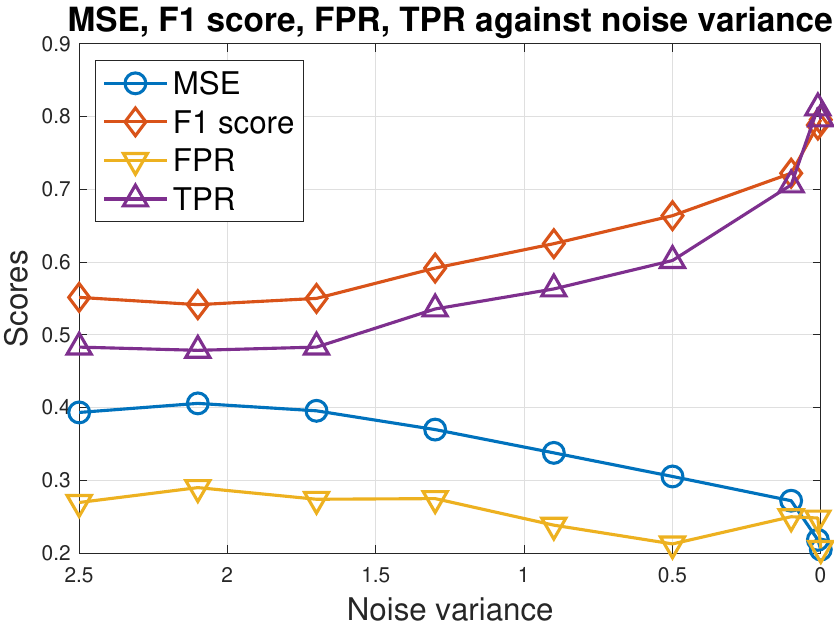}
         \caption{}
         \label{fig:sn}
     \end{subfigure}
        \caption{Experiment 1 -- Figure~(\ref{fig:true field}) represents the true binary spatial field. In the figure, the red crosses indicate the sensor locations, the blue area represents the spatial locations where the values of the true binary spatial random field are ``0", and the yellow area represents the spatial locations where the values are ``1". Figure~(\ref{fig:reconstructed field}) visualizes the reconstructed binary spatial field from our algorithm. Figure~(\ref{fig:sn}) shows the average MSE, F1 score, FPR, and TPR over $100$ realizations against the noise variance.}
        \label{fig:simulated data}
\end{figure}

\begin{table}[h]
\centering
\begin{tabular}{|c |c | c |c|c|c|} 
 \hline
  Algorithm& MSE & F1 score & FPR & TPR \\ [0.5ex] 
  \hline 
  Oracle & 0.1201& 0.8799& 0.1224 & 0.8800 \\
 \hline
 S-BLUE& 0.2578&  0.7412 &0.2595  &  0.7400 \\
 \hline
 S-BLUE (pt. obs.) & 0.3259 & 0.6685& 0.3147&0.6591\\
 \hline
 S-BLUE (int. obs.) & 0.3067&0.6946 &0.3188 & 0.7005\\
 \hline 
  KNN &   0.3073 & 0.6846 &0.2970   &  0.6739  \\
 \hline
\end{tabular}
\caption{Experiment 1 -- Average MSE, F1 score, FPR, TPR over $100$ realizations.}
\label{table:synthetic F1}
\end{table}

\begin{table}[h]
\centering
\begin{tabular}{|c |c | c |c|c|c|} 
 \hline
  Algorithm& Offline phase & Online phase\\ [0.5ex] 
 \hline
% WGPLRT& 3.6450&  3.4504e-04\\
% \hline
% NLRT &21.7846&4.1589\\
% \hline
 S-BLUE&$4.912\times 10^{-1}$&$3.161 \times 10^{-4}$\\
 \hline
 S-BLUE (pt. obs.) & $1.244 \times 10^{-1}$ & $1.716 \times 10^{-4}$\\
 \hline
 S-BLUE (int. obs.) &$1.239\times 10^{-1}$&$1.724 \times 10^{-4}$\\
 \hline 
  KNN &   $3.965\times 10^{-3}$ & $4.818\times 10^{-3}$   \\
 \hline
\end{tabular}
\caption{Experiment 1 -- Average computational time (in seconds) over 100 realizations.}
\label{table:synthetic time}
\end{table}

\subsection{Experiment 2: Evaluation on Synthetic Datasets -- Sensitivity Analyses on LRTs}\label{Section:synthetic sensitivity}
In the second experiment, we perform a sensitivity analysis on the two LRT algorithms by varying  the number of observations as well as the noise variance. We aim to differentiate between the same hypotheses defined in \eqref{Synthetic:hypothese C0} and \eqref{Synthetic:hypothese C1}. 
Except for the number of observations $M$, $K$ and the noise variances $\sigma_{\TP}^2$, $\sigma_{\TI}^2$, all other parameters in the two LRTs are set up as in Subsection~\ref{Section:synthetic overall}. Figure~(\ref{fig:M}) examines the impact of the number of observations on the Receiver Operating Characteristic (ROC) curves of WGPLRT and NLRT. We evaluate the performance using the Area Under Curve (AUC) metric under two different noise standard deviation values: $\sigma_\TP = \sigma_\TI = 0.1$ and $\sigma_\TP = \sigma_\TI = 0.01$. Notably, the AUC of WGPLRT and NLRT increases as the number of observations increases. In addition, we study the effect of the noise variance on both WGPLRT and NLRT in Figure~(\ref{fig:AUCSN}) when the number of observations is again set to be $M=K=50$. Unsurprisingly, the AUC of both WGPLRT and NLRT increases as the noise variance decreases.

\begin{figure}[t]
     \centering
     \begin{subfigure}[b]{0.49\textwidth}
         \centering
         \includegraphics[width=1\textwidth]{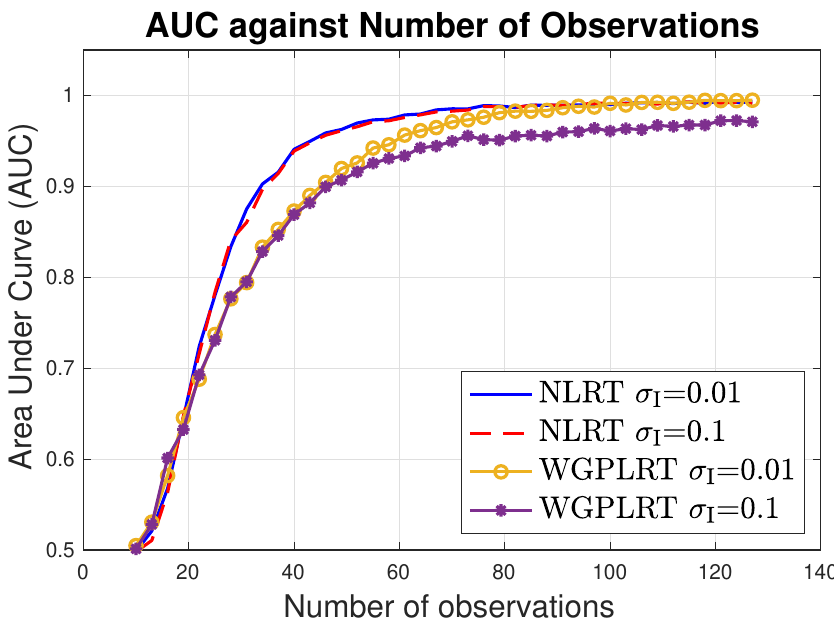}
         \caption{}
         \label{fig:M}
     \end{subfigure}
     \begin{subfigure}[b]{0.49\textwidth}
         \centering
         \includegraphics[width=1\textwidth]{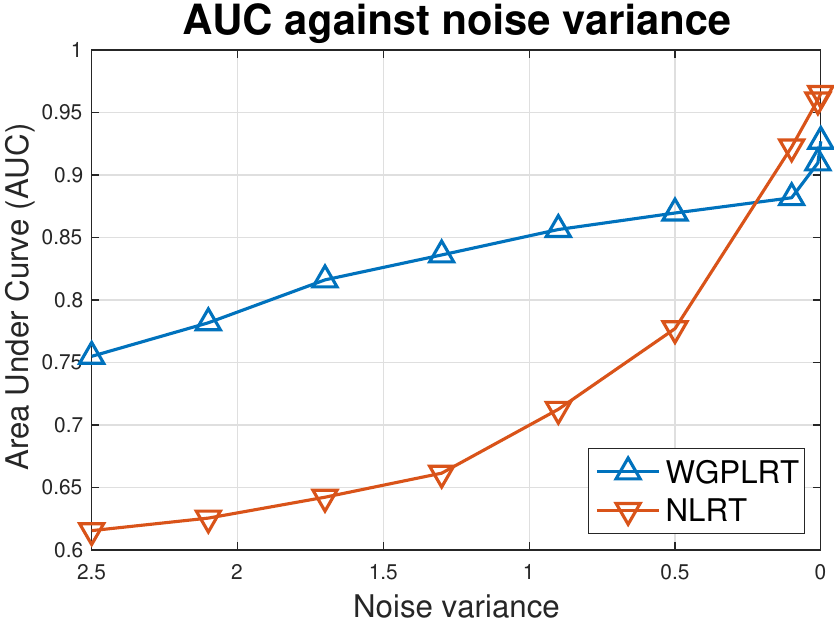}
         \caption{}
         \label{fig:AUCSN}
     \end{subfigure}
        \caption{Experiment 2 -- Figure~(\ref{fig:M}) shows the effect of the number of observations on the AUC of WGPLRT and NLRT. Figure~(\ref{fig:AUCSN}) shows the AUC of WGPLRT and NLRT against the noise variance when $M=K=50$.} 
        \label{fig:Simulated}
\end{figure}

\section{Evaluation on a Real Dataset}\label{Sec:Experiments Semisynthetic}
In the real-world experiments, we study a weather dataset\footnote{\url{https://data.gov.sg/search?groups=environment}, retrieved on 7 January 2022.} from the National Environment Agency (NEA) of Singapore. The dataset contains hourly measurements of Temperature, Wet Bulb Temperature, Dew Point Temperature, Scalar Mean Wind Direction, Relative Humidity, Scalar Mean Wind Speed, and Sea-Level Pressure at $21$ weather stations. The measurements are mainly between 2010 to 2018, yet their starting dates and ending dates vary depending on the stations. Among the $21$ stations, $5$ stations (with station number $\mathrm{S06, S23, S24, S25, S80}$) are installed as early as 2005 and they also take hourly measurements of Total Rainfall and Cloud Cover. The spatial locations of the weather stations are visualized in Figure~(\ref{fig:stns}).

In this section, we aim to synthetically generate sensor observations based on the dataset and reconstruct the binary spatial random field at other spatial locations over Singapore where no weather stations are deployed.

\subsection{Data Preprocessing}
To begin with, we observe that some weather stations only take measurements between 5:00~a.m.\ to 23:00~p.m., hence all measurements outside this period are removed for consistency. Furthermore, to reduce seasonal variation, we take the data over $17$ weeks during the Southwest Monsoon Season\footnote{See details in \url{http://www.weather.gov.sg/climate-climate-of-singapore}.} in 2012, i.e., from 03/06/2012 (Sunday) to 29/09/2012 (Saturday). Next, we select the fields of interest to represent the binary spatial random field and the temporal processes. For the binary spatial random field, we take the average weekly relative humidity as the latent GP and the constant threshold $c$ is taken to be the median of all data. Figure~\ref{fig:spatial} shows the histogram and normal Q-Q plot of the average weekly relative humidity. Observe that the normality assumption holds. For the temporal processes, hourly measurements of temperature are used. For each hour from 5:00~a.m.\ to 23:00~p.m., the average hourly temperature over all weather stations is subtracted from the measurements to center the measurements. The processed measurements are named as \textit{centered temperatures~(CT)}. Then, they are further partitioned into two subsets based on the level of the average weekly relative humidity (i.e., over or below the constant threshold $c$) to represent the temporal processes under $\CH_0$ and $\CH_1$. Figure~\ref{fig:temporal histogram} shows the density plots of the centered temperatures from the two subsets. We notice that the data are left-skewed and the transformed data under $x \mapsto 10 -x$ can be modeled using Gamma distributions for suitable hyperparameters fitted from data; see also Table~\ref{table:temporal hyper}. Figure \ref{fig:temporal QQplot} presents the Q-Q plots of the transformed data from the two subsets against the fitted Gamma distributions, respectively.

\begin{figure}[t]
     	\centering
     \begin{subfigure}[b]{0.49\textwidth}
         \centering
         \includegraphics[width=1\textwidth]{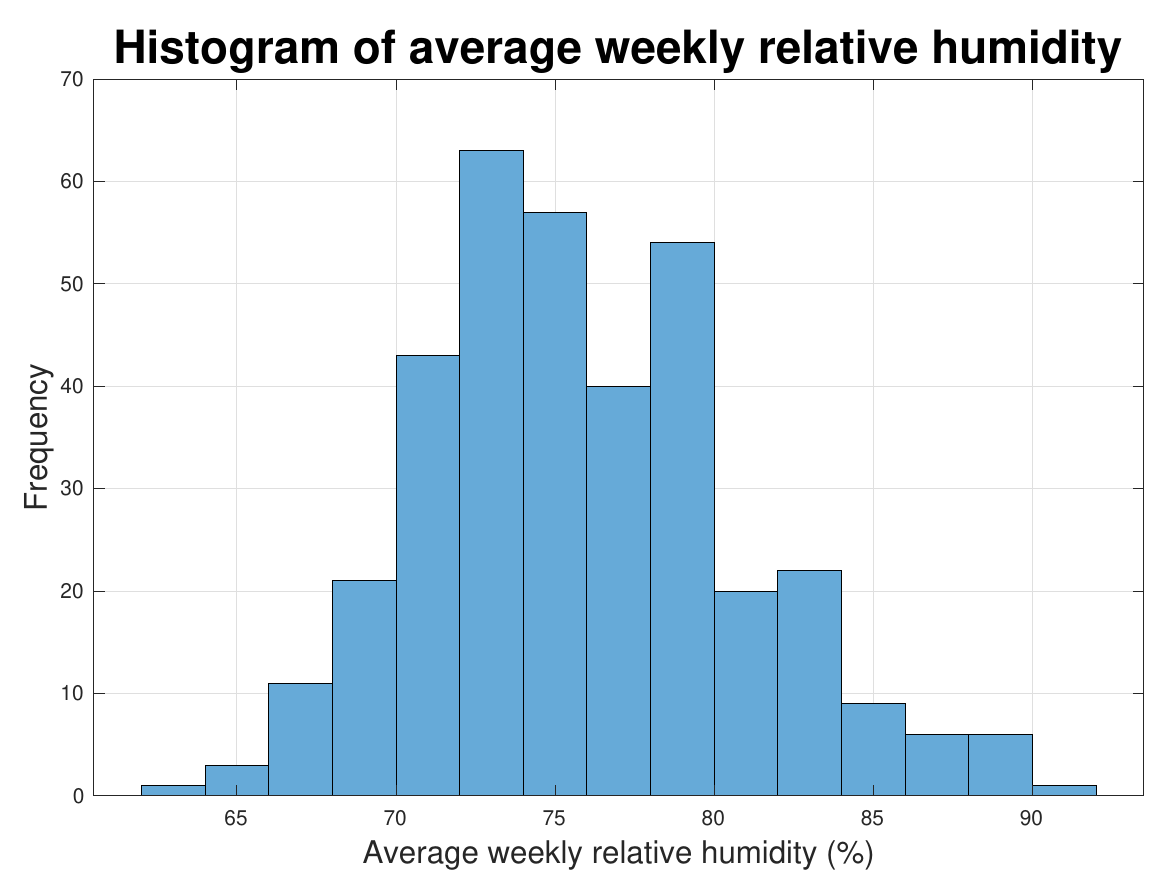}
          \label{fig:spatial QQplot}
     \end{subfigure}
     \hfill
     \begin{subfigure}[b]{0.49\textwidth}
         \centering
         \includegraphics[width=1\textwidth]{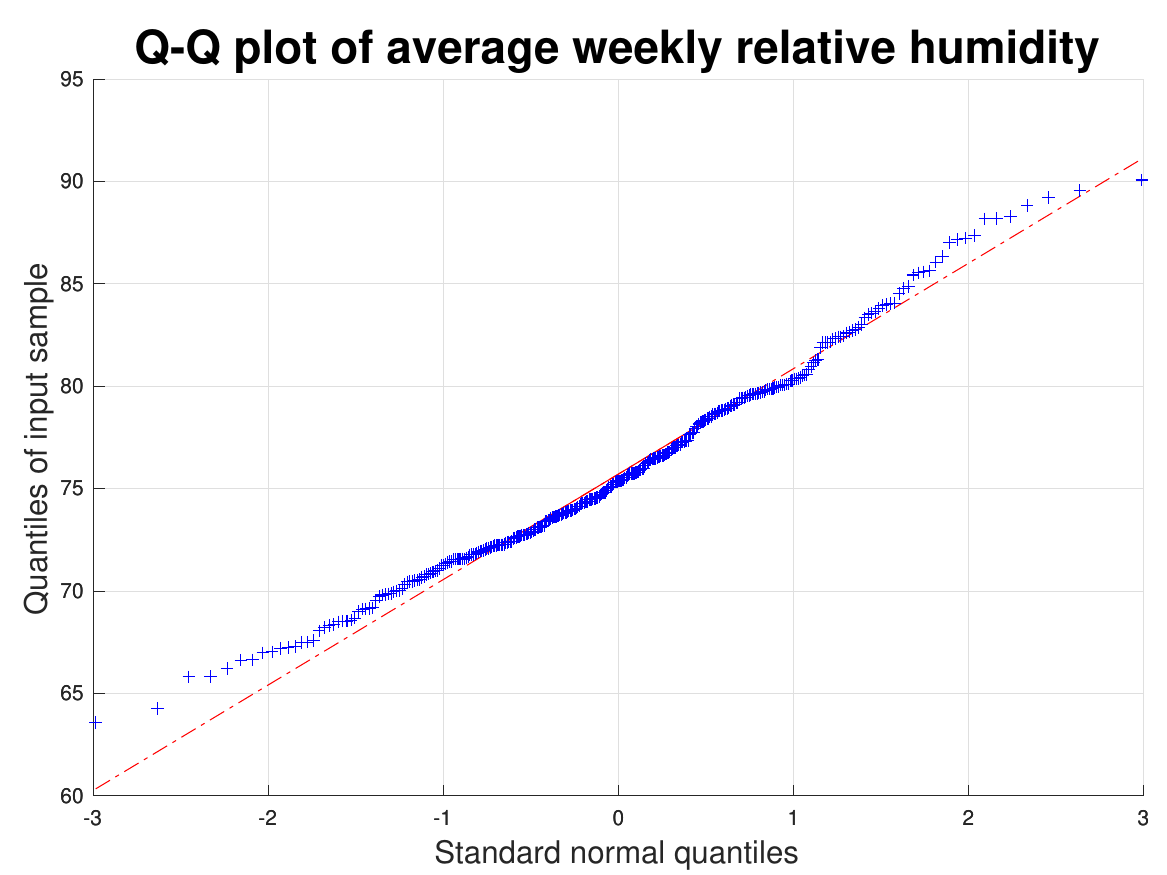}
         \label{fig:spatial histogram}
     \end{subfigure}
        \caption{Real-world experiment -- Histogram and normal Q-Q plot of the average weekly relative humidity.}
        \label{fig:spatial}
\end{figure}

\begin{figure*}[t]  % spans both columns
\begin{subfigure}{0.49\textwidth}
\includegraphics[width=\linewidth]{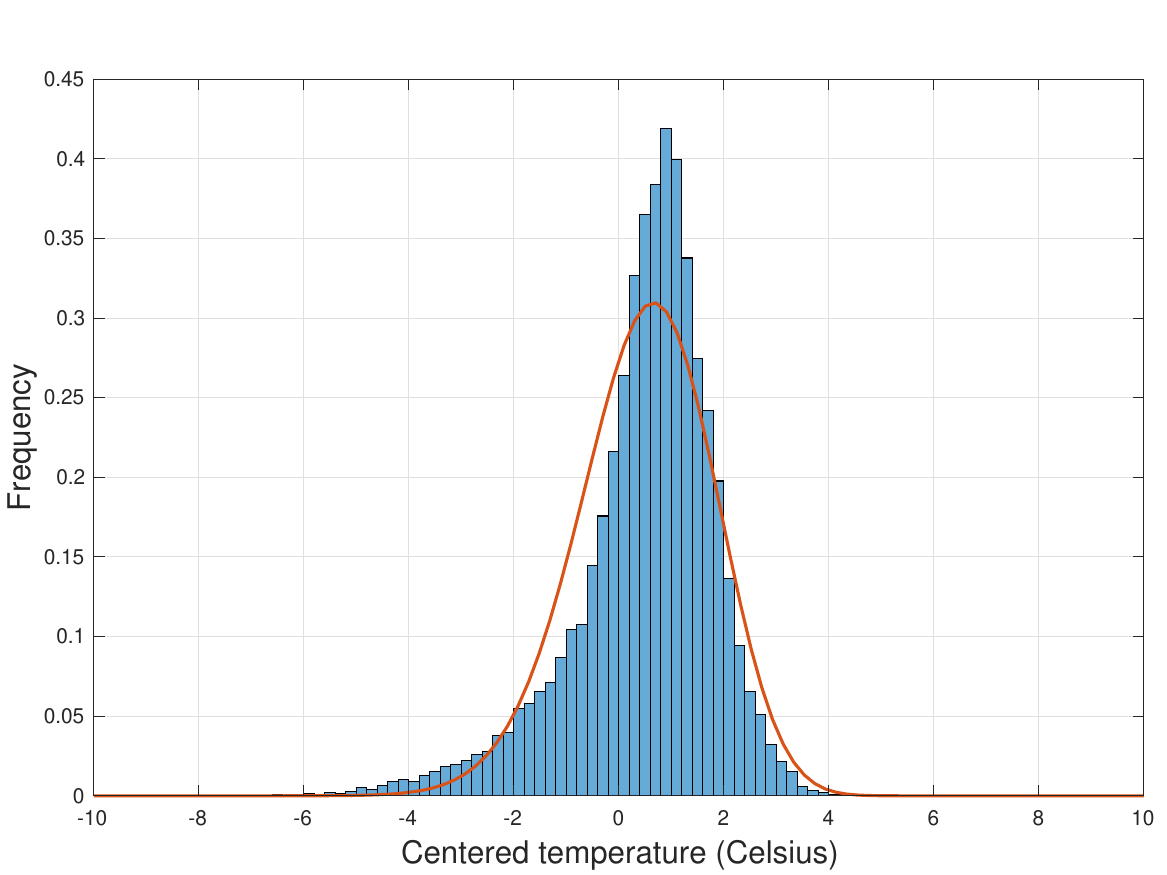}
\caption{}
\label{fig:histogram temperature below}
\end{subfigure}
\hfill
\begin{subfigure}{0.49\textwidth}
\includegraphics[width=\linewidth]{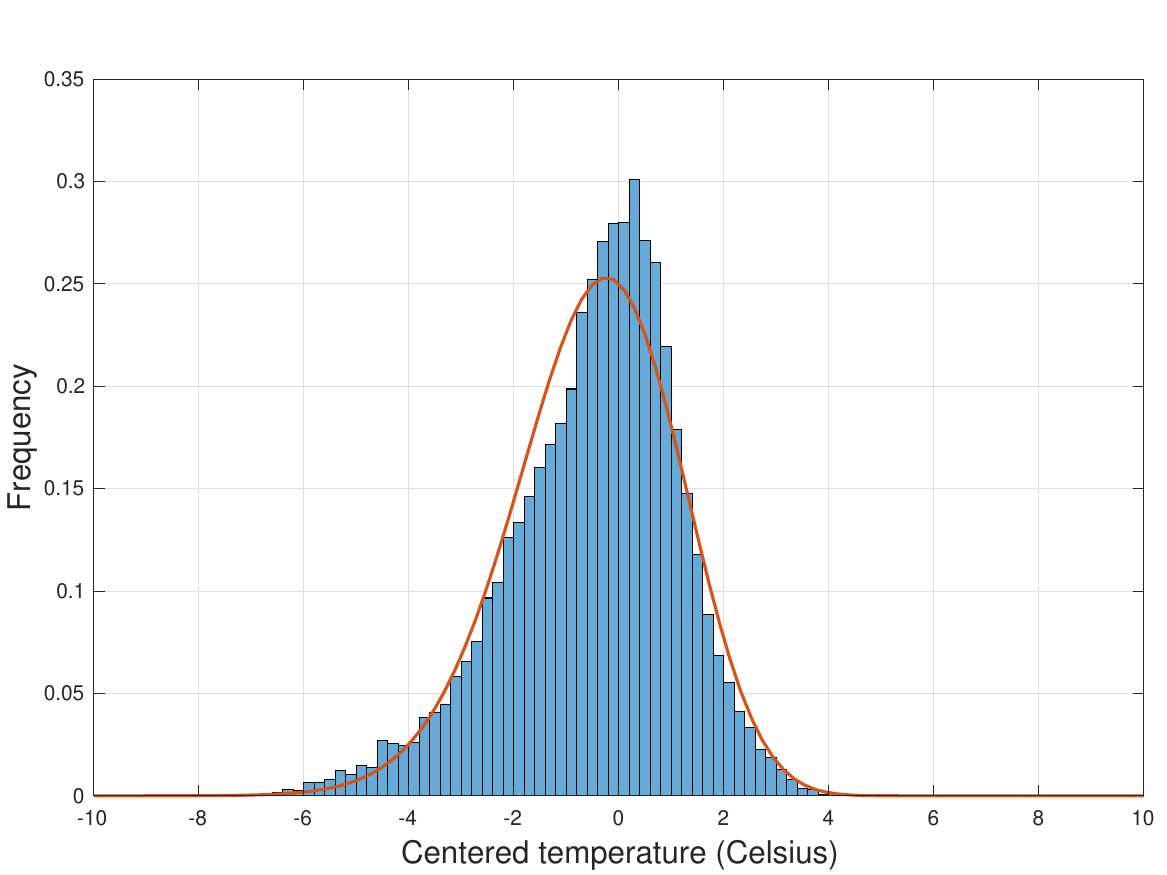}
\caption{}
\label{fig:histogram temperature above}
\end{subfigure}
\caption{Real-world experiment --  Figure~(\ref{fig:histogram temperature below}) shows the density plot of the centered temperatures when the average weekly relative humidity is below the constant threshold $c$. Figure~(\ref{fig:histogram temperature above}) shows the histogram of the centered temperatures when the average weekly relative humidity is above the constant threshold $c$. The orange-colored curves in both Figure~(\ref{fig:histogram temperature below}) and Figure~(\ref{fig:histogram temperature above}) indicate the flipped and shifted graphs of the probability density functions of the fitted Gamma distributions, respectively.} % Overall figure \
\label{fig:temporal histogram}
\end{figure*}

\begin{figure*}  % spans both columns
\begin{subfigure}{0.49\textwidth}
\includegraphics[width=\linewidth]{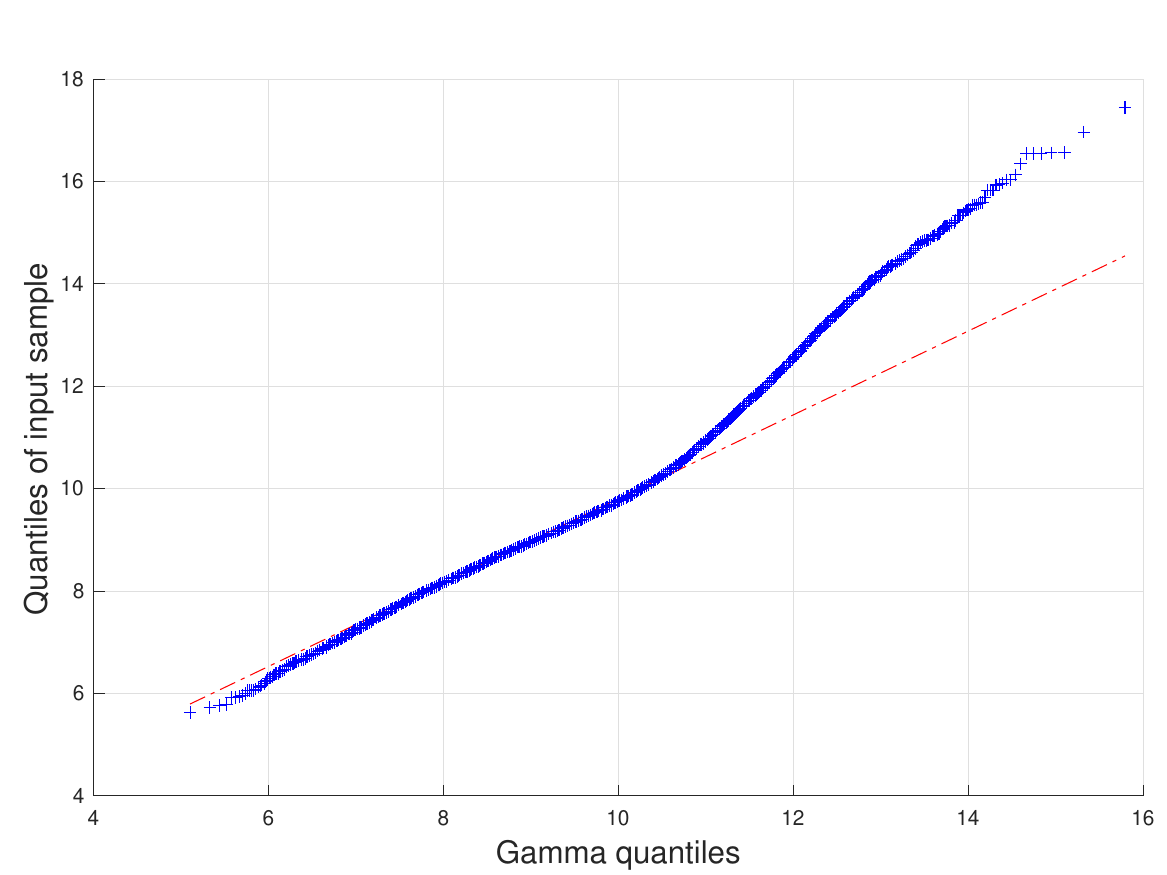}
\caption{}
\label{fig:QQplot gamma temperature below}
\end{subfigure}
\hfill
\begin{subfigure}{0.49\textwidth}
\includegraphics[width=\linewidth]{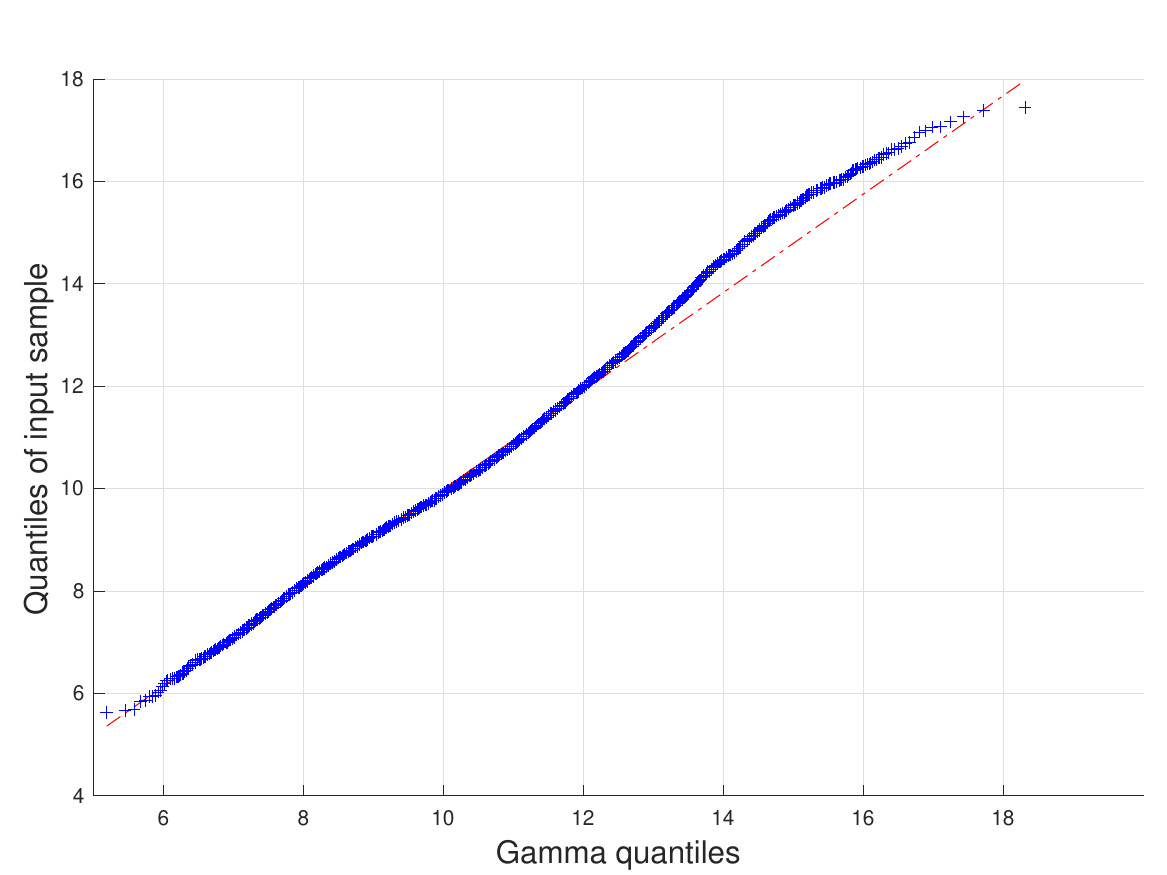}
\caption{}
\label{fig:QQplot gamma temperature above}
\end{subfigure}
\caption{Real-world experiment -- Figure~(\ref{fig:QQplot gamma temperature below}) presents the Q-Q plot of the transformed centered temperatures when the average weekly relative humidity is below the constant threshold $c$ against the fitted Gamma distribution.  Figure~(\ref{fig:QQplot gamma temperature above}) shows the corresponding Q-Q plot when the average weekly relative humidity is above the constant threshold $c$.} 
\label{fig:temporal QQplot}
\end{figure*}

\subsection{Model Selection}\label{Real:MS}
After determining the spatial and temporal fields of interest, we model the measurements of the average weekly relative humidity as noisy samples from a GP with a constant mean function and a Mat\'ern covariance function with $\nu=5/2$ defined in \eqref{Synthetic:hypothese C1} with $r \leftarrow \|x-x'\|$ for $x, x' \in \CX$. For each week from 03/06/2012 to 29/09/2012, we estimate the signal mean, signal standard deviation, noise standard deviation, and length-scale via \textit{maximum a posterior} (MAP) estimation (see Chapter~2 of \cite{C.E.Rasmussen2006} for details) based on the measurements from the $21$ weather stations. The prior distributions of the logarithms of the signal standard deviation, length-scale, and noise standard deviation are $\CN(-2, 0.1)$, $\CN(\log{3.8}, 0.1)$, $\CN(\log{0.1}, 0.01)$, respectively. The median of the estimated values over $17$ weeks are then taken as the final hyperparameters of the GP. The estimated hyperparameters are shown in Table~\ref{table:spatial hyper}. 

A similar approach is also applied to specify the model for the centered temperatures (CT). For each set of measurements, we model the data after the transformation $x\mapsto 10-x$ as noise-free samples from a GP warped by a Gamma random variable. First, the hyperparameters for the Gamma distribution are estimated based on the measurements from the $21$ weather stations. Next, the data are transformed by $z \mapsto \Phi^{-1}\circ F (z)$, where $F$ is the CDF of the fitted Gamma distribution. Then, the transformed data are modeled by a GP in which the mean function is constant zero and the covariance function is again the Mat\'ern covariance function with $\nu=5/2$ defined in \eqref{Synthetic:hypothese C1}. The estimated hyperparameters based on the transformed measurements over all sensors via \textit{maximum marginal likelihood estimation} (see Chapter 5 of \cite{C.E.Rasmussen2006} for details) are shown in Table~\ref{table:temporal hyper}.

\begin{table}[h]
\centering
\begin{tabular}{ |c | c |c|c|c|} 
 \hline
  threshold $c$ & signal mean & signal std. dev & length-scale & noise std. dev\\ [0.5ex] 
 \hline
 75.3692& 75.0566 & 5.3068 & 0.0344 & 0.1000\\
 \hline 
\end{tabular}
\caption{Real-world experiment -- Estimated GP hyperparameters }
\label{table:spatial hyper}
\end{table}

\begin{table}[h]
\centering
\begin{tabular}{|c |c | c |c|c|} 
 \hline
  &$\mathrm{Gamma}(a, b)$& GP mean & GP std. dev & GP length-scale\\ [0.5ex] 
 \hline
$\CH_0$ & $(53.7457, 0.1771)$ & 0 & 1 & 3.7622\\
\hline
$\CH_1$ & $(43.3694, 0.2417)$ & 0 & 1 & 4.0654\\
 \hline 
\end{tabular}
\caption{Real-world experiment -- Estimated WGP hyperparameters}
\label{table:temporal hyper}
\end{table}

\subsection{Experiment Setting}\label{Subsection: Experiment Setting}
In this experiment, we create a $50 \times 50$ grid of spatial locations over Singapore to evaluate the performance of our algorithms. 
Subsequently, we randomly generate each realization of the dataset via the following three-step procedure.

\underline{Step~1: generation of the binary spatial random field.}
The latent spatial random field is a GP with a constant mean function and a Mat\'ern covariance function with $\nu=5/2$, where the hyperparameters are specified in Table~\ref{table:spatial hyper}. The process is then warped by a Bernoulli random variable to generate the binary spatial random field over both the $50 \times 50$ grid of spatial locations and the $21$ weather station locations.

\underline{Step~2: generation of the temporal processes at the weather stations.}  At each weather station, based on the value of the binary spatial field, a temporal process is generated according to \ref{Temporal:LatentProcess}. The corresponding mean functions are equal to zero and the covariance functions are Mat\'ern covariance functions with $\nu=5/2$, where the hyperparamters are specified in Table~\ref{table:temporal hyper}. Then, the GP is warped by a Gamma random variable with hyperparameters presented in Table~\ref{table:temporal hyper}.

\underline{Step~3: generation of the sensor observations.} We equip all $21$ weather stations with P-sensors, i.e., $N^\TP = 21, N^\TI= 0$. Each weather station collects hourly measurements of the temporal process from $5:00$ a.m. to $23:00$ p.m., $7$ days a week. For simplicity, we set $T=133$.

To evaluate the performance of the proposed algorithms in reconstructing the spatial field, we set $M=133$ and $\sigma_{\TP} =0.1$.
Given the point observations, WGPLRT is applied with $\gamma^{\TP}=\exp(-6.5387)$ such that the significance level $\alpha$ is approximately $0.1$, which results in the following transition matrix  (see~\eqref{LRT:noisy channel})
\begin{equation}
U=\begin{pmatrix}
0.8996 & 0.1004 \\
0.0774 & 0.9226
\end{pmatrix}.
\end{equation}
Then, S-BLUE is used to predict the values of the binary spatial field at the $50 \times 50$ grid of spatial locations where no weather station is deployed.

\subsection{Result and Discussion}

Figure~\ref{fig:real} shows the spatial locations of the $21$ weather stations, the true binary spatial random field, and the reconstructed binary spatial field based on the point observations at the $21$ weather stations from a single realization of the dataset. Note that the values of the binary spatial field are not directly observed at the sensors. Instead, they are inferred using the temporal observations at the sensors. Consequently, there are instances, such as at sensor S23, where the value of the reconstructed binary spatial field differs from the true value. This discrepancy arises due to the estimation error that occurs during the first step of our proposed approach.
Figure~\ref{fig:real Latent} presents the true latent GP, the reconstructed latent GP, and the map of the Bayes risk given in \eqref{SBLUE:TheoreticalRisk} from that particular realization of the dataset. 
From Figure~(\ref{fig:real bayes risk}), we see that the Bayes risks are small around the weather stations and high over the regions where no weather stations are deployed. The average performance of the oracle, S-BLUE, and KNN over 100 realizations is presented in Table~\ref{table:real performance}. We see that S-BLUE achieves performance comparable to the oracle, showcasing the effectiveness of our proposed algorithm.

\begin{figure}[htp]
     \centering
     \begin{subfigure}[b]{0.328\textwidth}
         \centering
         \includegraphics[width=1\textwidth]{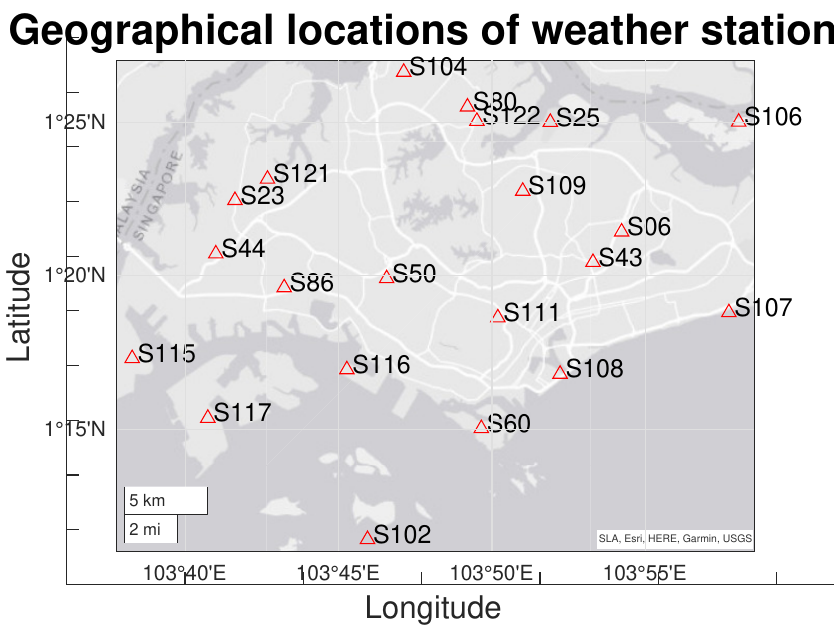}
         \caption{}
         \label{fig:stns}
     \end{subfigure}
     \hfill
     \begin{subfigure}[b]{0.328\textwidth}
         \centering
         \includegraphics[width=1\textwidth]{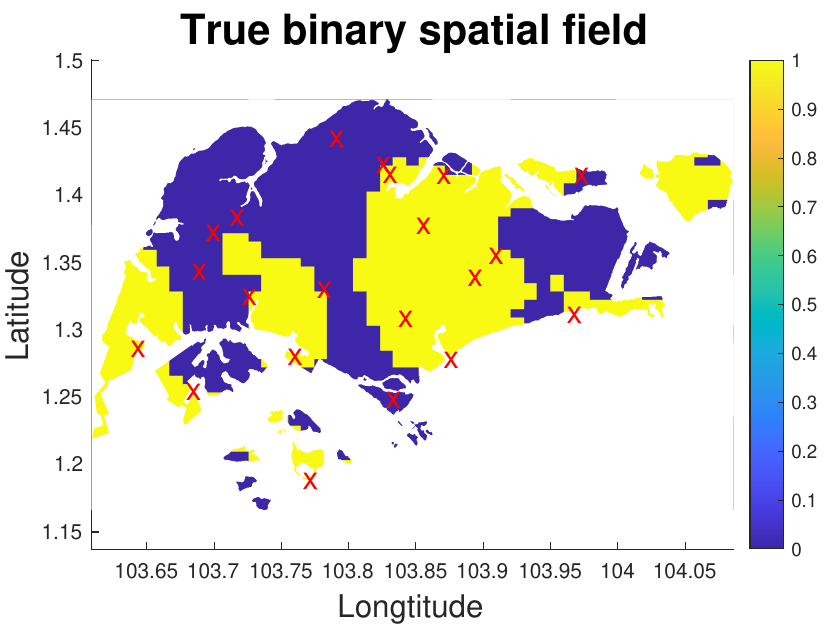}
         \caption{}
         \label{fig:real true field}
     \end{subfigure}
     \hfill
     \begin{subfigure}[b]{0.328\textwidth}
         \centering
         \includegraphics[width=1\textwidth]{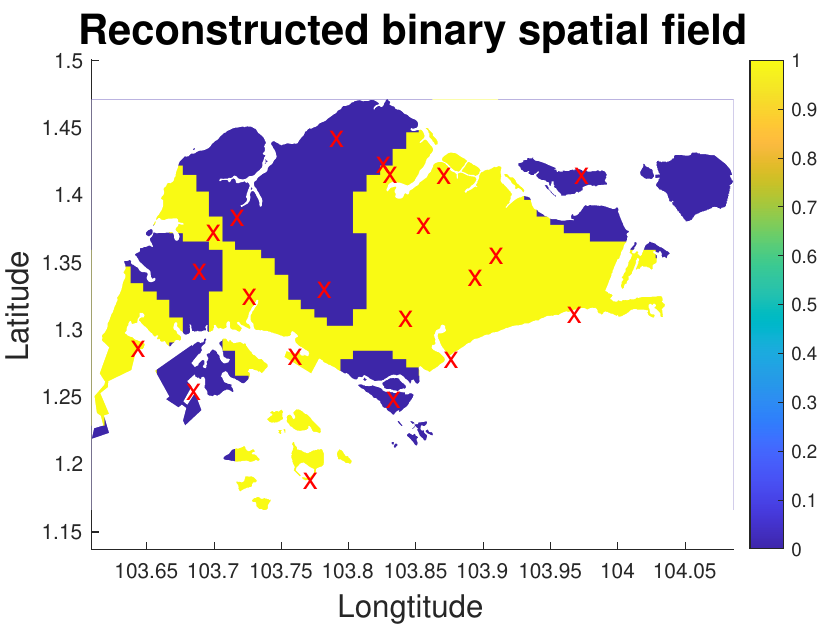}
         \caption{}
         \label{fig:real reconstructed field}
     \end{subfigure}
        \caption{Real-world experiment -- Figure~(\ref{fig:stns}) shows the spatial locations of the $21$ weather stations over Singapore. Figure~(\ref{fig:real true field}) represents the true binary spatial field, where the red crosses denote the weather stations, the blue area represents the spatial locations at which the values of the true binary spatial random field are ``0", and the yellow area represents the spatial locations at which the values are ``1". Figure~(\ref{fig:real reconstructed field}) visualizes the reconstructed binary spatial field from our algorithm based on the point observations at the $21$ weather stations.} 
        \label{fig:real}
\end{figure}

\begin{figure}[h!]
     \centering
     \begin{subfigure}[b]{0.328\textwidth}
         \centering
         \includegraphics[width=1\textwidth]{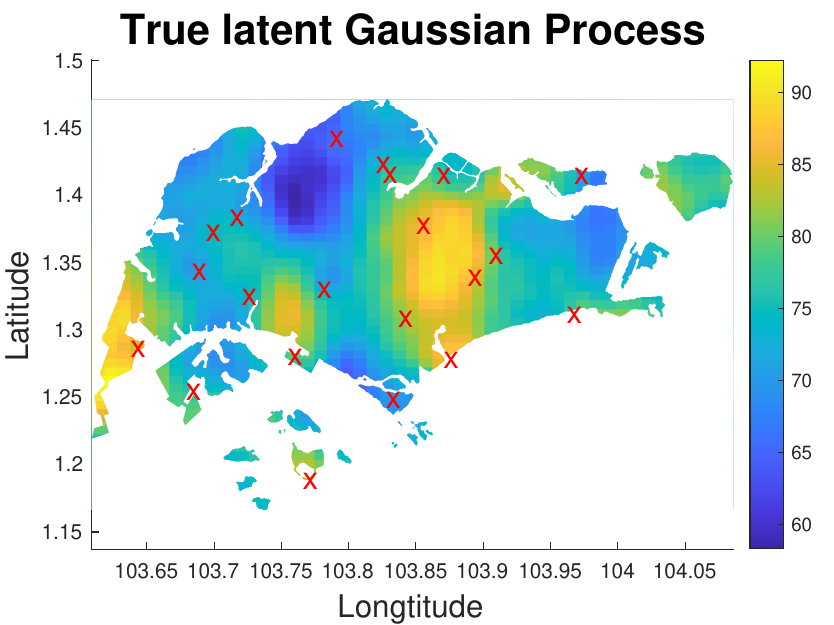}
         \caption{}
         \label{fig:real true latent field}
     \end{subfigure}
     \hfill
     \begin{subfigure}[b]{0.328\textwidth}
         \centering
         \includegraphics[width=1\textwidth]{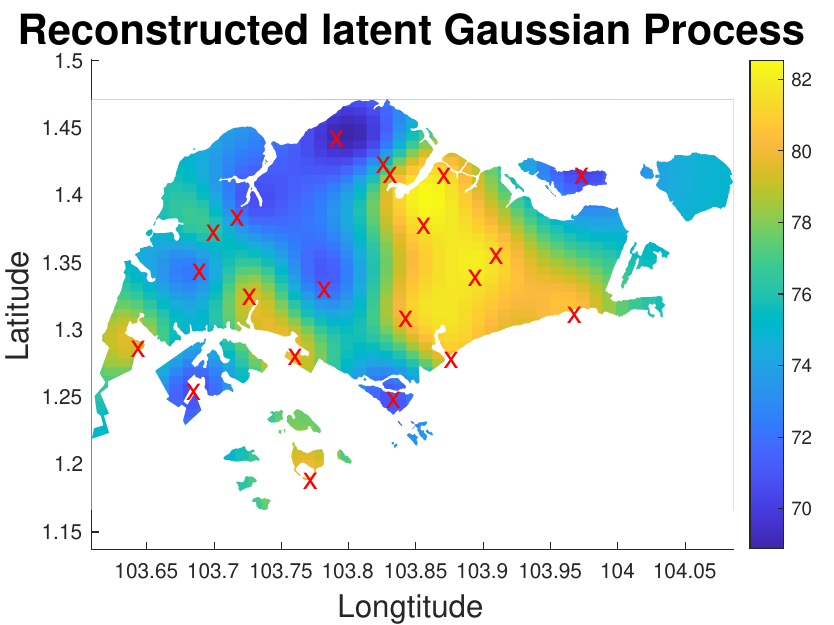}
         \caption{}
         \label{fig:real reconstructed latent field}
     \end{subfigure}
     \hfill
     \begin{subfigure}[b]{0.328\textwidth}
         \centering
         \includegraphics[width=1\textwidth]{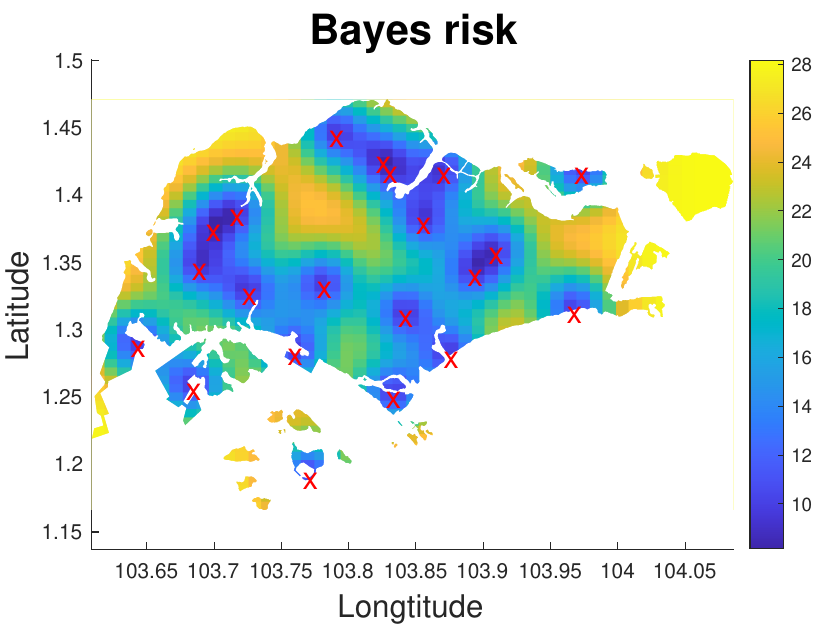}
         \caption{}
         \label{fig:real bayes risk}
     \end{subfigure}
        \caption{Real-world experiment -- Figure~(\ref{fig:real true latent field}) shows the true latent GP. Figure~(\ref{fig:real reconstructed latent field}) visualizes the reconstructed latent GP from our algorithm based on the point observations at the $21$ weather stations. Figure~(\ref{fig:real bayes risk}) shows the map of the Bayes risks defined in \eqref{SBLUE:TheoreticalRisk}.} 
        \label{fig:real Latent}
\end{figure}

\begin{table}[h!]
\centering
\begin{tabular}{|c |c | c |c|c|c|} 
 \hline
  Algorithm&MSE & F1 score  & FPR &  TPR \\ [0.5ex] 
    \hline
 Oracle &0.3403 &0.5962&0.2282&0.5326 \\
 \hline
 S-BLUE &0.3790& 0.5544  & 0.2740 & 0.5046\\
 \hline
 KNN &0.4252 & 0.5391 & 0.4430 &  0.5678 \\
 \hline 
\end{tabular}
\caption{Real-world experiment -- Average MSE, F1 score, FPR, TPR over $100$ realizations.}
\label{table:real performance}
\end{table}

\subsection{Sensitivity Analysis}
In this subsection, we analyze the effects of the significance level, number of point observations, and noise variance on the proposed algorithm. To start with, Figure~(\ref{fig:real roc}) shows the ROC curve of WGPLRT under the experiment setting specified in Subsection \ref{Subsection: Experiment Setting}. To study the impact of the significance level, i.e., FPR, the transition matrix for S-BLUE is determined by points along the ROC curve in Figure~(\ref{fig:real roc}) at each significance level. We plot the MSE against the significance level in Figure~(\ref{fig:real SBLUEvsROC}). Notice that the MSE decreases initially and then increases because of the trade-off between the TPR and FPR as suggested in Figure~(\ref{fig:real roc}). 

\begin{figure}[t]
     \centering
     \begin{subfigure}[b]{0.49\textwidth}
         \centering
         \includegraphics[width=1\textwidth]{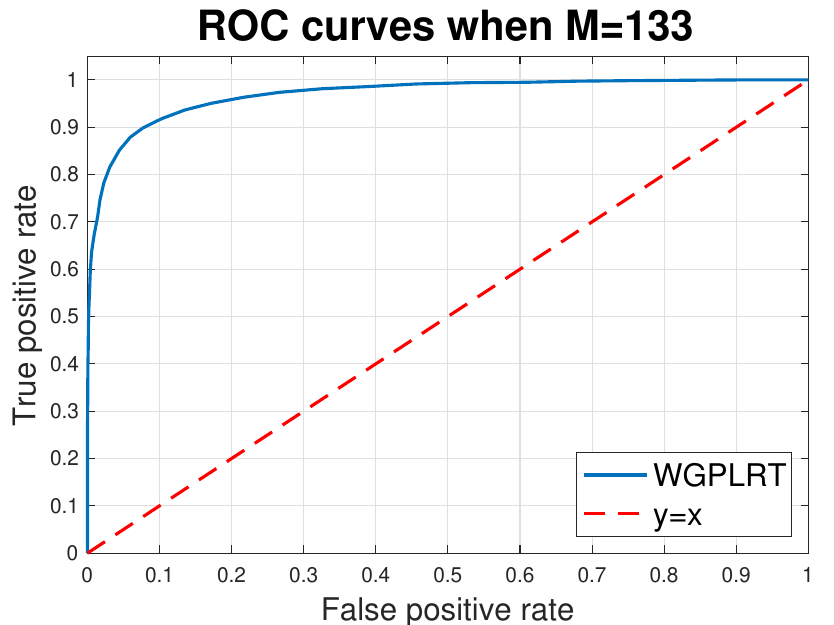}
         \caption{}
         \label{fig:real roc}
     \end{subfigure}
     \hfill
     \begin{subfigure}[b]{0.49\textwidth}
         \centering
         \includegraphics[width=1\textwidth]{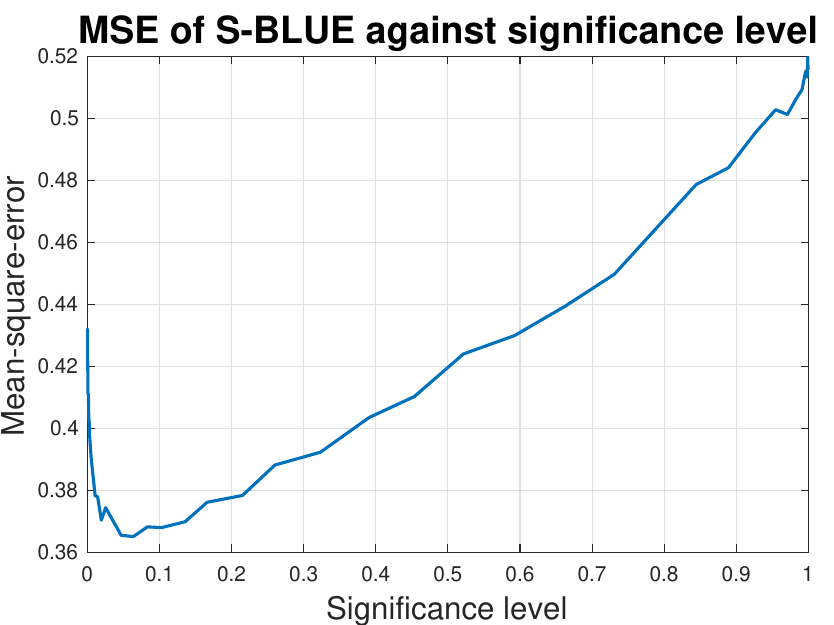}
         \caption{}
         \label{fig:real SBLUEvsROC}
     \end{subfigure}
      \caption{Real-world experiment -- Figure~(\ref{fig:real roc}) shows the ROC curve of WGPLRT under the experiment setting in Subsection \ref{Subsection: Experiment Setting}.  Figure~(\ref{fig:real SBLUEvsROC}) shows the MSE of S-BLUE against the significance level.}
\end{figure}

We proceed to study the effects of the number of point observations and the noise variance. Figure~(\ref{fig:real AUCvsM}) shows the AUC of WGPLRT against the number of point observations. We observe that even though the AUC improves when the number of point observations increases, the slope is close to zero when the number of point observations exceeds approximately $19$. Therefore, we take $M = 19$ when analyzing the effects of the noise variance. Figure~(\ref{fig:real scoresVsSn}) presents the average MSE, F1 score, FPR, TPR over $100$ realizations against the noise variance $\sigma^2_{\TP}$. As in the synthetic experiment, the FPR stays relatively constant, the F1 score, TPR increase, and the MSE decreases when the noise variance decreases.

\begin{figure}[!h]  % spans both columns
\begin{subfigure}{0.49\textwidth}
\includegraphics[width=\linewidth]{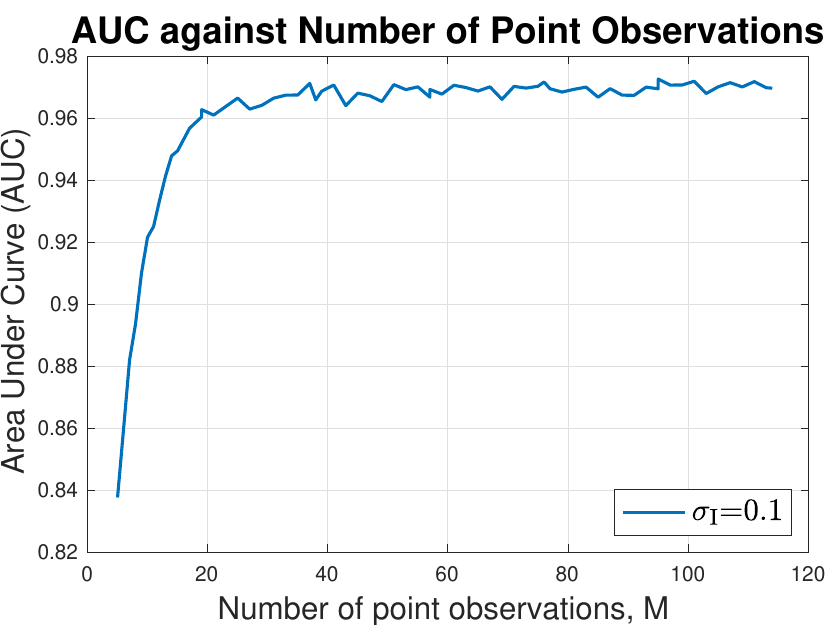}
\caption{}
\label{fig:real AUCvsM}
\end{subfigure}
\hfill
\begin{subfigure}{0.49\textwidth}
\includegraphics[width=\linewidth]{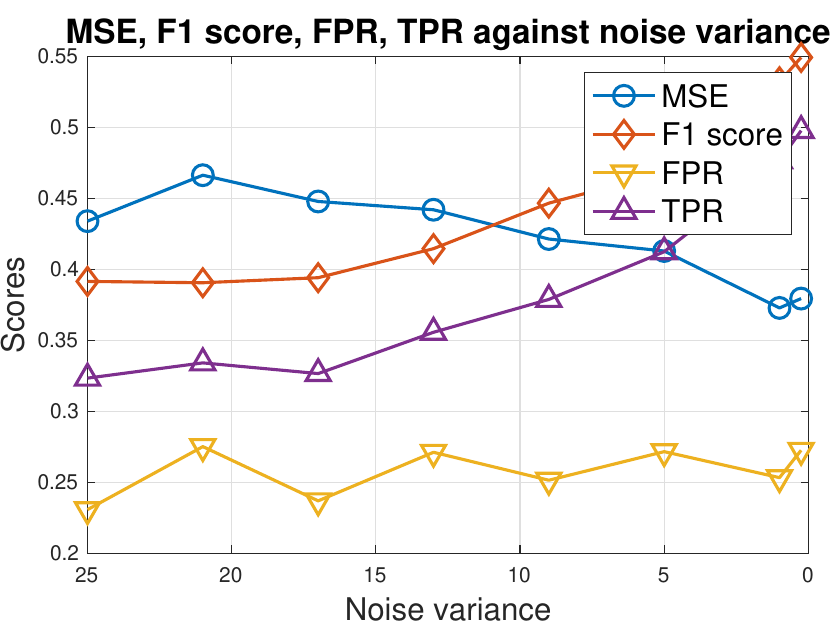}
\caption{}
\label{fig:real scoresVsSn}
\end{subfigure}
\caption{Real-world experiment -- Figure~(\ref{fig:real AUCvsM}) shows the effect of the number of point observations on the AUC of WGPLRT. Figure~(\ref{fig:real scoresVsSn}) presents the average MSE, F1 score, FPR, and TPR over $100$ realizations against the noise variance of Algorithm~\ref{Overall:algo} when $M=19$.} % Overall figure \
\label{fig:Real}
\end{figure}

\section{Conclusion}\label{Sec:Conclusion}
This paper addressed the problem of \textit{binary spatial random field reconstruction} of a hierarchical spatial-temporal system based on sensor observations of the temporal processes. A novel model was proposed to represent a hierarchical spatial-temporal physical phenomenon using WGPs such that the processes may follow arbitrary distributions appeared in real-world applications. A sensor network was deployed over a vast geographical region to monitor the hierarchical spatial-temporal physical phenomenon, where two types of sensors were considered; one collects \textit{point observations} at specific time points while the other collects \textit{integral observations} over time intervals.} We developed two algorithms: the \textit{Warped Gaussian Process Likelihood Ratio Test (WGPLRT)} and the \textit{Neighborhood-density-based Likelihood Ratio Test (NLRT)} to compress the local sensor observations of the temporal processes to a single-bit. Next, based on the local inferences, we developed the \textit{Spatial Best Linear Unbiased Estimator (S-BLUE)} to solve the problem of binary spatial random field reconstruction. We preformed both synthetic experiments and real-world experiments, the latter of which are based on a weather dataset from the National Environment Agency (NEA) of Singapore. The results showed that our proposed algorithms can effectively reconstruct the binary spatial random field.

\section{Acknowledgments}
The research was conducted under the Undergraduate Research Experience on Campus (URECA) project, supported by the School of Physics and Mathematical Sciences at Nanyang Technological University. AN gratefully acknowledges the financial support by his Nanyang Assistant Professorship Grant (NAP Grant) \textit{Machine Learning based Algorithms in Finance and Insurance}.

\appendix
\section{Proof of Proposition~\ref{L1}}
\label{L1_P}
\begin{proof}[Proof of Proposition~\ref{L1}]
Let $g_*:=g(\Bx_*)$ denote the value of the latent Gaussian Process at a fixed un-monitored location $\Bx_* \in \CX$. Note that $y_*:=y(\Bx_*)\in\{0,1\}$, and hence it is $\mathrm{Bernoulli}(\pi_*)$ distributed for some $\pi_* \in [0,1]$. The Bernoulli parameter $\pi_*$ of the posterior predictive distribution of $y_*$ conditional on $(\BZ^{\TP}_{1:N^{\TP}},\BZ^{\TI}_{1:N^{\TI}})$ is given by 
\begin{align}\label{Lemma2:1}
\begin{split}
\pi_*=\PROB( y_*=1|\BZ^{\TP}_{1:N^{\TP}},\BZ^{\TI}_{1:N^{\TI}})=\int_{\R} \PROB(y_*=1|g_*,\BZ^{\TP}_{1:N^{\TP}},\BZ^{\TI}_{1:N^{\TI}})p(g_*|\BZ^{\TP}_{1:N^{\TP}},\BZ^{\TI}_{1:N^{\TI}})\DIFFX{g_*}.
\end{split}
\end{align}
Since $y_*$ is independent of $(\BZ^{\TP}_{1:N^{\TP}}, \BZ^{\TI}_{1:N^{\TI}})$ conditional on $g_*$ (see Figure~\ref{fig:plate}), we obtain that 
\begin{equation}
\PROB(y_*=1|g_*,\BZ^{\TP}_{1:N^{\TP}},\BZ^{\TI}_{1:N^{\TI}}) = \PROB(y_*=1|g_*)=\INDI_{\{g_*\geq c\}},
\end{equation}
which is a deterministic function of $g_*$. 

By the Law of Total Probability, the term $p(g_*|\BZ^{\TP}_{1:N^{\TP}},\BZ^{\TI}_{1:N^{\TI}})$ can be written as
\begin{equation}\label{Lemma2:2}
\begin{split}
       p(g_*|\BZ^{\TP}_{1:N^{\TP}},\BZ^{\TI}_{1:N^{\TI}})&=\int_{\R^N} p(g_*|\BIg,\BZ^{\TP}_{1:N^{\TP}},\BZ^{\TI}_{1:N^{\TI}})p(\BIg|\BZ^{\TP}_{1:N^{\TP}},\BZ^{\TI}_{1:N^{\TI}})\DIFFX{\BIg}.
\end{split}
\end{equation}
Then, by the independence of $g_*$ and $(\BZ^{\TP}_{1:N^{\TP}}, \BZ^{\TI}_{1:N^{\TI}})$ conditional on $\BIg$ (see Figure~\ref{fig:plate}), we derive that
\begin{equation}
        p(g_*|\BZ^{\TP}_{1:N^{\TP}},\BZ^{\TI}_{1:N^{\TI}})=\int_{\R^N} p(g_*|\BIg)p(\BIg|\BZ^{\TP}_{1:N^{\TP}},\BZ^{\TI}_{1:N^{\TI}})\DIFFX{\BIg}.
\end{equation}
Moreover, by the Bayes' rule, we have
\begin{align}
\label{Lemma2:3}
p(\BIg|\BZ^{\TP}_{1:N^{\TP}},\BZ^{\TI}_{1:N^{\TI}})=\frac{p(\BZ^{\TP}_{1:N^{\TP}},\BZ^{\TI}_{1:N^{\TI}}|\BIg)p(\BIg)}{\int_{\R^N}p(\BZ^{\TP}_{1:N^{\TP}},\BZ^{\TI}_{1:N^{\TI}}|\BIg')p(\BIg')\DIFFX{\BIg'}}.
\end{align}
Due to the conditional independence of $\big\{\BZ^{\TP}_n\big\}_{n=1:N^{\TP}}$, $\big\{\BZ^{\TI}_n\big\}_{n=1:N^{\TI}}$ given $\BIg=(g(\Bx^{\TP}_1),$ $\ldots,$ $g(\Bx^{\TP}_{N^{\TP}}),$ $g(\Bx^{\TI}_1),$ $\ldots,$ $g(\Bx^{\TI}_{N^{\TI}}))$ stated in \ref{Temporal:LatentProcess}, one can factorize $p(\BZ^{\TP}_{1:N^{\TP}},\BZ^{\TI}_{1:N^{\TI}}|\BIg)$ as follows
\begin{align}\label{Lemma2:4}
p(\BZ^{\TP}_{1:N^{\TP}},\BZ^{\TI}_{1:N^{\TI}}|\BIg)=
\left(\prod_{n=1}^{N^{\TP}}p(\BZ^{\TP}_{n}|\BIg)\right)\left(\prod_{n=1}^{N^{\TI}}p(\BZ^{\TI}_{n}|\BIg)\right).
\end{align}
Therefore, (\ref{Lemma2:1})--(\ref{Lemma2:4}) imply that $\PROB( y_*=1|\BZ^{\TP}_{1:N^{\TP}},\BZ^{\TI}_{1:N^{\TI}})$ is given by
\begin{equation}\label{Lemma2:5}
    \int_{c}^{\infty}
\int_{\R^N} p(g_*|\BIg)\frac{
\left(\prod_{n=1}^{N^{\TP}}p(\BZ^{\TP}_{n}|\BIg)\right)\left(\prod_{n=1}^{N^{\TI}}p(\BZ^{\TI}_{n}|\BIg)\right)
p(\BIg)}{\int_{\R^N}
\left(\prod_{n=1}^{N^{\TP}}p(\BZ^{\TP}_{n}|\BIg')\right)\left(\prod_{n=1}^{N^{\TI}}p(\BZ^{\TI}_{n}|\BIg')\right)
p(\BIg')\DIFFX{\BIg'}}\DIFFX{\BIg}
\DIFFX{g_*}.
\end{equation}
This proves \eqref{Spatial:Y_posterior}.

Now, let $y_n^j:=y(\Bx_n^j)$, $g_n^j:=g(\Bx_n^j)$ for $j\in \{\TP,\TI\}$ and observe that 
\begin{equation}\label{Lemma5:6}
\begin{split}
&\;\;\;\;\int_{\R^N}\textstyle\Big(\prod_{n=1}^{N^{\TP}}p(\BZ^{\TP}_{n}|\BIg')\Big)\Big(\prod_{n=1}^{N^{\TI}}p(\BZ^{\TI}_{n}|\BIg')\Big)p(\BIg')\DIFFX{\BIg'}\\
&=\int_{\R^N}\textstyle\Big(\prod_{n=1}^{N^{\TP}}\big(\sum_{l=0,1}p(\BZ^{\TP}_{n}|y_n^{\TP}=l,\BIg')p(y_n^{\TP}=l|\BIg')\big)\Big)\\
&\qquad\qquad\textstyle\times\Big(\prod_{n=1}^{N^{\TI}}\big(\sum_{l=0,1}p(\BZ^{\TI}_{n}|y_n^{\TI}=l,\BIg')p(y_n^{\TI}=l|\BIg')\big)\Big)p(\BIg')\DIFFX{\BIg'}\\
&=\int_{\R^N}\textstyle\Big(\prod_{n=1}^{N^{\TP}}\big(p(\BZ^{\TP}_{n}|y_n^{\TP}=0)\INDI_{\{g_n^{\TP}<c\}}+p(\BZ^{\TP}_{n}|y_n^{\TP}=1)\INDI_{\{g_n^{\TP}\geq c\}}\big)\Big)\\
&\qquad\qquad\times\textstyle\Big(\prod_{n=1}^{N^{\TI}}\big(p(\BZ^{\TI}_{n}|y_n^{\TI}=0)\INDI_{\{g_n^{\TI}< c\}}+p(\BZ^{\TI}_{n}|y_n^{\TI}=1)\INDI_{\{g_n^{\TI}\geq c\}}\big)\Big)p(\BIg')\DIFFX{\BIg'}.
\end{split}
\end{equation}
Hence, we see that the denominator of the integrand in equation (\ref{Lemma2:5}) contains a sum of $2^{N^{\TI}+N^{\TP}}$ terms and is therefore computationally intractable.
\end{proof}

\section{Proof of Proposition~\ref{L2}}
\label{L2_P}
The following Lemma is needed in the proof of Proposition~\ref{L2}.
\begin{lemma}
\label{GaussianCopula}
Let $N\in\N$, let $(\Bx_n)_{n=1:N}\subset\CX$, and let $\Bz:=(z(\Bx_1),z(\Bx_2),\ldots,z(\Bx_N))^\TRANSP$ where $z=W\circ f $ is a WGP defined as in Definition~\ref{Definition:WGP} with $W: \R\to \text{range}(W)\subseteq \R$ being strictly increasing and continuously differentiable. Let $W^{-1}$ denote the inverse of $W$ and let $\Sigma_{\Bx_{1:N}}$ denote the covariance matrix of $f$ evaluated at $(\Bx_n)_{n=1:N}$. 
Then, the probability density function of $\Bz$ is given by
\begin{align}\label{Lemma: gaussianCopula general}
\begin{split}
p_\Bz(\Bz)&=\phi_{\CN}\Big(\left(W^{-1}(z(\Bx_1)),\cdots,W^{-1}(z(\Bx_N))\right);\veczero,\Sigma_{\Bx_{1:N}}\Big)\prod \limits_{i=1}^N \frac{\partial W^{-1}(z(\Bx_i))}{\partial z(\Bx_i)},
\end{split}
\end{align}
where $\phi_{\CN}(\cdot; \veczero, \Sigma_{\Bx_{1:N}}):\R^N \to \R$ denotes the density function of a multivariate normal distribution with mean vector $\veczero$ and covariance matrix $\Sigma_{\Bx_{1:N}}$. 
\\
Moreover, if $W :=F^{-1}\circ \Phi$, where $F$ is strictly increasing and is the CDF of a continuous random variable with continuous density, and $\Phi$ is the CDF of $\CN(0,1)$, then we have
\begin{equation}\label{Lemma: gaussianCopula distribution}
p_\Bz(\Bz) = \phi_{\CN}\Big(\left(W^{-1}(z(\Bx_1)),\cdots,W^{-1}(z(\Bx_N))\right);\veczero,\Sigma_{\Bx_{1:N}}\Big)\prod \limits_{i=1}^N
\frac{F'\left(z(\Bx_i)\right)}{\Phi'\big(\Phi^{-1}(F(z(\Bx_i)))\big)}.
\end{equation}
\end{lemma}

\begin{proof}%explain more
As $\Bz=\left(W(f(\Bx_1)),\cdots,W(f(\Bx_N))\right)$ and $\left(f(\Bx_1),\cdots,f(\Bx_N)\right)$ follows a multivariate normal distribution with mean vector $\veczero$ and covariance matrix $\Sigma_{\Bx_{1:N}}$, the Change of Variable Theorem (see\textit{,} e.g., \cite[Theorem~10.9]{Rudin1976}) yields that the joint density of $\Bz$ is given by
\begin{equation*}
\begin{split}
    p_\Bz(\Bz)&=\phi_{\CN}\Big(\left(W^{-1}(z(\Bx_1)),\cdots,W^{-1}(z(\Bx_N))\right);\veczero,\Sigma_{\Bx_{1:N}}\Big)\left|\frac{\partial W^{-1} (\Bz)}{\partial \Bz}\right|\\
    &=\phi_{\CN}\Big(\left(W^{-1}(z(\Bx_1)),\cdots,W^{-1}(z(\Bx_N))\right);\veczero,\Sigma_{\Bx_{1:N}}\Big)\det
\begin{pmatrix}
\frac{\partial W^{-1}(z(\Bx_1))}{\partial z(\Bx_1)}&\cdots&0\\
\vdots&\ddots&\vdots\\
0&\cdots&\frac{\partial W^{-1}(z(\Bx_N))}{\partial z(\Bx_N)}
\end{pmatrix}.
\end{split}
\end{equation*}
Therefore, we obtain that
\begin{equation}
p_\Bz(\Bz)=\phi_{\CN}\Big(\left(W^{-1}(z(\Bx_1)),\cdots,W^{-1}(z(\Bx_N))\right);\veczero,\Sigma_{\Bx_{1:N}}\Big)\prod \limits_{i=1}^N
\frac{\partial W^{-1}(z(\Bx_i))}{\partial z(\Bx_i)},
\end{equation}
which proves \eqref{Lemma: gaussianCopula general}.

Moreover, if $W=F^{-1}\circ \Phi$, where $F$ is strictly increasing and is the CDF of a continuous random variable with continuous density, and $\Phi$ is the CDF of $\CN(0,1)$, then $W$ is also strictly increasing and continuously differentiable. Therefore, since
\begin{equation}
    \frac{\partial W^{-1}(y)}{\partial y}\bigg|_{y=z(\Bx_i)}=\frac{\partial \Phi^{-1}\big(F(y)\big)}{\partial y}\bigg|_{y=z(\Bx_i)}=
    \frac{F'\left(z(\Bx_i)\right)}{\Phi'\big(\Phi^{-1}(F(z(\Bx_i)))\big)},
\end{equation}
we can conclude the desired result.
\end{proof}

\begin{proof}[Proof of Proposition~\ref{L2}]
For $i=0,1$, as $W_i = F_i^{-1} \circ \Phi$ where $F_i$ is strictly increasing and is the CDF of a continuous random variable and by Lemma~\ref{GaussianCopula}, the probability density of the ground-truth values $\tilde \BZ^{\TP}_n$ of the temporal process $\tilde z$ at $\Bx_n^{\TP}$ over $T_{1:M}^{\TP}$ with $G_i:=W_i^{-1}$ is given by
\begin{equation}
\begin{split}
  p(\tilde \BZ^{\TP}_n|\CH_i)&=(2\pi)^{-\frac{M}{2}}\det(K_i)^{-\frac{1}{2}}\exp\bigg(\!\!-\frac{1}{2}G_i(\tilde \BZ^{\TP}_n)^\TRANSP K_i^{-1}G_i(\tilde \BZ^{\TP}_n)+\sum_{m=1}^M \log \frac{\partial G_i(z)}{\partial z}\bigg|_{\tilde z^{\TP}_{n,m}}\bigg).
\end{split}
\end{equation}
Since $\BZ^{\TP}_n = \tilde \BZ^{\TP}_{n} + \Bepsilon_n^{\TP}$ and by $\Bepsilon_n^{\TP}\sim \CN(\veczero, \sigma^2_{\TP}I_{M})$, we have $\BZ^{\TP}_n | \tilde \BZ^{\TP}_{n} \sim \CN(\tilde \BZ^{\TP}_{n}, \sigma^2_{\TP}I_{M})$. Therefore, by the Law of Total Probability, the marginal likelihood of $ \BZ^{\TP}_n$ is given by
\begin{equation}
\begin{split}
  p(\BZ^{\TP}_n|\CH_i)&= \int_{\R^M}p(\BZ^{\TP}_n|\tilde \BZ^{\TP}_n; \CH_i)p(\tilde \BZ^{\TP}_n|\CH_i)\DIFFX{\tilde\BZ^{\TP}_n}
  \\&=\int_{\R^M}\exp\bigg(\!\!-\frac{1}{2}G_i(\tilde \BZ^{\TP}_n)^\TRANSP K_i^{-1}G_i(\tilde \BZ^{\TP}_n)-\frac{1}{2}\log \det K_i-\frac{M}{2}\log{2\pi}\\
   &\phantom{=}+\sum_{m=1}^M \log \frac{\partial G_i(z)}{\partial z}\bigg|_{\tilde z^{\TP}_{n,m}}\bigg)(2\pi\sigma_{\TP}^2)^{-\frac{M}{2}}\exp{\left(-\frac{1}{2}\sigma_{\TP}^{-2}(\BZ^{\TP}_n -\tilde \BZ^{\TP}_n)^{\TRANSP}(\BZ^{\TP}_n -\tilde \BZ^{\TP}_n)\right)}\DIFFX{\tilde\BZ^{\TP}_n}.
\end{split}
\end{equation}
Finally, by applying the Change of Variable Theorem with $\Bepsilon_n^{\TP} =\BZ^{\TP}_{n} -\tilde \BZ^{\TP}_n$, we conclude that
\begin{equation}
\begin{split}
  p(\BZ^{\TP}_n|\CH_i)&=\int_{\R^M}\exp\bigg(\!\!-\frac{1}{2}G_i(\BZ^{\TP}_n-\Bepsilon_n^{\TP})^\TRANSP K_i^{-1}G_i(\BZ^{\TP}_n-\Bepsilon_n^{\TP})-\frac{1}{2}\log \det K_i-\frac{M}{2}\log{2\pi}
 \\&\qquad\qquad+\sum_{m=1}^M \log \frac{\partial G_i(z)}{\partial z}\bigg|_{z^{\TP}_{n,m}-\epsilon^{\TP}_{n,m}}\bigg )(2\pi\sigma_{\TP}^2)^{-\frac{M}{2}}\exp{\left(-\frac{1}{2}\sigma_{\TP}^{-2}(\Bepsilon_n^{\TP})^{\TRANSP}\Bepsilon_n^{\TP}\right)}\DIFFX \Bepsilon_n^{\TP}.
\end{split}
\end{equation}
\end{proof}

\section{Proof of Proposition~\ref{Temporal: Laplace Approximation}}
\label{L3_P}
\begin{proof}
In the following proof, for the ease of notation, we drop the subscript $i$, denote $\BZ_n^{\TP}$ as $\Bz$, denote $\Bepsilon_n^{\TP}$ as $\Bepsilon$, denote $\widehat p(\BZ_n^{\TP}|\CH_i)$ as $\widehat p(\Bz)$, and let $K_M:=\CC_i(T_{1:M}^{\TP},T_{1:M}^{\TP})$ denote the covariance matrix evaluated at $T^{\TP}_{1:M}$.
We have
\begin{align*}
      \widehat p(\Bz)&:= \int_{\R^M}\exp\Big(\widehat Q(\Bz-\Bepsilon)-\frac{1}{2}\log \det K_M-\frac{M}{2}\log{2\pi}\Big)(2\pi\sigma^2_{\TP})^{-\frac{M}{2}}\exp\Big(-\frac{1}{2}\sigma^{-2}_{\TP}\Bepsilon^\TRANSP\Bepsilon\Big)\DIFFX \Bepsilon\\
      &=C\int_{\R^M} \exp\left(-\frac{1}{2}(\Bz-\Bepsilon-\widehat \Bv)^\TRANSP A (\Bz-\Bepsilon-\widehat \Bv) -\frac{1}{2}\sigma^{-2}_{\TP}\Bepsilon^\TRANSP\Bepsilon\right)\DIFFX \Bepsilon, 
\end{align*}
where the coefficient $C$ is given by
\begin{align*}
    C&=\exp\left(-\frac{1}{2}\log \det K_M-\frac{M}{2}\log 2\pi\right)\exp\left(Q(\widehat \Bv)\right)\exp\left(-\frac{M}{2}\log 2\pi-\frac{M}{2}\log \sigma^2_{\TP}\right)\\
    &=\exp\left(-\frac{1}{2}\log \det K_M -M\log 2\pi-M\log \sigma_{\TP}+ Q(\widehat \Bv)\right).
\end{align*}
Therefore, by completing the square and by $A^\TRANSP=A$, we obtain that
\begin{equation*}
\begin{split}
     \widehat p(\Bz)&=C\int_{\R^M} \exp\Bigg(-\frac{1}{2}(\Bz-\Bepsilon-\widehat \Bv)^\TRANSP A (\Bz-\Bepsilon-\widehat \Bv)-\frac{1}{2}\sigma^{-2}_{\TP}\Bepsilon^T\Bepsilon)\Bigg)\DIFFX \Bepsilon\\
      &=C \int_{\R^M} \exp\Big(-\frac{1}{2}\big[(\Bz-\widehat \Bv)^\TRANSP A(\Bz-\widehat \Bv)-2\Bepsilon^\TRANSP A(\Bz-\widehat \Bv)+\Bepsilon^\TRANSP(A+\sigma^{-2}_{\TP} I)\Bepsilon \big]\Big)\DIFFX \Bepsilon\\ 
      &=C\exp\left(-\frac{1}{2}(\Bz-\widehat \Bv)^\TRANSP A(\Bz-\widehat \Bv)\right)\\
      &\phantom{=}\times\int_{\R^M} \exp\Big(-\frac{1}{2}[\Bepsilon-(A+\sigma^{-2}_{\TP}I)^{-1}A(\Bz-\widehat \Bv)]^\TRANSP(A+\sigma^{-2}_{\TP}I)[\Bepsilon-(A+\sigma^{-2}_{\TP}I)^{-1}A(\Bz-\widehat \Bv)]\\
      &\qquad\qquad\qquad+\frac{1}{2}(\Bz-\widehat \Bv)^\TRANSP A(A+\sigma^{-2}_{\TP}I)^{-1}A(\Bz-\widehat \Bv)\Big)\DIFFX\Bepsilon\\
      &=C\exp\left(-\frac{1}{2}(\Bz-\widehat \Bv)^\TRANSP A(\Bz-\widehat \Bv)+\frac{1}{2}(\Bz-\widehat \Bv)^\TRANSP A(A+\sigma^{-2}_{\TP}I)^{-1}A(\Bz-\widehat \Bv)\right)\\
      &\phantom{=}\times\int_{\R^M} \exp\Big(-\frac{1}{2}[\Bepsilon-(A+\sigma^{-2}_{\TP}I)^{-1}A(\Bz-\widehat \Bv)]^\TRANSP(A+\sigma^{-2}_{\TP}I)[\Bepsilon-(A+\sigma^{-2}_{\TP}I)^{-1}A(\Bz-\widehat \Bv)]
     \Big)\DIFFX\Bepsilon\\
     \end{split}
\end{equation*}
\begin{equation}
\begin{split}\label{Lemma4:1}
     \phantom{\widehat p(\Bz} &=C\exp\left(\frac{1}{2}(\Bz-\widehat \Bv)^\TRANSP\left(A(A+\sigma^{-2}_{\TP}I)^{-1}A-A\right)(\Bz-\widehat \Bv)\right)(2\pi)^{\frac{M}{2}}\big(\det (A+\sigma^{-2}_{\TP}I)\big)^{-\frac{1}{2}}\\
      &\phantom{=}\times\int_{\R^M}\exp\Big (-\frac{1}{2}[\Bepsilon-(A+\sigma^{-2}_{\TP}I)^{-1}A(\Bz-\widehat \Bv)]^\TRANSP(A+\sigma^{-2}_{\TP}I)[\Bepsilon-(A+\sigma^{-2}_{\TP}I)^{-1}A(\Bz-\widehat \Bv)]\Big) \\
      &\qquad\qquad\times(2\pi)^{-\frac{M}{2}}\big(\det (A+\sigma^{-2}_{\TP}I)\big)^{\frac{1}{2}}\DIFFX\Bepsilon.
\end{split}
\end{equation}
Moreover, note that the integrand in the last equality is the density function of $\CN\left(\widehat \Bmu,\widehat \BSigma\right)$ with $\widehat \Bmu:=(A+\sigma^{-2}_{\TP}I)^{-1}A(\Bz-\widehat \Bv)$ and $ \widehat \BSigma:=(A+\sigma^{-2}_{\TP}I)^{-1}$, and hence
\begin{align*}
   \int_{\R^M}  (2\pi)^{-\frac{M}{2}}\big(\det (A+\sigma^{-2}_{\TP}I)\big)^{\frac{1}{2}}\exp\Big (-\frac{1}{2}[\Bepsilon-(A+\sigma^{-2}_{\TP}I)^{-1}A(\Bz-\widehat \Bv)]^\TRANSP(A+\sigma^{-2}_{\TP}I)& \\
      \quad\quad\quad\times[\Bepsilon-(A+\sigma^{-2}_{\TP}I)^{-1}A(\Bz-\widehat \Bv)]\Big)\DIFFX\Bepsilon&=1.
\end{align*}
Combining this with (\ref{Lemma4:1}) yields that
\[
\widehat p(\Bz)=\widehat C \exp\Bigg(\frac{1}{2}(\Bz-\widehat \Bv)^\TRANSP(A(A+\sigma^{-2}_{\TP}I)^{-1}A-A)(\Bz-\widehat \Bv)\Bigg),
\]
where the coefficient $\widehat C$ is defined as follows
\begin{align*}
      \widehat C&:=\exp\Big(-\frac{1}{2}\log \det K_M -M\log 2\pi-M\log \sigma_{\TP}+ Q(\widehat \Bv)\Big )\exp\Big(\frac{M}{2}\log 2\pi-\frac{1}{2}\log\det (A+\sigma^{-2}_{\TP}I)\Big)\\
      &=\exp\Bigg(-\frac{1}{2}\log \det K_M -\frac{M}{2}\log 2\pi-M\log \sigma_{\TP}+ Q(\widehat \Bv)-\frac{1}{2}\log\det (A+\sigma^{-2}_{\TP}I)\Bigg ).
\end{align*}
Finally, using Woodbury's matrix identity (with $Z\leftarrow A^{-1}$, $U\leftarrow I$, $V\leftarrow I$, and $W \leftarrow \sigma_{\TP}^2I$ in the notation of Appendix~A.3 in \cite{C.E.Rasmussen2006}), we therefore see that
\begin{equation*}
    \widehat p(\Bz)=\widehat C \exp\Bigg(-\frac{1}{2}(\Bz-\widehat \Bv)^\TRANSP(A^{-1}+\sigma_{\TP}^{2}I)^{-1}(\Bz-\widehat \Bv)\Bigg).
\end{equation*}
\end{proof}

\section{Proof of Theorem~\ref{SBLUE:Noisy_diff_U}}
\label{NoisySBLUE_proof}

In order to derive the analytic expression for the S-BLUE, the following Lemma is needed. Let $g_*:=g(\Bx_*)$ denote the value of the latent GP at a fixed un-monitored location $\Bx_* \in \CX$. 
\begin{lemma}\label{cond_pred}
The conditional predictive distribution of $g_*$ given $\BIg$ is $\CN(\mu_{g_*|\BIg},\sigma^2_{g_*|\BIg})$ with
\[
\mu_{g_*|\BIg}:=\mu(\Bx_*)+\CC(\Bx_*,X_{1:N})\CC(X_{1:N},X_{1:N})^{-1}(\BIg-\mu(X_{1:N})),
\]
\[
\sigma^2_{g_*|\BIg}:=\CC(\Bx_*,\Bx_*)-\CC(\Bx_*,X_{1:N})\CC(X_{1:N},X_{1:N})^{-1}\CC(X_{1:N},\Bx_*).
\]
\end{lemma}
\begin{proof}
Since $g$ is a GP, the joint distribution of $\BIg$ and $g_*$ follows a multivariate normal distribution given by
\begin{equation}
\begin{bmatrix}
\BIg\\g_*
\end{bmatrix}   \sim
\CN\Bigg(\begin{bmatrix}
\mu(X_{1:N})\\ \mu(\Bx_*)
\end{bmatrix}
,\begin{bmatrix}
\CC(X_{1:N},X_{1:N})& \CC(X_{1:N},\Bx_*)\\
\CC(\Bx_*,X_{1:N})& \CC(\Bx_*,\Bx_*)
\end{bmatrix}\Bigg).
\end{equation}
Hence, the result follows from the property of the multivariate normal distribution, see\textit{,} e.g., Appendix~A.2 in \cite{C.E.Rasmussen2006}.
\end{proof}

\begin{proof}[Proof of Theorem~\ref{SBLUE:Noisy_diff_U}]
Under the quadratic loss function, the Bayes risk $\CR[h]$ for any $h\in\CH$ is given by
\begin{equation}
\begin{split}
     \CR[h]&=\EXP \left[(\Bw^{\TRANSP}\widehat \BY_{1:N}+b-g_*)^2\right]\\
     &=\Bw^\TRANSP\EXP\left[\widehat \BY_{1:N}\widehat \BY_{1:N}^{\TRANSP}\right]\Bw+b^2-2b\EXP[g_*]+\EXP[g_*^2]+2\Bw^\TRANSP\EXP[\widehat \BY_{1:N}]b-2\Bw^\TRANSP\EXP[\widehat \BY_{1:N}g_*].
\end{split}
\end{equation}
Let $\mu_*:=\EXP[g_*]=\mu(\Bx_*)$. Differentiating $\CR[h]$ with respect to $\Bw$ and $b$ yields that
\begin{equation}
    \begin{split}
        &\frac{\partial \CR}{\partial \Bw}=2\EXP[\widehat \BY_{1:N}\widehat \BY_{1:N}^\TRANSP]\Bw+2b\EXP[\widehat \BY_{1:N}]-2\EXP[\widehat \BY_{1:N}g_*],\\
        &\frac{\partial \CR}{\partial b}=2b-2\mu_*+2\Bw^\TRANSP\EXP[\widehat \BY_{1:N}].
    \end{split}
\end{equation}
Setting the partial derivatives to zero, we obtain that
\begin{equation}
    \begin{split}
        b&=\mu_*-\Bw^\TRANSP\EXP[\widehat \BY_{1:N}],\\
        \Bw&=\left(\EXP[\widehat \BY_{1:N}\widehat \BY_{1:N}^\TRANSP]-\EXP[\widehat \BY_{1:N}]\EXP[\widehat \BY_{1:N}]^T\right)^{-1}\left(\EXP[\widehat \BY_{1:N}g_*]-\mu_*\EXP[\widehat \BY_{1:N}]\right)\\
        &=\COV[\widehat \BY_{1:N}]^{-1}\COV[\widehat \BY_{1:N},g_*].
    \end{split}
\end{equation}
Therefore, the S-BLUE is given by 
\begin{equation}
    \widehat h_\text{S-BLUE}(\widehat \BY_{1:N})= \mu_*+\COV[g_*, \widehat \BY_{1:N}]\COV[\widehat \BY_{1:N}]^{-1}(\widehat \BY_{1:N}-\EXP[\widehat \BY_{1:N}]).
\end{equation}
The mean and covariance terms involved in the S-BLUE are presented below. For the ease of notation, let $g_i:=g(\Bx_i)$, $y_i:=y(\Bx_i)$, $\mu_i:=\mu(\Bx_i)$, and $\sigma_i^2:=\CC(\Bx_i,\Bx_i)$ for $i=1\cdots,N$. Therefore, for $i=1,\cdots,N$,  as $g_i\sim \CN\left(\mu_i,\sigma_i^2\right)$, 
\begin{equation}\label{SBLUE: E Y}
    \begin{split}
        \EXP[\widehat y_i]&=\EXP[\widehat y_i|y_i=1]\PROB[y_i=1]+\EXP[\widehat y_i|y_i=0]\PROB[y_i=0]\\
        &=p^i_{11}\PROB[g_i\geq c]+p^i_{01}\PROB[g_i<c]\\
    &=p^i_{11}\Phi\left(-\frac{c-\mu_i}{\sigma_i}\right)+p^i_{01}\Phi\left(\frac{c-\mu_i}{\sigma_i}\right).
    \end{split}
\end{equation}
Similarly, for $i,j=1,\cdots,N$, due to the conditional independence of $\big\{\BZ^{\TP}_n\big\}_{n=1:N^{\TP}}$, $\big\{\BZ^{\TI}_n\big\}_{n=1:N^{\TI}}$ given $\left(g(\Bx^{\TP}_1),\ldots, g(\Bx^{\TP}_{N^{\TP}}), g(\Bx^{\TI}_1), \ldots, g(\Bx^{\TI}_{N^{\TI}})\right)$ stated in \ref{Temporal:LatentProcess}, we have
\begin{equation}\label{SBLUE: COV Y}
    \begin{split}
        \EXP[\widehat y_i \widehat y_j]&=\EXP\left[\EXP[\widehat y_i \widehat y_j|g_i,g_j]\right]\\
    &=\EXP\left[\EXP[\widehat y_i |g_i,g_j]\EXP[\widehat y_j |g_i,g_j]\right],
    \end{split}
\end{equation}
where $(g_i,g_j)^\TRANSP\sim \CN\left (\left(\begin{smallmatrix}
\mu_i \\ \mu_j
\end{smallmatrix}\right),
\left(\begin{smallmatrix}
\sigma_i^2&&\CC(\Bx_i,\Bx_j)\\
\CC(\Bx_i,\Bx_j)&&\sigma_j^2
\end{smallmatrix}\right)\right)$. Moreover, by Remark~\ref{SBLUE:noisy channel}, we have, for $i=1,\cdots, N$, $\EXP[\widehat y_i|g_i] = p^{i}_{11}\INDI_{\{g_i\geq c\}}+p^{i}_{01}\INDI_{\{g_i<c\}}$. Therefore, we see that $\EXP[\widehat y_i|g_i]$ is $\sigma(g_i)$-measurable,  where $\sigma(g_i)$ denotes the $\sigma$-algebra generated by $g_i$. Then \eqref{SBLUE: COV Y} becomes
\begin{equation}
    \begin{split}
        \EXP[\widehat y_i \widehat y_j]&=\EXP\left[\EXP[\widehat y_i |g_i]\EXP[\widehat y_j |g_j]\right]\\
        &=\EXP\left[\left(p^{i}_{11}\INDI_{\{g_i\geq c\}}+p^{i}_{01}\INDI_{\{g_i<c\}} \right)\left(p^{j}_{11}\INDI_{\{g_j\geq c\}}+p^{j}_{01}\INDI_{\{g_j<c\}} \right)\right]\\
     &=p^i_{01}p^j_{01}\EXP[\INDI_{\{g_i<c,g_j<c\}}]+p^i_{01}p^j_{11}\EXP[\INDI_{\{g_i<c,g_j\geq c\}}]\\
    &\phantom{=}+p^i_{11}p^j_{01}\EXP[\INDI_{\{g_i\geq c,g_j<c\}}]+p^i_{11}p^j_{11}\EXP[\INDI_{\{g_i\geq c,g_j\geq c\}}]\\
    &=p^i_{01}p^j_{01}\PROB(g_i<c,g_j<c)+p^i_{01}p^j_{11}\PROB(g_i<c,g_j\geq c)\\
    &\phantom{=}+p^i_{11}p^j_{01}\PROB(g_i\geq c,g_j<c)+p^i_{11}p^j_{11}\PROB(g_i\geq c,g_j\geq c).
    \end{split}
\end{equation}
Then, the covariance $\COV[\widehat y_i,\widehat y_j]$ is given by
\begin{equation}
\begin{split}
    \COV[\widehat y_i,\widehat y_j]&=\EXP[\widehat y_i\widehat y_j]-\EXP[\widehat y_i]\EXP[\widehat y_j]\\
    &=p^i_{01}p^j_{01}\PROB(g_i<c,g_j<c)+p^i_{01}p^j_{11}\PROB(g_i<c,g_j\geq c)\\
    &\phantom{=}+p^i_{11}p^j_{01}\PROB(g_i\geq c,g_j<c)+p^i_{11}p^j_{11}\PROB(g_i\geq c,g_j\geq c)\\
    &\phantom{=}-\left [p^i_{11}\Phi\left(-\frac{c-\mu_i}{\sigma_i}\right)+p^i_{01}\Phi\left(\frac{c-\mu_i}{\sigma_i}\right)\right] \left [p^j_{11}\Phi\left(-\frac{c-\mu_j}{\sigma_j}\right)+p^j_{01}\Phi\left(\frac{c-\mu_j}{\sigma_j}\right)\right].
\end{split}
\end{equation}
Finally, for $i=1,\cdots,N$, by the conditional independence of $g_*$ and $\widehat y_i$ given $g_i$ (see Figure~\ref{fig:plate}), we obtain that
\begingroup
\allowdisplaybreaks
\begin{align*}
     \EXP[\widehat y_ig_*]&=\EXP[\EXP[g_*\widehat y_i|g_i]]\\
    &=\EXP\left[\EXP\left[g_*|\widehat y_i=1,g_i\right]\PROB(\widehat y_i=1|g_i)\right]\\
    &=\EXP\left[\EXP[g_*|g_i]\PROB(\widehat y_i=1|g_i)\right]\\
   &=\EXP\Bigg[\EXP[g_*|g_i]\Big(\PROB(\widehat y_i=1|y_i=1, g_i)\PROB(y_i=1|g_i)+\PROB(\widehat y_i=1|y_i=0, g_i)\PROB(y_i=0|g_i)\Big)\Bigg]\\
   &=\EXP\Bigg[\EXP[g_*|g_i]\Big(p^i_{11}\INDI_{\{g_i\geq c\}}+p^i_{01}\INDI_{\{g_i<c\}}\Big)\Bigg]. \numberthis \label{SBLUE: COV*1}
\end{align*}
\endgroup
By Lemma \ref{cond_pred}, we have $\EXP[g_*|g_i]=\mu_*+\CC(\Bx_*,\Bx_i)(g_i-\mu_i)/\sigma_i^2$, and then (\ref{SBLUE: COV*1}) becomes
\begin{equation}\label{SBLUE: COV*i}
\begin{split}
    \EXP[\widehat y_ig_*]&=\mu_*\left(p^i_{11}\EXP[\INDI_{\{g_i\geq c\}}]+p^i_{01}\EXP[\INDI_{\{g_i<c\}}]\right)\\
    &\phantom{=}+\CC(\Bx_*,\Bx_i)\left(p^i_{11}\EXP\left[\frac{(g_i-\mu_i)}{\sigma_i^2}\INDI_{\{g_i\geq c\}}\right]+p^i_{01}\EXP\left[\frac{(g_i-\mu_i)}{\sigma_i^2}\INDI_{\{g_i<c\}}\right]\right)\\
    &=\mu_*\left( p^i_{11}\Phi\left(-\frac{c-\mu_i}{\sigma_i}\right)+p^i_{01}\Phi\left(\frac{c-\mu_i}{\sigma_i}\right)\right)\\
    &\phantom{=}+\CC(\Bx_*,\Bx_i)\left(p^i_{11}\int_{\{g_i\geq c\}}\frac{g_i-\mu_i}{\sqrt{2\pi}\sigma_i^3}\exp\left(-\frac{(g_i-\mu_i)^2}{2\sigma_i^2}\right)\DIFFX g_i\right .\\
    &\phantom{=}\left.+p^i_{01}\int_{\{g_i<c\}}\frac{g_i-\mu_i}{\sqrt{2\pi}\sigma_i^3}\exp\left(-\frac{(g_i-\mu_i)^2}{2\sigma_i^2}\right)\DIFFX g_i\right).
\end{split}
\end{equation}
Let $u = \frac{g_i-\mu_i}{\sigma_i}$, we obtain that
\begin{equation}\label{SBLUE: Cov g integral}
\int_{\{g_i\geq c\}}\frac{g_i-\mu_i}{\sqrt{2\pi}\sigma_i^3}\exp\left(-\frac{(g_i-\mu_i)^2}{2\sigma_i^2}\right)\DIFFX g_i= \int_{\frac{c-\mu_i}{\sigma_i}}^\infty \frac{u}{\sqrt{2\pi}\sigma_i}\exp(-u^2/2)\DIFFX u,
\end{equation}
and 
\begin{equation}\label{SBLUE: Cov g integral 2}
\int_{\{g_i< c\}}\frac{g_i-\mu_i}{\sqrt{2\pi}\sigma_i^3}\exp\left(-\frac{(g_i-\mu_i)^2}{2\sigma_i^2}\right)\DIFFX g_i= \int_{-\infty}^{\frac{c-\mu_i}{\sigma_i}} \frac{u}{\sqrt{2\pi}\sigma_i}\exp(-u^2/2)\DIFFX u.
\end{equation}
Subsequently, since $\R \ni u\mapsto u\exp(-u^2/2) \in \R$ is an odd function, we have 
\begin{equation}\label{SBLUE: Cov g integral 3}
 \int_{\frac{c-\mu_i}{\sigma_i}}^\infty \frac{u}{\sqrt{2\pi}\sigma_i}\exp(-u^2/2)\DIFFX u + \int_{-\infty}^{\frac{c-\mu_i}{\sigma_i}} \frac{u}{\sqrt{2\pi}\sigma_i}\exp(-u^2/2)\DIFFX u =0.
\end{equation}
If $c\geq \mu_i$, then using the substitution $u\mapsto v:=u^2/2$ gives
\begin{equation}
\begin{split}
\int_{\frac{c-\mu_i}{\sigma_i}}^\infty \frac{u}{\sqrt{2\pi}\sigma_i}\exp(-u^2/2)\DIFFX u
= \int_{\frac{(c-\mu_i)^2}{2\sigma_i^2}}^\infty \frac{1}{\sqrt{2\pi}\sigma_i} \exp(-v)\DIFFX v
=\frac{1}{\sqrt{2\pi}\sigma_i}\exp\left(-\frac{(c-\mu_i)^2}{2\sigma_i^2}\right).
\end{split}
\end{equation}
If $c<\mu_i$, then using \eqref{SBLUE: Cov g integral 3} and the substitution $u\mapsto v:=(-u)^2/2$ yields that
\begin{equation}
\begin{split}\label{SBLUE: Cov g integral 4}
\int_{\frac{c-\mu_i}{\sigma_i}}^\infty \frac{u}{\sqrt{2\pi}\sigma_i}\exp(-u^2/2)\DIFFX u 
& = -\int^{\frac{c-\mu_i}{\sigma_i}}_{-\infty} \frac{u}{\sqrt{2\pi}\sigma_i}\exp(-u^2/2)\DIFFX u \\
&=  \int_{\frac{(c-\mu_i)^2}{2\sigma_i^2}}^\infty \frac{1}{\sqrt{2\pi}\sigma_i} \exp(-v)\DIFFX v\\
& =\frac{1}{\sqrt{2\pi}\sigma_i}\exp\left(-\frac{(c-\mu_i)^2}{2\sigma_i^2}\right).
\end{split}
\end{equation}
Therefore, $\mu_* :=\EXP[g_*]$, \eqref{SBLUE: E Y}, and \eqref{SBLUE: COV*i}--\eqref{SBLUE: Cov g integral 4} imply that
\begin{equation}
    \begin{split}
            \COV[g_*,\widehat y_i]&=\EXP[\widehat y_ig_*]-\EXP[\widehat y_i]\EXP[g_*]\\
            &=\mu_*\left( p^i_{11}\Phi\left(-\frac{c-\mu_i}{\sigma_i}\right)+p^i_{01}\Phi\left(\frac{c-\mu_i}{\sigma_i}\right)\right)\\
            &\phantom{=}+\frac{1}{\sqrt{2\pi}\sigma_i}\CC(\Bx_*,\Bx_i)\left(p^i_{11}\exp\left(-\frac{(c-\mu_i)^2}{2\sigma_i^2}\right)-p^i_{01}\exp\left(-\frac{(c-\mu_i)^2}{2\sigma_i^2}\right)\right)\\
            &\phantom{=}-\mu_*\left( p^i_{11}\Phi\left(-\frac{c-\mu_i}{\sigma_i}\right)+p^i_{01}\Phi\left(\frac{c-\mu_i}{\sigma_i}\right)\right)\\
    &=\frac{1}{\sqrt{2\pi}\sigma_i}(p^i_{11}-p^i_{01})\CC(\Bx_*,\Bx_i)\exp\left(-\frac{(c-\mu_i)^2}{2\sigma_i^2}\right).
    \end{split}
\end{equation}
\end{proof}

\section{Proof of Corollary \ref{SBLUE: bayes risk}}
\label{proof: bayes risk}
\begin{proof}
By substituting (\ref{SBLUE:Noisy_estimator}) into (\ref{SBLUE:BayesRisk}), the Bayes risk $\CR[\widehat h_\text{S-BLUE}(\widehat \BY_{1:N})]$ associated with \linebreak$\widehat h_\text{S-BLUE}(\widehat \BY_{1:N})$ is given by
\begin{equation}\label{SBLUE: bayes risk 1}
\begin{split}
&\phantom{=}\EXP[(\widehat h_{\text{S-BLUE}}(\widehat \BY_{1:N})-g_*)^2]\\
&=\EXP\left[\left(\mu_*+\COV[g_*, \widehat \BY_{1:N}]\COV[\widehat \BY_{1:N}]^{-1}(\widehat \BY_{1:N}-\EXP[\widehat \BY_{1:N}])-g_* \right)^2\right]\\
&=\EXP\left[(g_*-\mu_*)^2-2(g_*-\mu_*)\COV[g_*, \widehat \BY_{1:N}]\COV[\widehat \BY_{1:N}]^{-1}(\widehat \BY_{1:N}-\EXP[\widehat \BY_{1:N}])\right.\\
&\left.\phantom{=}+ \COV[g_*, \widehat \BY_{1:N}]\COV[\widehat \BY_{1:N}]^{-1}(\widehat \BY_{1:N}-\EXP[\widehat \BY_{1:N}])(\widehat \BY_{1:N}-\EXP[\widehat \BY_{1:N}])^\TRANSP\COV[\widehat \BY_{1:N}]^{-1} \COV[g_*, \widehat \BY_{1:N}]^\TRANSP\right].
\end{split}
\end{equation}
Since $\COV[g_*, \widehat \BY_{1:N}], \COV[\widehat \BY_{1:N}]$ are constant matrices once $\Bx_*$ and $\BX_{1:N}$ are known (see \eqref{SBLUE:cov} in Theorem~\ref{SBLUE:Noisy_diff_U}), we have
\begin{equation}
\begin{split}
&\phantom{=}\EXP[(\widehat h_{\text{S-BLUE}}(\widehat \BY_{1:N})-g_*)^2]\\
&=\COV[g_*] - 2\COV[g_*, \widehat \BY_{1:N}]\COV[\widehat \BY_{1:N}]^{-1} \EXP\left[(\widehat \BY_{1:N}-\EXP[\widehat \BY_{1:N}])(g_*-\mu_*)\right]\\
&\phantom{=}+ \COV[g_*, \widehat \BY_{1:N}]\COV[\widehat \BY_{1:N}]^{-1}\EXP\left[(\widehat \BY_{1:N}-\EXP[\widehat \BY_{1:N}])(\widehat \BY_{1:N}-\EXP[\widehat \BY_{1:N}])^\TRANSP\right]\COV[\widehat \BY_{1:N}]^{-1} \COV[g_*, \widehat \BY_{1:N}]^\TRANSP\\
&=\COV[g_*] - 2\COV[g_*, \widehat \BY_{1:N}]\COV[\widehat \BY_{1:N}]^{-1} \COV[g_*, \widehat \BY_{1:N}]^\TRANSP \\
&\phantom{=}+ \COV[g_*, \widehat \BY_{1:N}]\COV[\widehat \BY_{1:N}]^{-1}\COV[\widehat \BY_{1:N}]\COV[\widehat \BY_{1:N}]^{-1} \COV[g_*, \widehat \BY_{1:N}]^\TRANSP\\
&=\COV[g_*] - \COV[g_*, \widehat \BY_{1:N}]\COV[\widehat \BY_{1:N}]^{-1} \COV[g_*, \widehat \BY_{1:N}]^\TRANSP.
\end{split}
\end{equation}
\end{proof}

\section{Details of the Computational Cost Analyses}
\label{proof: computational complexity}
In the following, let us analyze the computational cost incurred at each P-sensor, each I-sensor, and the FC in Algorithm~\ref{Overall:algo} in detail.
Recall that $\CT_{\text{opt}}$ denotes the computational cost of the optimization in the Laplace approximation in WGPLRT (Line~\ref{Temporal:WGPLRT-offline} of Algorithm~\ref{Temporal:WGPLRT}), $\CT_{\text{samp}}$ denotes the computational cost of generating a sample of integral observations, $\CT_{\text{summ}}$ denotes the computational cost of the summary statistics of each sample of integral observations in NLRT, $\CT_{\text{dist}}$ denotes the computational cost of each pairwise distance $d(\cdot,\,\cdot)$ in NLRT, $N$ denotes the total number of sensors, $M$ denotes the number of point observations at each P-sensor, and $J$ denotes the number of generated samples in NLRT.

At each P-sensor (i.e., in Algorithm~\ref{Temporal:WGPLRT}), the \textbf{offline phase} consists of first solving two optimization problems $\widehat{\Bv}_0=\argmax_{\Bv}Q_0(\Bv)$, $\widehat{\Bv}_1=\argmax_{\Bv}Q_1(\Bv)$, which incurs computational cost $2\CT_{\text{opt}}$, and then computing the values of $(A_1^{-1}+\sigma_{\TP}^{2}I)^{-1}$, $(A_0^{-1}+\sigma_{\TP}^{2}I)^{-1}$, and $\frac{1}{2}\big(\log\det (A_0+\sigma^{-2}_{\TP}I)+\log\det K_0-2Q(\widehat \Bv_0)-\log\det (A_1+\sigma^{-2}_{\TP}I)-\log\det K_1+2Q(\widehat \Bv_1)\big)$.
Note that $(A_1^{-1}+\sigma_{\TP}^{2}I)^{-1}$, $(A_0^{-1}+\sigma_{\TP}^{2}I)^{-1}$, $\log\det (A_0+\sigma^{-2}_{\TP}I)$, and $\log\det (A_0+\sigma^{-2}_{\TP}I)$ can all be computed  easily after diagonalizing $A_0$ and $A_1$, which costs $O(M^3)$. All subsequent computations cost $O(M^2)$. 
Therefore, the total computational cost during the offline phase at each P-sensor is $O(\CT_{\text{opt}}+M^3)$.
In the \textbf{online phase}, for each time-series of point observations $\BZ^{\TP}_n$, the computation of the test statistic $-\log \widehat \Lambda (\BZ^{\TP}_n)$ from (\ref{Temporal:WGPtest}) involves evaluating two vector-matrix-vector products, and thus the total computational cost during the online phase at each P-sensor is $O(M^2)$.

At each I-sensor (i.e., in Algorithm~\ref{Temporal:NLRT}), the \textbf{offline phase} consists of first generating $2J$ samples $\big\{\widehat \BZ_j^{\TI,i}\big\}_{j=1:J,\,i=0,1}$ of integral observations, which costs $2J\CT_{\text{samp}}$, and then computing their summary statistics $\big\{S(\widehat \BZ_j^{\TI,i})\big\}_{j=1:J,\,i=0,1}$, which costs $2J\CT_{\text{summ}}$.
Hence, the total computational cost during the offline phase at each I-sensor is $O(J(\CT_{\text{samp}}+\CT_{\text{summ}}))$.
In the \textbf{online phase}, for each time-series of integral observations $\BZ_n^{\TI}$, the computation of the test statistic $\widehat\Lambda(\BZ_n^{\TI})$ requires the evaluation of the summary statistics $S(\BZ_n^{\TI})$, which costs $\CT_{\text{summ}}$, as well as the evaluation of the pairwise distances between $S(\BZ_n^{\TI})$ and $\big\{S(\widehat \BZ_j^{\TI,i})\big\}_{j=1:J,\,i=0,1}$, which costs $2J\CT_{\text{dist}}$.
As a result, for each time-series of integral observations, the total computational cost during the online phase at each I-sensor is $O(\CT_{\text{summ}}+J\CT_{\text{dist}})$.

At the FC (i.e., in Algorithm~\ref{SBLUE:algo}), the \textbf{offline phase} computes the values of $\mu_*$, $\EXP[\widehat \BY_{1:N}]$, $\COV[g_*, \widehat \BY_{1:N}]\COV[\widehat \BY_{1:N}]^{-1}$, and $\CR[\widehat h_\text{S-BLUE}(\widehat \BY_{1:N})]$.
It follows from (\ref{SBLUE:cov}) and Remark~\ref{SBLUE:individual-terms} that the evaluation of each entry in $\EXP[\widehat \BY_{1:N}]$, $\COV[g_*, \widehat \BY_{1:N}]$, and $\COV[\widehat \BY_{1:N}]$ costs $O(1)$.
Hence, the computational cost in offline phase at the FC is dominated by the inversion of the matrix $\COV[\widehat \BY_{1:N}]$, which has computational complexity $O(N^3)$. As a result, the total computational cost during the offline phase at the FC is $O(N^3)$.
The \textbf{online phase} at the FC for each set of binary decisions $\widehat{\BY}_{1:N}$ simply computes $\mu_*+\COV[g_*, \widehat \BY_{1:N}]\COV[\widehat \BY_{1:N}]^{-1}(\widehat \BY_{1:N}-\EXP[\widehat \BY_{1:N}])$. This consists of the computation of a vector-vector subtraction followed by a vector inner-product and the addition of a constant. Therefore, the online phase at the FC costs $O(N)$ for each set of binary decisions from the P-sensors and I-sensors.

Furthermore, the computational cost of the $k$-nearest neighbor (KNN) algorithm is analyzed as follows.
The \textbf{offline phase} of the KNN algorithm consists of computing the pairwise distances between $\Bx_*$ and $\{\Bx_{n}\}_{n=1:N}$, which costs $O(N)$, and finding the $k$ nearest sensors, which costs\footnote{The computational complexity $O(N+k\log N)$ can be achieved by, for example, a heap-based sorting algorithm; see\textit{,} e.g., \cite{martinez2004partial}. Note that algorithms with lower computational complexity also exist; see\textit{,} e.g., the discussion in \cite{martinez2004partial}.} $O(N+k\log N)$, and hence the total computational cost is $O(N+k\log N)$.
The \textbf{online phase} of the KNN algorithm simply performs a plurality vote among the binary decisions of the $k$ nearest sensors, which incurs computational cost $O(k)$.

\bibliographystyle{elsarticle-num-names}
\bibliography{ref}

%%%%%%%%%%%%%%%%%%%%%%%%%%%%%%%%%%%%%%%%%%%%%%%%%%%%%%%%%%%%%%%%%%%%%%%%%%%%%%%%
\end{document}